\documentclass[aps,pra,reprint,superscriptaddress]{revtex4-1}

\usepackage{amsmath}
\usepackage{amssymb}
\usepackage{amsthm}
\usepackage{braket}
\usepackage{graphicx}
\usepackage{bm}
\usepackage[subrefformat=parens]{subcaption}
\usepackage{color}

\theoremstyle{definition}
\newtheorem{definition}{Definition}[section]
\newtheorem{theorem}{Theorem}[section]

\newtheorem{lemma}[theorem]{Lemma}

\begin{document}
\title{General framework for constructing fast and near-optimal machine-learning-based decoder of the topological stabilizer codes}
\author{Amarsanaa Davaasuren}
\email[]{amarsanaa137@qc.rcast.u-tokyo.ac.jp}
\affiliation{Department of Applied Physics, Graduate School of Engineering, The University of Tokyo, 7-3-1 Hongo, Bunkyo-ku, Tokyo 113-8656, Japan}
\author{Yasunari Suzuki}
\email[]{suzuki@qi.t.u-tokyo.ac.jp}
\affiliation{Department of Applied Physics, Graduate School of Engineering, The University of Tokyo, 7-3-1 Hongo, Bunkyo-ku, Tokyo 113-8656, Japan}
\affiliation{Photon Science Center, Graduate School of Engineering, The University of Tokyo, 7-3-1 Hongo, Bunkyo-ku, Tokyo 113-8656, Japan}
\author{Keisuke Fujii}
\email[]{fujii@qi.t.u-tokyo.ac.jp}
\affiliation{JST, PRESTO, 4-1-8 Honcho, Kawaguchi, Saitama, 332-0012, Japan}
\affiliation{Department of Physics, Graduate School of Science, Kyoto University, Kitashirakawa-Oiwakecho, Sakyo, Kyoto 606-8502, Japan}
\author{Masato Koashi}
\email[]{koashi@qi.t.u-tokyo.ac.jp}
\affiliation{Department of Applied Physics, Graduate School of Engineering, The University of Tokyo, 7-3-1 Hongo, Bunkyo-ku, Tokyo 113-8656, Japan}
\affiliation{Photon Science Center, Graduate School of Engineering, The University of Tokyo, 7-3-1 Hongo, Bunkyo-ku, Tokyo 113-8656, Japan}
\date{\today}
\begin{abstract}
Quantum  error correction is an essential technique for constructing a scalable quantum computer. 
In order to implement quantum error correction with near-term quantum devices, a fast and near-optimal decoding method is demanded. 
A decoder based on machine learning is considered as one of the most viable solutions for this purpose, since its prediction is fast once training has been done, and it is applicable to any quantum error correcting code and any noise model. 
So far, various formulations of the decoding problem as the task of machine learning have been proposed. 
Here, we discuss general constructions of machine-learning-based decoders. 
We found several conditions to achieve near-optimal performance, and proposed a criterion which should be optimized when a size of training data set is limited. 
We also discuss preferable constructions of neural networks, and proposed a decoder using spatial structures of topological codes using a convolutional neural network.
We numerically show that our method can improve the performance of machine-learning-based decoders in various topological codes and noise models.
\end{abstract}
\pacs{}

\maketitle

\section{Introduction}
In order to build a scalable quantum computer, quantum error correction (QEC) \cite{kitaev1997quantum, aharonov1997fault, knill1998resilient} is a vital technique for achieving reliable computation. According to the theory of QEC, if the noise strength is smaller than a certain threshold value, we can protect logical qubits encoded in physical qubits from the noise. Supported by extensive experimental efforts, the noise level of the quantum operations on arrays of qubits is now approaching and meets the threshold value. Therefore, a demonstration of QEC in a fully fault-tolerant settings is considered to be a milestone for the near-term quantum devices \cite{kelly2015state,corcoles2015demonstration,riste2015detecting}. 
Topological codes \cite{kitaev2003fault,dennis2002topological,lidar2013quantum} are a family of quantum error correcting codes inspired by topological nature in the condensed matter physics \cite{kitaev2003fault}. Since the topological codes such as surface codes \cite{dennis2002topological, bravyi1998quantum, fowler2012surface} have both high experimental feasibility and high performance \cite{wang2003confinement,wang2011surface,fowler2012towards,stephens2014fault}, they are considered as the most promising candidate of quantum error correcting codes. 

In QEC, information on occurrence of physical errors is measured as a syndrome value. A suitable recovery operation is estimated from the syndrome so that the original state of the logical qubits is decoded with high success probability. 
Unfortunately, constructing an optimal decoder is computationally hard in general.
Thus, massive efforts have been paid for developing efficient and near-optimal decoders.
One approach is to use the most likely physical errors that are consistent with the observed syndrome value as a recovery operation. This scheme is called the {\it minimum-distance (MD) decoder}. Though this decoding method is not necessarily optimal, it shows almost optimal performance \cite{wang2011surface,fowler2012towards,stephens2014fault}. 
In the case of the surface codes, if we can assume that bit-flip (Pauli $X$) and phase-flip (Pauli $Z$) errors are uncorrelated, we can construct an efficient MD decoder using minimum-weight perfect matching. 
However, if bit-flip and phase-flip errors are correlated or if we use other codes, even MD decoding is not efficiently implementable \cite{hsieh2011np}. 
Some of these problems can be avoided by the use of geometrically local features of the topological codes. 
For example, as for color codes \cite{bombin2007homological}, we can perform decoding by projecting color code to a surface code \cite{delfosse2014decoding}. 
Another approach is to use renormalization group method \cite{duclos2010fast}, which is applicable to any topological codes including the surface and color codes. 
While these approaches have been improved, there is unavoidable trade-off between the performance and time efficiency of the decoder. 
For the first experimental realization of QEC on near-term devices, more efficient and near-optimal decoders are demanded.

In this article, we discuss a general construction of machine-learning-based decoders.
Recently, the technology of machine learning has been applied to various theoretical and experimental researches of quantum physics, such as classification of readout signals in experiments \cite{magesan2015machine}, simulation of a quantum system \cite{carleo2017solving}, classification of the phase of matter \cite{carrasquilla2017machine}, data compression of the quantum state \cite{romero2017quantum}, and decoding in QEC \cite{torlai2017neural,varsamopoulos2017decoding,baireuther2018machine,krastanov2017deep,breuckmann2018scalable}. 
In the machine-learning-based decoder, we construct a prediction model which outputs a recovery operator from a given syndrome value. 
The prediction model is trained with many correct pairs of syndrome values and correct recovery operations before prediction.
While the training task may take a long time, it is required only once before many runs of prediction, and each prediction is expected to be performed fast. 
Thus, the machine-learning-based decoder is one of the best solutions for demonstrating experimental QEC in near-term quantum devices. 

As a prediction model, artificial neural network is believed to have large representation power, and is suitable for constructing machine-learning-based decoder. 
Recently, the performances of machine-learning-based decoders with various neural networks have been numerically studied, such as restricted Boltzmann machine \cite{torlai2017neural}, multi-layer perceptron \cite{varsamopoulos2017decoding}, recurrent neural network \cite{baireuther2018machine}, and deep neural network \cite{krastanov2017deep}.
The machine-learning-based decoder using a neural network is called {\it neural decoders} \cite{torlai2017neural}. 
All these existing methods numerically showed that the performance of the neural decoder is superior to the known efficient decoders when sufficiently large amount of the training data set is supplied. 
However, the following three points have yet to be understood.
The first one is how the decoding problem should be translated to the task of machine learning in order to obtain faster learning and better prediction. So far, each of the previous studies introduces its own construction of the data set and neural network with little consideration on this point.
Second, the spatial feature of the topological codes has not been considered in the construction of the neural decoder, except a very recent study \cite{breuckmann2018scalable} that was carried out independently of this work. 
While it is expected that the performance of the neural decoder is improved by explicitly considering the spatial arrangement of the syndrome, the spatial information has not been given to the neural network explicitly. 
Finally, the applicability of the neural decoder to various topological codes is not known. The neural decoder is benchmarked only with surface codes \cite{torlai2017neural,varsamopoulos2017decoding,baireuther2018machine,krastanov2017deep,breuckmann2018scalable}. Therefore, it has not been known whether the neural decoder is applicable to other codes, such as color codes. 

We have addressed all of these points in this paper. 
First, we discuss how the decoding problem should be formulated as the task of machine learning. 
We propose a general framework for constructing a neural decoder, {\it linear prediction framework}, to elucidate the factors that determine the performance of the decoders. 
We propose a criterion called {\it normalized sensitivity} which should be optimized for constructing a near-optimal neural decoder.
Then, we propose specific construction of a training data set which minimizes the normalized sensitivity. We call these constructions as {\it uniform data construction}.
We also propose the use of construction of neural networks, which explicitly utilize spatial structure of the topological codes. 
We show that the performance of the neural decoder is improved with these techniques, and it shows better performance than that of a decoder using minimum-weight perfect matching with $10^6$ data set at distance $d=11$ in the surface code under a depolarizing noise. 
We show that the neural decoder is also applicable to the color codes. 
The performance of the neural decoder for the color codes also reaches that of the MD decoder in small distances.

\subsection*{Organization of the article}
In Sec II, we overview preliminary topics. 
We review a scheme of QEC in the case of stabilizer codes. We explain specific constructions of the topological codes, the surface and color codes. We also review the basics of the supervised machine learning with neural networks in this section.
In Sec III, we address the question of how the neural decoder should be constructed. We propose a general framework, linear prediction framework, in this section. 
We introduce a quantity called the normalized sensitivity, and argue that it serves as a criterion for better performance of decoders for topological stabilizer codes.
We also propose uniform data construction, which consists of specific instructions to optimize the normalized sensitivity for surface codes and color codes.
We numerically confirm that the performance of the neural decoder is improved with this construction in the case of the surface and color codes.
In Sec IV, we propose a network construction which explicitly utilize the spatial information of the topological codes. We confirm that this construction also improves the performance of the neural decoder.
Finally, we summarize this paper in Sec V.

\section{Preliminary} 
In this section, we review the basic concepts and introduce notations used in this paper. We first review a scheme of QEC. We also introduce well-known topological codes and decoders.
The scheme of supervised machine learning with neural network and its terminologies are also explained in this section. 

\subsection{Quantum error correction}
We consider the case where $k$ logical qubits are encoded in $n$ physical qubits. 
We assume that any noise can be represented as a probabilistic Pauli operation on the $n$ physical qubits.
We denote Pauli operators on a single qubit as $\{I,X,Y,Z\}$, and the Pauli operator $A$ on the $i$-th physical qubit as $A_i$. 
When we consider operations on the $n$ physical qubits, we ignore the global phase of the state and operator. 
Then, we can represent any physical error as $E \in \{I,X,Y,Z\}^{\otimes n}$.
A weight $w(E)$ is defined for a Pauli operator $E$ on the $n$ physical qubits as the number of the physical qubits to which the Pauli operator $E$ is non-trivially applied. 

In the framework of stabilizer codes \cite{gottesman1997stabilizer}, the code is defined by $2^{n-k}$ stabilizer operators $\mathcal{L}_I$ generated by $n-k$ Pauli operators $\mathcal{L}_I := \langle \{ S_i \} \rangle$ ($1 \leq i \leq n-k$), where $S_i \in \pm \{I,X,Y,Z\}^{\otimes n}$, $-I \notin \langle \{ S_i \} \rangle$, and they commute with each other. 
The logical space of the code is defined as the subspace which has eigenvalue $+1$ for all the stabilizer operators, i.e., $S_i \ket{\psi} = \ket{\psi}$ for all $i$.
We denote the normalizer of the stabilizer operators as $\mathcal{L}$. We call elements in $\mathcal{L} \setminus \mathcal{L}_I$ as logical operators. Each stabilizer operator acts on the logical space trivially, and each logical operator acts on the logical space non-trivially. 
A distance $d$ of the code is defined as $d := \min_{L \in \mathcal{L}\setminus \mathcal{L}_I} w(L)$. 
The code which encodes $k$ logical qubits in $n$ physical qubit with distance $d$ is called [[$n,k,d$]] code. 

The occurrence of a physical error is detected as the outcome of stabilizer measurement $\bm{s}$, where $\bm{s}^{\rm T} \in \{0,1\}^{n-k}$ and the $i$-th element $s_i$ is the measurement outcome of the $i$-th stabilizer operator $S_i$. 
We call $\bm{s}$ the syndrome vector.
To recover the original state of the logical qubits, we estimate a recovery Pauli operator $\hat{T}(\bm{s}) \in \{I,X,Y,Z\}^{\otimes n}$ from the observed syndrome vector $\bm{s}$ so that the total operation including the physical error acts on the logical space trivially with high probability. 
The mapping from the syndrome $\bm{s}$ to the recovery operator $\hat{T}(\bm{s})$ is called decoder $\hat{T}$. 
The logical error probability $p_{\rm L}$ is defined as the probability with which the total operation becomes logically non-trivial.
Our purpose is to construct efficient decoder $\hat{T}$ which minimizes the logical error probability $p_{\rm L}$.

\subsection{Binary representation of stabilizer code}
It is convenient to translate the calculation in the stabilizer codes into a binary calculation in GF(2). 
In GF(2), addition $\oplus$ is performed with modulo 2. 
We relate the Pauli operators on the $i$-th physical qubit to another representation 
\begin{eqnarray}
I_i \mapsto \sigma_{00}^{(i)}, X_i \mapsto \sigma_{10}^{(i)}, Y_i \mapsto \sigma_{11}^{(i)}, Z_i \mapsto \sigma_{01}^{(i)}.
\end{eqnarray}
Then, a Pauli operator $P$ on the $n$ physical qubits can be described as 
\begin{eqnarray}
P = \alpha \bigotimes_{i=1}^{n} \sigma_{v_i v_{n+i}}^{(i)},
\end{eqnarray}
where $\alpha \in \{\pm 1, \pm i\}$ and $v_i \in \{0,1\}$ ($1 \leq i \leq 2n$).
We define a binary mapping 
\begin{eqnarray}
b(P) := \bm{v},
\end{eqnarray}
where $\bm{v} := (v_1,v_2,\cdots , v_{2n-1}, v_{2n}) \in \{0,1\}^{2n}$ is a row vector, for the Pauli operator $P = \alpha \bigotimes_{i=1}^{n} \sigma_{v_i v_{n+i}}$. 
For arbitrary two Pauli operators $P$ and $P'$, $b(P)=b(P')$ means that the two Pauli operators are equivalent up to a global phase.
The product of two Pauli operators $P$ and $P'$ is represented by the sum $b(PP') = b(P) \oplus b(P')$.
With $2n \times 2n$ matrix $\Lambda = \left( \begin{matrix} 0 & I \\ I & 0 \end{matrix} \right)$, the commutation relation of two Pauli operators $P$ and $P'$ is given by $b(P) \Lambda b(P')^{\rm T}$, which is 0 if $P$ and $P'$ commute, and 1 if anti-commute. 
We denote this commutation relation in terms of the binary representation $\bm{v},\bm{v}' \in \{0,1\}^{2n}$ as $c(\bm{v},\bm{v}') := \bm{v} \Lambda \bm{v}^{\prime {\rm T}}$.
The weight of the binary representation of a Pauli operator $w(\bm{v})$ is defined so as to be $w(b(P)) = w(P)$, which is equivalent to define the weight as the number of indices $i$ ($1\leq i \leq n$) such that 
\begin{eqnarray}
\label{eq:weightdef}
v_i \oplus v_{i+n} \oplus v_i v_{i+n} = 1.
\end{eqnarray}
We use $h(\bm{v})$ for the hamming weight of $\bm{v}$ as a binary string, namely, the number of indices $i$ ($1\leq i \leq 2n$) such that $v_i = 1$.
We denote the $i$-th row vector of the matrix $M$ as $(M)_i$. 
The length of the vector $\bm{v}$ is represented as $|\bm{v}|$.
With this definition, the normalizer of the stabilizer operators $\mathcal{L}$ is defined as 
\begin{eqnarray}
b(\mathcal{L}) = \{\bm{v} | \bm{v} \in \{0,1\}^{2n}, c(\bm{v}, \bm{v}') = 0, \forall \bm{v}' \in b(\mathcal{L}_I) \} 
\end{eqnarray}
since the normalizer of the stabilizer operators is equivalent to the centralizer of that in the current formalism.
Note that the stabilizer group can be defined with the normalizer $\mathcal{L}$ as 
\begin{eqnarray}
b(\mathcal{L}_I) = \{\bm{v} | \bm{v} \in \{0,1\}^{2n}, c(\bm{v}, \bm{v}') = 0, \forall \bm{v}' \in b(\mathcal{L}) \}.
\end{eqnarray}

With this formalism, QEC is translated as follows. 
The physical error $E$ can be represented as a row binary vector $\bm{e} := b(E) \in \{0,1\}^{2n}$ which occurs with a certain probability $p_{\bm{e}}$. 
The syndrome vector $\bm{s}$ is given by a column vector $\bm{s}(\bm{e}) := H_c \Lambda \bm{e}^{\rm T}$, where $H_c$ is an $(n-k) \times 2n$ matrix of which the $i$-th row vector $(H_c)_i$ is $b(S_i)$. 
The matrix $H_c$ is called check matrix.
In binary representation, we denote a decoder as $\bm{r}$ which maps a given syndrome vector $\bm{s}^{\rm T} \in \{0,1\}^{n-k}$ to a binary representation of a recovery operator $\bm{r}(\bm{s}) \in \{0,1\}^{2n}$. 
It is convenient to define {\it pure error} $\bm{t}(\bm{s})$ \cite{duclos2010renormalization} to represent various vectors succinctly. 
The pure error is a function which maps a syndrome vector $\bm{s}^{\rm T} \in \{0,1\}^{n-k}$ to a vector $\bm{t}(\bm{s}) \in \{0,1\}^{2n}$, and satisfies $\bm{t}(\bm{s}(\bm{e})) \oplus \bm{e} \in b(\mathcal{L})$ for an arbitrary $\bm{e} \in \{0,1\}^{2n}$. 
We also introduce a $2k \times 2n$ generator matrix $G$ such that the elements of $\mathcal{L}$ is uniquely represented as follows:
\begin{eqnarray}
\label{binary_unique_rep}
b(\mathcal{L}) = \{ \bm{l}_0 \oplus \bm{w} G | \bm{l}_0 \in b(\mathcal{L}_I), \bm{w} \in \{0,1\}^{2k} \}.
\end{eqnarray}
Note that the generator matrix $G$ satisfies $H_c \Lambda G^{\rm T} = 0$. 
We define the cosets $\mathcal{L}_{\bm{w}}$ with $\bm{w} \in \{0,1\}^{2k}$ as 
\begin{eqnarray}
\mathcal{L}_{\bm{w}} = \{ \bm{l}_0 \oplus \bm{w} G | \bm{l}_0 \in \mathcal{L}_{0}\}.
\end{eqnarray}
Note that $\mathcal{L}_0 = b(\mathcal{L}_I)$.
Given $\bm{t}(\bm{s})$ and $G$, an arbitrary physical error $\bm{e} \in \{0,1\}^{2n}$ is uniquely decomposed as 
\begin{eqnarray}
\bm{e} = \bm{l}(\bm{e}) \oplus \bm{w}(\bm{e}) G \oplus \bm{t}(\bm{s}(\bm{e}))
\end{eqnarray}
with $\bm{l}(\bm{e}) \in \mathcal{L}_0$ and $\bm{w}(\bm{e}) \in \{0,1\}^{2k}$. 
We say $\bm{w}(\bm{e})$ as the class of $\bm{e}$.

A logical decoder with a recovery operation $\bm{r}(\bm{s})$ can correct an error $\bm{e}$ if and only if $\bm{e} \oplus \bm{r}(\bm{s}(\bm{e})) \in \mathcal{L}_0$. Under an error model $\{p_{\bm{e}}\}$, the logical error probability is given by 
\begin{multline}
p_{\rm L} = {\rm Pr}_{\bm{e} \sim \{p_{\bm{e}}\}} [\bm{e} \oplus \bm{r}(\bm{s}(\bm{e})) \notin \mathcal{L}_{0} \}] \\
= {\rm Pr}_{\bm{e} \sim \{p_{\bm{e}}\}} \left[\bm{r}(\bm{s}(\bm{e})) \oplus \bm{w}(\bm{e})G \oplus \bm{t}(\bm{s}(\bm{e})) \notin \mathcal{L}_0 \right] \}].
\end{multline}

\subsection{Optimal and near-optimal decoders}
An optimal decoder is defined as the decoder which minimizes the logical error probability. 
Let us write the conditional probability of $\bm{w}(\bm{e}) \in \{0,1\}^{2n}$ for a given syndrome vector $\bm{s}$ as
\begin{eqnarray}
    \label{condition:optimaldecoder}
	q_{\bm{s}} (\bm{w}) := {\rm Pr}_{\bm{e} \sim \{p_{\bm{e}}\}} \left[ \bm{w}(\bm{e}) = \bm{w} | \bm{s}(\bm{e}) = \bm{s}  \right].
\end{eqnarray}
Since the decoder is only provided with $\bm{s}$ and distinct recovery operators are needed for correcting errors with different values of $\bm{w}(\bm{e})$, the maximum probability of successful correction given $\bm{s}$ is $\max_{\rm \bm{w} \in \{0,1\}^{2k}} q_{\bm{s}(\bm{w})}$. We thus say a decoder is optimal if $\bm{r}(\bm{s})$ satisfies 
\begin{eqnarray}
\label{cond:optimal}
{\rm Pr}_{\bm{e} \sim \{p_{\bm{e}}\}}\left[ \bm{e} \oplus \bm{r}(\bm{s}) \in \mathcal{L}_0 | \bm{s}(\bm{e}) = \bm{s} \right] = \max_{\bm{w} \in \{0,1\}^{2k}} q_{\bm{s}}(\bm{w}) \nonumber  \\
\end{eqnarray}
for any $\bm{s}$ with 
\begin{eqnarray}
	{\rm Pr}_{\bm{e} \sim \{p_{\bm{e}}\}}\left[ \bm{s}(\bm{e}) = \bm{s} \right] > 0.
\end{eqnarray}
Though the definition of $\bm{w}(\bm{e})$ is dependent on the choice of $\bm{t}(\bm{s})$ and $G$, the optimality of a decoder $\bm{r}(\bm{s})$ is independent of the choice.

Another important definition of a near-optimal decoder is the minimum-distance (MD) decoder. An MD decoder chooses the most probable physical error $\bm{e}^*(\bm{s})$ which satisfies
\begin{eqnarray}
p_{\bm{e}^*(\bm{s})} \geq p_{\bm{e}} \forall \bm{e} \in \{\bm{e} | \bm{s}(\bm{e}) = \bm{e}\}
\end{eqnarray}
as a recovery operation.
Though the maximally likelihood physical error $\bm{e}^*(\bm{s})$ does not necessarily satisfy the condition Eq.\,(\ref{cond:optimal}), it is empirically known that the MD decoder achieves near-optimal performance. 

It is known that the MD decoder can be constructed efficiently in limited cases of the code and the error model. 
For example, we can construct an efficient MD decoder for the surface code under independent bit-flip and phase-flip errors. In this case, we can reduce the decoding problem into minimum-weight perfect matching (MWPM), which can be efficiently solved with blossom algorithm \cite{edmonds1965paths}. 
When bit-flip and phase-flip errors are correlated, we can still construct a decoder with MWPM by ignoring the correlation, resulting in an sub-optimal decoder.
We call such a decoder as a MWPM decoder.

\subsection{Topological code}
We consider two types of the topological codes in this article: surface codes and color codes.
The qubit allocation of the surface code is shown in Fig.\,\ref{fig:sc}.
The $[[2d^2-2d+1,1,d]]$ code and the $[[d^2,1,d]]$ code are shown in Fig.\,\ref{fig:sc}\subref{fig:sc1} and \subref{fig:sc2}, respectively. 
In both figures, the physical qubits are located on the vertices of the colored faces. Each red face represents a stabilizer operator which is a product of Pauli $X$ operators on the physical qubits of its vertices. Each blue face represents one with Pauli $Z$ operators.
\begin{figure*}[tp]
 \centering
  \begin{minipage}[]{0.49\hsize}
   \includegraphics[clip,width=7.5cm]{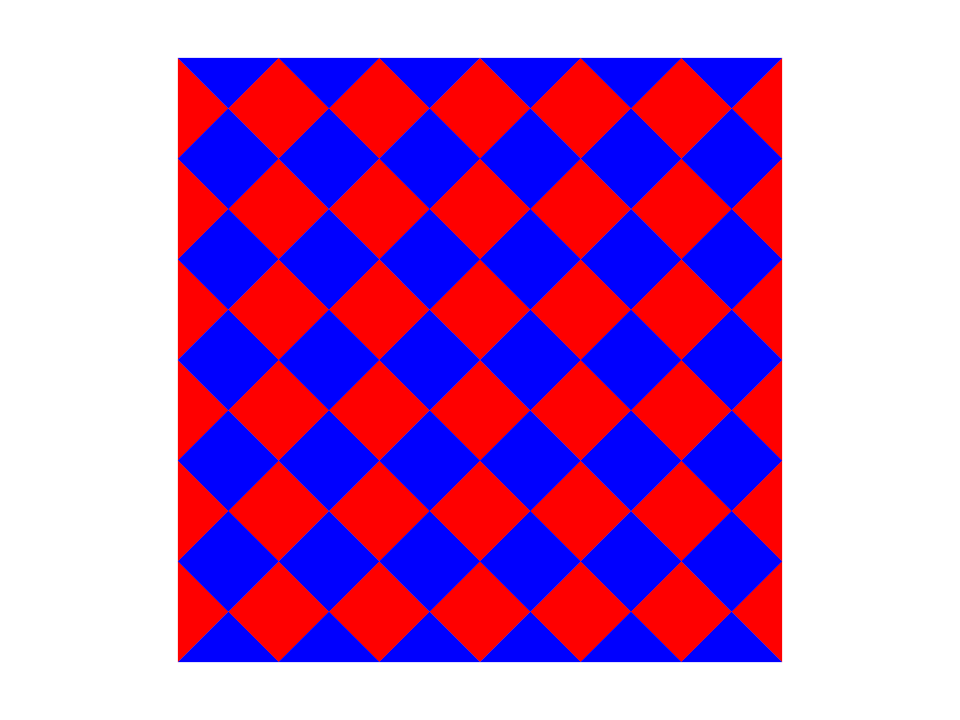}
   \subcaption{}
   \label{fig:sc1}
  \end{minipage}
  \begin{minipage}[]{0.49\hsize}
   \includegraphics[clip,width=7.5cm]{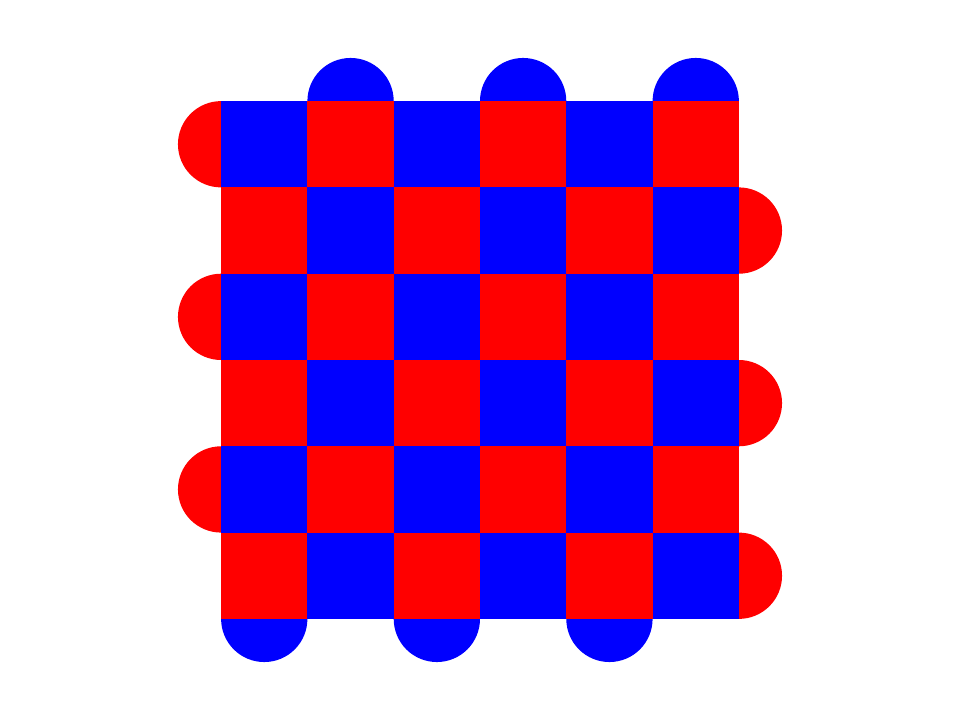}
   \subcaption{}
   \label{fig:sc2}
  \end{minipage}
  \caption{The qubit allocation of the surface codes with (a) [[$2d^2 -2d + 1,1,d$]] code and (b) [[$d^2,1,d$]] code. Each vertex corresponds to a physical qubit. Red and blue faces correspond to stabilizer measurements with $X$ and $Z$ Pauli operators, respectively.}
  \label{fig:sc}
\end{figure*}

The color codes consist of the lattice which has 3-colored faces: red, green, and blue. 
Two types of codes, the [4,8,8]-color code and the [6,6,6]-color code, are shown in Fig.\,\ref{fig:cc}\subref{fig:cc_layout1} and \subref{fig:cc_layout2}, respectively.
The physical qubits are also located on each vertex of the faces. Each colored face represents a stabilizer operator, including nontrivial Pauli operators for its vertices.
The $[4,8,8]$-color code is a [[$\frac{1}{2}d^2 + d - \frac{1}{2}, 1 , d$]] code, and the $[6,6,6]$-color code is a [[$\frac{3}{4}d^2 + \frac{1}{4},1,d$]] code.
\begin{figure*}[tp]
 \centering
  \begin{minipage}[]{0.49\hsize}
   \includegraphics[clip,width=7.5cm]{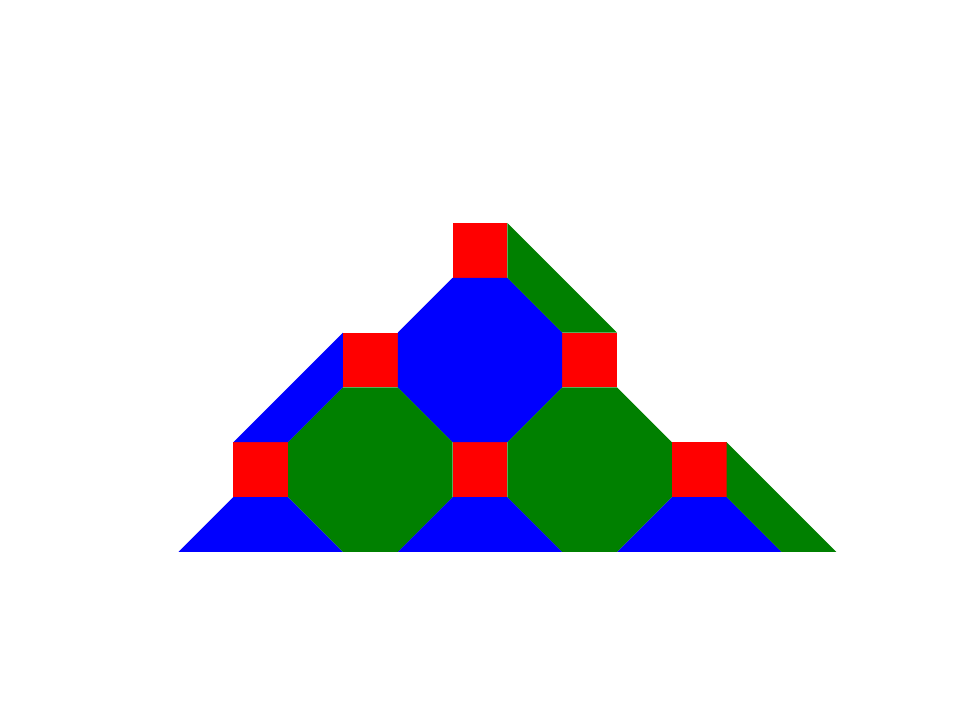}
   \subcaption{}
   \label{fig:cc_layout1}
  \end{minipage}
  \begin{minipage}[]{0.49\hsize}
   \includegraphics[clip,width=7.5cm]{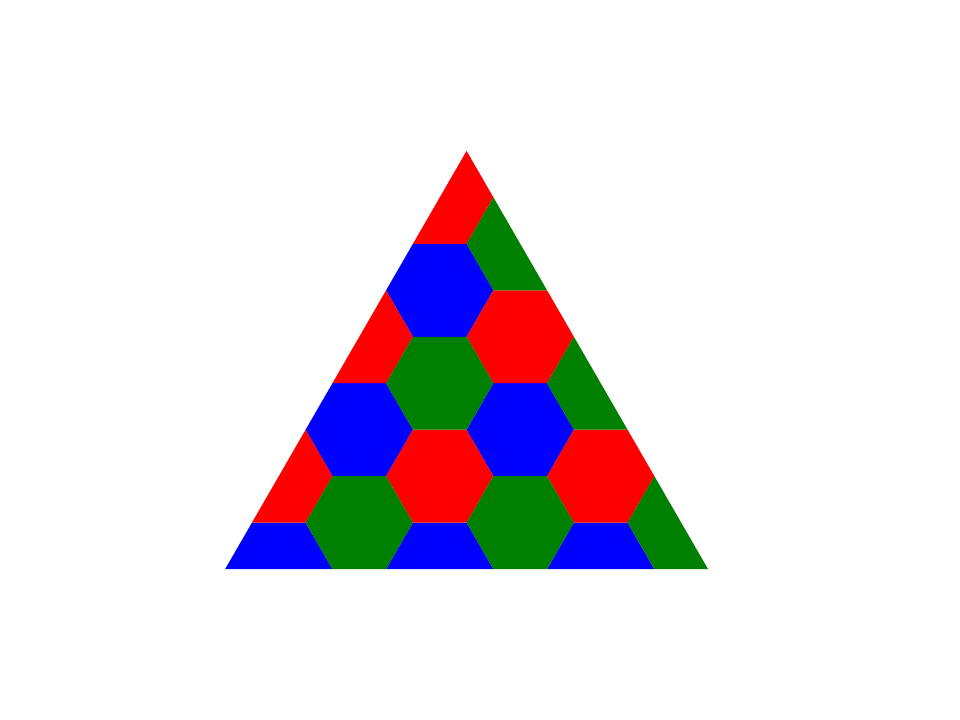}
   \subcaption{}
   \label{fig:cc_layout2}
  \end{minipage}
 \caption{The qubit allocation of the [4,8,8]-color code and the [6,6,6]-color code. Each vertex corresponds to a physical qubit, and each face corresponds to a stabilizer operator. }
 \label{fig:cc}
\end{figure*}

\subsection{Supervised machine learning}

Supervised machine learning is a branch of artificial intelligence that requires a training data set $\{(\bm{x}_1,\bm{y}_1),\ldots,(\bm{x}_N,\bm{y}_N)\}$ which consists of feature data $\bm{x}_i$ and its corresponding label data $\bm{y}_i$. Its aim is to prepare a model that takes the feature data as input and outputs an inferred label for it. The model has a predetermined structure and {\it trainable} parameters $\bm{\theta}$. 

Unlike a simple dictionary, the model is expected to infer a label even for an unseen feature data. 
This is achieved by optimizing the model parameters $\bm{\theta}$ for the training data set. This process is commonly called {\it training}.
Specifically, during its training, the difference between the output of the model $\bm{y}'$ to a feature and the correct label $\bm{y}$ is evaluated with real-valued loss function $L(\bm{y},\bm{y}')$. 
The loss is minimized if and only if the prediction is exactly the same as the correct label. 
The training data is used to optimize the model parameters $\bm{\theta}$ to reduce the loss. This can be done with standard optimization methods such as stochastic gradient descent:
\begin{eqnarray}
	\bm{\theta} \leftarrow \bm{\theta} - \gamma \nabla_{\bm{\theta}} L ,
\end{eqnarray}
where $\gamma\in \mathbb{R}$ is a learning rate and $L$ is calculated for a randomly chosen subset, called a {\it batch}, of the training data set. 
As we can see here, it is required that the loss function should be differentiable, such as L2 distance $||\bm{y}-\bm{y}'||^2_2$.
Once trained, we can apply the model to an unseen feature data, and obtain its predicted label with simple calculations of the network parameters and the input feature data.

Artificial Neural Network (ANN) is a machine learning model inspired by neural structure found in nature. Here, we assume that neurons are real-valued functions and a layer $\textit{\textbf{h}}$ is a vector of the neurons.
Multilayer perceptron (MLP) is one of the simplest ANN  which, as its name suggests, consists of multiple layers of neurons including the input and output layers. Here each neuron in a layer is connected to all neurons in the neighboring layers with trainable weights and biases, and yet completely independent of the other neurons in its own layer. Mathematically, this can be described as
\begin{eqnarray}
\textit{{h}}_{i}^{(n)}=A(\sum_iW_{ij}^{(n,n-1)} \textit{h}_j^{(n-1)} +b_i^{(n)})
\end{eqnarray}
where $A$ is a nonlinear activation function, $\textit{h}_{i}^{(n)}$ is the $i$-th neuron in the $n$-th layer, $b_i^{(n)}$ is the bias added to the $i$-th neuron in the $n$-th layer, and $W_{ij}^{(n,n-1)}$ is the weight connecting the $i$-th neuron in the $n$-th layer to the $j$-th neuron in the ($n-1$)-th layer.
We illustrate this in Fig.\,\ref{fig:mlp}.
\begin{figure}[tp]
	\centering
	\includegraphics[clip, width=7.5cm]{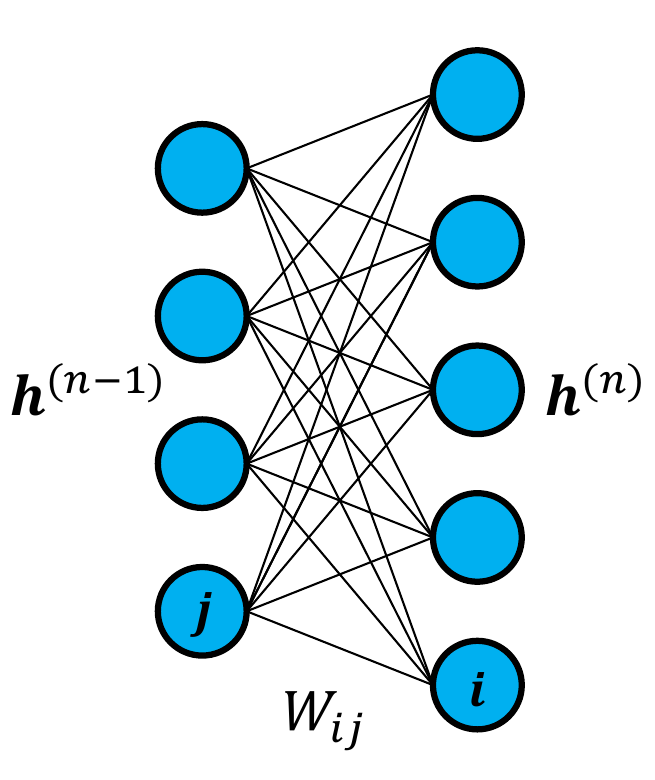}
	\caption{Feed forward network. The $j$-th neuron in the ($n-1$)-th layer is connected to the $\textit{i}$-th neuron in the $\textit{n}$ via weight $W_{ij}$.} \label{fig:mlp}
\end{figure}
Here, the model parameters are the weights and biases. In this model, the input information propagates in forward direction, i.e., from the input nodes to the output nodes. At the output nodes, the loss value is calculated from the model output and the correct label. 
In order to update the model parameters, the gradient of the loss function $\nabla_{\bm{\theta}} L$ is evaluated with the back-propagation method. 
According to the universal approximation theorem \cite{hornik1989multilayer}, any continuous function can be approximated by an MLP model of a finite size, though its structure is simple and compact.
Thus, we expect a neural decoder with a MLP model can achieve near-optimal performance under an appropriate training process.


\section{ Construction of tasks of machine-learning-based decoders}
In general, achievable accuracy in machine learning with a given size of training data depends on the formulation of the prediction task. In order to construct a near-optimal neural decoder, it is vital to consider what is a preferable formulation of the prediction task. However, this point has not been discussed in a unified view in the existing methods \cite{torlai2017neural,varsamopoulos2017decoding,baireuther2018machine,krastanov2017deep}.
In this section, we discuss how the decoding problem should be formulated as a task of machine learning in order to achieve near-optimal performance. 
To this end, we propose a general framework, which we call {\it linear prediction framework}. 
In this framework, we can analytically study the behavior of the neural decoder, and can discuss requirements for achieving near-optimal performance. 
Based on the discussion, we propose a criterion, {\it normalized sensitivity}, which should be optimized in defining the label for constructing a good decoder.
We show specific constructions which minimize normalized sensitivity for the surface codes and the color codes, which we call {\it uniform data construction}.
Then, we numerically confirm that the performance of the neural decoder is improved with the construction. We also confirm that this construction is also applicable to the color codes.

\subsection{Linear prediction framework}
\label{LPF_sec}
In order to discuss the behavior of the neural decoder in a unified view, we consider a neural decoder with the following two specifications.
First, the neural decoder uses the syndrome vector $\bm{s}$ as the feature data to be fed to the trainable model. 
Second, the label data is a binary vector, and the correct label is linearly generated from the physical error vector $\bm{e}$ in GF(2). 
We call a linearly generated label vector $\bm{g}$ as a {\it diagnosis}, and a matrix $H_g$ which generates the diagnosis $\bm{g}:=H_g \Lambda \bm{e}^{\rm T}$ as a diagnosis matrix. We denote the length of the diagnosis vector $\bm{g}$ as $L_g$.
The recovery operator $\bm{r}$ is calculated from the predicted diagnosis $\bm{g}$ and the syndrome $\bm{s}$.
We use an assumed physical error distribution $\{p_{\bm{e}}\}$ only for generating a training data set $\{(\bm{s}_i, \bm{g}_i)\}$, and do not use it for constructing $H_g$ or in the calculation of the recovery operator $\bm{r}$ from $\bm{g}$ and $\bm{s}$.
Though this framework restricts the label to be linearly generated from the physical error, this is general enough to formulate all the constructions described in the existing methods as special cases \cite{torlai2017neural,varsamopoulos2017decoding,baireuther2018machine,krastanov2017deep,breuckmann2018scalable} with small technical exceptions.

Since the actual performance of the neural decoder depends on many factors such as configurations of the training process, the size of the training data set, and details of the network construction, we start with considering the problem under an ideal limit. 
We first consider the problem under the simple 0-1 loss function with an unlimited size of the training data set. 
Then, we relax these impractical assumptions to practical ones.
Though we numerically investigate the case of a single logical qubit ($k=1$) later, we present the formalism for a general value of $k$.

\subsubsection{The neural decoder with the 0-1 loss function and an unlimited training data set}
\label{optimal_delta_sec}
We first consider a hypothetical decoder that can minimize any loss function with an unlimited number of the training data set. Though such an assumption is not practical, it is convenient to reveal the conditions for performing optimal decoding with machine learning in the ideal limit. 
We choose the 0-1 delta function $\delta(\bm{g},\bm{g}')$ as the loss function, which is zero if the predicted and the correct diagnosis are the same, and unity otherwise. 
Let us consider the portion of training data set with a specific value of $\bm{s}$ with ${\rm Pr}_{\bm{e} \sim \{p_{\bm{e}}\}}\left[ \bm{s}(\bm{e}) = \bm{s} \right] > 0$.
If the neural decoder returns diagnosis $\bm{g}$ for the input $\bm{s}$, the total loss for this portion is proportional to the following value,
\begin{eqnarray}
\label{eq:delta_loss}
	L^{(\delta)}_{\bm{s}}(\bm{g}) &:=& \mathbb{E}_{\bm{e} \sim \{p_{\bm{e}}\}} \left[ \delta(\bm{g}, H_g \Lambda \bm{e}^{\rm T}) \middle| \bm{s}(\bm{e}) = \bm{s} \right] \nonumber \\
	&=& 1 - {\rm Pr}_{\bm{e} \sim \{p_{\bm{e}}\}} \left[ H_g\Lambda \bm{e}^{\rm T} = \bm{g} \middle| \bm{s}(\bm{e}) = \bm{s} \right].
\end{eqnarray} 
Let $\bm{g}^{(\delta)}(\bm{s})$ be the output of the ideally trained neural decoder. 
Since it should minimize the total loss for every $\bm{s}$, it satisfies 
\begin{eqnarray}
\label{eq:delta_minimized_loss}
L^{(\delta)}_{\bm{s}} (\bm{g}^{(\delta)}(\bm{s})) &=& \min_{\bm{g}} L^{(\delta)}_{\bm{s}} (\bm{g}).
\end{eqnarray} 
We call this ideal decoder a {\it delta diagnosis decoder} and $\bm{g}^{(\delta)}(\bm{s})$ a {\it delta diagnosis vector}.

We show the condition for a diagnosis matrix $H_g$ to guarantee that we can perform the optimal decoding with the delta diagnosis decoder. 
To this end, we define a property of the diagnosis matrix and introduce a set of diagnosis vectors as follows.

\begin{definition} {\it faithful diagnosis matrix} ---
Given a check matrix $H_c$, we say diagnosis matrix $H_g$ is {\it faithful} if 
\begin{eqnarray}
\label{eq:faithfuldef}
{\rm span}(\{(H_{cg})_i\}) = b(\mathcal{L}),
\end{eqnarray}
or equivalently, 
\begin{eqnarray}
H_{cg} \Lambda \bm{e}^{\rm T} = 0 \leftrightarrow \bm{e} \in \mathcal{L}_{0},
\end{eqnarray}
where 
\begin{eqnarray}
H_{cg} := \left( \begin{matrix} H_c \\ H_g \end{matrix}\right).
\end{eqnarray}
\end{definition}

\begin{definition} {\it faithful diagnosis vectors} ---
Given a check matrix $H_c$, a pure error $\bm{t}(\bm{s})$, and a faithful diagnosis matrix $H_g$, we define $2^{2k}$ faithful diagnosis vectors $\{ \bm{g}_{\bm{s}}(\bm{w}) \}$ ($\bm{w} \in \{0,1\}^{2k}$) associated with a syndrome vector $\bm{s}$ by 
\begin{eqnarray}
\bm{g}_{\bm{s}}(\bm{w}) := H_g \Lambda (\bm{w} G \oplus \bm{t}(\bm{s}))^{\rm T}.
\end{eqnarray}
\end{definition}
Note that the faithful condition of $H_g$ implies that 
\begin{eqnarray}
\bm{w} \mapsto \bm{g}_{\bm{s}}(\bm{w})
\end{eqnarray}
is injective and 
\begin{eqnarray}
\label{eq:faithful_prob}
H_g \Lambda \bm{e}^{\rm T} = \bm{g}_{\bm{s}}(\bm{w}(\bm{e})),
\end{eqnarray}
with $\bm{s} = H_c \Lambda \bm{e}^{\rm T}$.
As a result, when $H_g$ is faithful, we have 
\begin{eqnarray}
\label{eq:loss_faithful}
1- L^{(\delta)}_{\bm{s}}(\bm{g}) = {\rm Pr}_{\bm{e}\sim \{p_{\bm{e}}\}} \left[ \bm{g}_{\bm{s}}(\bm{w}(\bm{e})) = \bm{g} | \bm{s}(\bm{e}) = \bm{s} \right]
\end{eqnarray}
from Eqs.\,(\ref{eq:delta_loss}) and (\ref{eq:faithful_prob}).
Then the injective property of $\bm{g}_{\bm{s}}(\bm{w})$ leads to 
\begin{eqnarray}
\label{eq:loss_and_faithprob}
1- L^{(\delta)}_{\bm{s}}(\bm{g}_{\bm{s}}(\bm{w})) =  q_{\bm{s}}(\bm{w}),
\end{eqnarray}
where $q_{\bm{s}}(\bm{w})$ is defined in Eq.\,(\ref{condition:optimaldecoder}).

When the diagnosis matrix is faithful, we can construct an optimal decoder as follows.
From Eqs.\,(\ref{eq:delta_minimized_loss}) and (\ref{eq:loss_and_faithprob}), we see that the delta diagnosis vector $\bm{g}^{(\delta)}(\bm{s})$ is one of the faithful diagnosis vectors.
We can thus write it in the form 
\begin{eqnarray}
\label{eq:faithful_to_faithprob}
\bm{g}^{(\delta)}(\bm{s}) = \bm{g}_{\bm{s}}(\bm{w}^* (\bm{s})).
\end{eqnarray}
Eqs.\,(\ref{eq:delta_minimized_loss}), (\ref{eq:loss_and_faithprob}), and (\ref{eq:faithful_to_faithprob}) imply that 
\begin{eqnarray}
1-q_{\bm{s}}(\bm{w}^*(\bm{s})) &=& L_{\bm{s}}^{(\delta)} (\bm{g}^{(\delta)}(\bm{s})) \nonumber \\
&=& \min_{\bm{w} \in \{0,1\}^{2k}} (1-q_{\bm{s}}(\bm{w})) \nonumber \\
&=& 1-\max_{\bm{w} \in \{0,1\}^{2k}} q_{\bm{s}}(\bm{w}).
\end{eqnarray}
Since $\bm{g}_{\bm{s}}(\bm{w})$ is injective, one can calculate $\bm{w}^*(\bm{s})$ from the diagnosis $\bm{g}^{(\delta)} (\bm{s})$ and syndrome $\bm{s}$. 
The recovery operator is then chosen as 
\begin{eqnarray}
\bm{r}(\bm{s}) = \bm{w}^*(\bm{s}) G \oplus \bm{t}(\bm{s}).
\end{eqnarray}
For the optimality, we have 
\begin{eqnarray}
{\rm Pr}_{\bm{e} \sim \{p_{\bm{e}}\}}\left[ \bm{e} \oplus \bm{r}(\bm{s}) \in \mathcal{L}_0 | \bm{s}(\bm{e}) = \bm{s} \right] &=& q_{\bm{s}}(\bm{w}^*(\bm{s})) \nonumber \\
&=& \max q_{\bm{s}}(\bm{w})
\end{eqnarray}
for any $\bm{s}$ with ${\rm Pr}_{\bm{e} \sim \{p_{\bm{e}}\}}\left[ \bm{s}(\bm{e}) = \bm{s} \right] > 0$, which satisfies Eq.\,(\ref{cond:optimal}).

We can also prove a converse statement for the cases where $H_g$ is not faithful (see Appendix A), arriving at the following lemma.
\begin{lemma}
\label{lemma_logical_decoder}
If the diagnosis matrix $H_g$ is faithful, there exists a map $\bm{r}^*(\bm{g},\bm{s})$ such that the decoder with $\bm{r}(\bm{s}) = \bm{r}^*(\bm{g}^{(\delta)}(\bm{s}),\bm{s})$ is optimal for arbitrary distribution $\{p_{\bm{e}}\}$.
If the diagnosis matrix $H_g$ is not faithful, no such map exists.
\end{lemma}

This lemma implies that we can perform optimal decoding with the delta diagnosis decoder only when the diagnosis matrix $H_g$ is faithful. 
Note that the set of the faithful vectors $\{\bm{g}_{\bm{s}}(\bm{w})| \bm{w} \in \{0,1\}^{2k}\}$ is independent of the choice of the generator $G$ and the pure error $\bm{t}(\bm{s})$. Whether we can perform the optimal decoding or not is dependent only on the construction of $H_g$.

\subsubsection{The neural decoder with the L2 loss function and an unlimited training data set}
\label{L2dec}
In this subsection, we replace the 0-1 loss function with a more practical one, which is the squared L2 distance.
We still consider the limit of an infinite size of the training data set and the perfect loss minimization.
In this case, the total loss for a fixed $\bm{s}$  under an unlimited training data set is proportional to the following value.
\begin{eqnarray}
	\label{def:l2loss}
	L^{(\rm L2)}_{\bm{s}} (\bm{g}) = \mathbb{E}_{\bm{e} \sim \{p_{\bm{e}}\}} \left[ || \bm{g} - H_g \Lambda \bm{e}^{\rm T} ||^2_2 \middle| \bm{s}(\bm{e}) = \bm{s}\right] 
\end{eqnarray}
We define a decoder which is ideally trained with the L2 loss function as an {\it L2 diagnosis decoder}. 
We also call the output of the L2 diagnosis decoder as an {\it L2 diagnosis vector} $\bm{g}^{(\rm L2)}(\bm{s})$.
The L2 diagnosis vector satisfies the following equation.
\begin{eqnarray}
	\label{def:l2diagloss}
	L^{(L2)}_{\bm{s}} (\bm{g}^{(\rm L2)}(\bm{s})) &=& \min_{\bm{g} \in \{0,1\}^{L_g}} L^{(\rm L2)}_{\bm{s}} (\bm{g}).
\end{eqnarray} 
When the chosen diagnosis matrix is faithful, we can analytically solve $\bm{g}^{({\rm L2})}(\bm{s})$ by differentiating Eq.\,(\ref{def:l2loss}), and the L2 diagnosis vector can be written as follows.
\begin{eqnarray}
\label{def:l2diagnosis}
\bm{g}^{({\rm L2})}(\bm{s}) &:=& \sum_{\bm{w} \in \{0,1\}^{2k}} q_{\bm{s}}(\bm{w}) \bm{g}_{\bm{s}}(\bm{w})
\end{eqnarray}
Let us define a column vector of order $2^{2k}$ as 
\begin{eqnarray}
\label{eq:probarray}
\bm{q}_{\bm{s}} := (q_{\bm{s}}(0^{2k}), \ldots q_{\bm{s}}(1^{2k}))^{\rm T}.
\end{eqnarray}
It satisfies the following matrix equation:
\begin{eqnarray}
\label{cond:L2inverse}
\left( \begin{matrix}
	\hat{\bm{g}}^{({\rm L2})}(\bm{s}) \\ 1 
\end{matrix} \right)
 = D_{\bm{s}} \bm{q}_{\bm{s}},
\end{eqnarray}
where
\begin{eqnarray}
	D_{\bm{s}} = \left( \begin{matrix} 
		\bm{g}_{\bm{s}}(0^{2k}) & \cdots & \bm{g}_{\bm{s}}(1^{2k}) \\ 1 & \cdots & 1 
	\end{matrix} \right).
\end{eqnarray}
We can solve it for $\bm{q}_{\bm{s}}$ if $D_{\bm{s}}$ has a left inverse $D_{\bm{s}}^{-1}$ such that $D_{\bm{s}}^{-1}D_{\bm{s}}=I$ in the real-valued calculation, namely, if the rank of $D_{\bm{s}}$ as a real-valued matrix is $2^{2k}$. 
If the rank is smaller, solution $\bm{q}_{\bm{s}}$ is not unique, and hence it is not always possible to determine $\bm{w}$ that maximizes $q_{\bm{s}}(\bm{w})$, which implies we cannot perform the optimal decoding.

Though the rank condition depends apparently on the syndrome $\bm{s}$, we can formulate it as a condition which is independent of $\bm{s}$.
Any faithful diagnosis $\bm{g}_{\bm{s}}(\bm{w})$ can be written as  
\begin{eqnarray}
\bm{g}_{\bm{s}}(\bm{w}) = H_g \Lambda (\bm{w}G)^{\rm T} \oplus \bm{\delta}(\bm{s})
\end{eqnarray}
with
\begin{eqnarray}
\bm{\delta}(\bm{s}) := H_g \Lambda \bm{t}(\bm{s})^{\rm T} \in \left( \{0,1\}^{L_g} \right)^{\rm T}.
\end{eqnarray}
We define a transformation $\sigma_{\bm{\delta}}$ by 
\begin{eqnarray}
(\sigma_{\bm{\delta}} (\bm{v}) )_i &:=& \delta_i + (-1)^{\delta_i} v_i
\end{eqnarray}
for $\bm{\delta} \in \{0,1\}^{2k}$ and $\bm{v} \in \mathbb{R}^{2k}$.
It is affine, isometric, and involutory.
Since $\bm{g}_{\bm{s}}(\bm{w}) = \sigma_{\bm{\delta}(\bm{s})}(H_g \Lambda (\bm{w}G)^{\rm T})$, we have 
\begin{widetext}	
\begin{eqnarray}
	D_{\bm{s}} &=&	
	\left( \begin{matrix} 
	\sigma_{\bm{\delta}(\bm{s})}(H_g \Lambda ( (0^{2k}) G)^{\rm T}) &
	\cdots &
	 \sigma_{\bm{\delta}(\bm{s})}(H_g \Lambda ( (1^{2k}) G)^{\rm T}) \\
	  1 & \cdots & 1 \end{matrix} \right).
\end{eqnarray}
\end{widetext}
We see that a transformation $\sigma_{\bm{\delta}}$ is an affine transformation, and this transformation satisfies
\begin{eqnarray}
\sigma_{\bm{\delta}}(\sigma_{\bm{\delta}}(\bm{v})) &=& \bm{v} \\
\sigma_{0}(\bm{v}) &=& \bm{v}.
\end{eqnarray}
Thus, when we apply the transformation $\sigma_{\bm{\delta}(\bm{s})}$ to Eq.\,(\ref{cond:L2inverse}), we obtain
\begin{eqnarray}
\left( \begin{matrix}
\sigma_{\bm{\delta}} ( \bm{g}^{({\rm L2})}(\bm{s})) \\ 1 
\end{matrix} \right)
= D \bm{q}_{\bm{s}},
\end{eqnarray}
where
\begin{widetext}	
	\begin{eqnarray}
	\label{matD}
	D &:=&	
	\left( \begin{matrix} 
	H_g \Lambda ( (0,\ldots,0) G)^{\rm T} &
	\cdots &
	H_g \Lambda ( (1,\ldots,1) G)^{\rm T} \\
	1 & \cdots & 1 \end{matrix} \right).
	\end{eqnarray}
\end{widetext}
Thus, we can uniquely calculate $\bm{q}_{\bm{s}}$ for an arbitrary $\bm{s}$ if a matrix $D$ has a left inverse, which is equivalent to the condition that $\{H_g \Lambda (\bm{w}G)^{\rm T} | \bm{w} \in \{0,1\}^{2k}\}$ is affinely independent. 
We will call a diagnosis matrix satisfying this condition to be decomposable:
\begin{definition} {\it decomposable diagnosis matrix} ---
	Given a generator matrix $G$, we say a diagnosis matrix $H_g$ is decomposable if a set of real vectors $\{H_g \Lambda (\bm{w}G)^{\rm T} | \bm{w} \in \{0,1\}^{2k}\}$ is affinely independent, namely, the rank of a matrix $D$ defined in Eq.\,(\ref{matD}) is $2^{2k}$ when we consider $D$ as a real-valued matrix.
\end{definition}
When $H_g$ is faithful, the above definition is independent of $G$, because the set $\{H_g \Lambda (\bm{w}G)^{\rm T} | \bm{w} \in \{0,1\}^{2k}\}$ is independent of $G$ then.

We show a scheme to perform the optimal decoding using L2 diagnosis decoder when a diagnosis matrix is faithful and decomposable.
When $H_g$ is decomposable, there exists a left inverse $D^{-1}$ such that $D^{-1} D = I$ in real vector space.
When we observe a syndrome vector $\bm{s}$, we obtain the L2 diagnosis $\bm{g}^{(\rm L2)}(\bm{s})$ using the trained L2 diagnosis decoder, and calculate $\bm{\delta}(\bm{s}) = H_g \Lambda \bm{t}(\bm{s})$. 
Since the diagnosis matrix is faithful, the probabilities of the faithful diagnosis vectors are given by 
\begin{eqnarray}
	\bm{q}_{\bm{s}} = D^{-1} \left( \begin{matrix} \sigma_{\bm{\delta}(\bm{s})}( \bm{g}^{(\rm L2)}(\bm{s})) \\ 1 \end{matrix} \right).
\end{eqnarray}
Then, we construct a recovery operator as 
\begin{eqnarray}
	\bm{r}(\bm{s}) = \bm{w}^*(\bm{s}) G \oplus \bm{t}(\bm{s}),
\end{eqnarray}
where $\bm{w}^*(\bm{s})$ satisfies 
\begin{eqnarray}
	q_{\bm{s}}(\bm{w}^*(\bm{s})) = \max_{\bm{w}} q_{\bm{s}}(\bm{w}).
\end{eqnarray}
With this recovery operator, we obtain 
\begin{eqnarray}
{\rm Pr}_{\bm{e} \sim \{p_{\bm{e}}\}} \left[ \bm{e} \oplus \bm{r}(\bm{s}) \in \mathcal{L}_0  | \bm{s}(\bm{e}) = \bm{s} \right] = q_{\bm{s}}(\bm{w}^*(\bm{s})),
\end{eqnarray}
and thus this decoder satisfies Eq.\,(\ref{cond:optimal}).

When the diagnosis matrix $H_g$ is faithful, we can also prove a converse statement for the cases where a faithful diagnosis matrix $H_g$ is not decomposable (see Appendix A), arriving at the following lemma.
\begin{lemma}
	\label{lemma:decomposable}
If the diagnosis matrix $H_g$ is faithful and decomposable, there exists a map $\bm{r}^*(\bm{g},\bm{s})$ such that the decoder with $\bm{r}(\bm{s}) = \bm{r}^*(\bm{g}^{(\rm L2)} (\bm{s}), \bm{s})$ is optimal for arbitrary distribution $\{p_{\bm{e}}\}$. 
If the diagnosis matrix $H_g$ is faithful but not decomposable, no such map exists.
\end{lemma}

We show a simple example of a faithful and decomposable matrix $H_g$ in the case of $k=1$. 
We choose vectors $\bm{l}_{01}$, $\bm{l}_{10}$, and $\bm{l}_{11}$ from $\mathcal{L}_{01}$, $\mathcal{L}_{10}$, and $\mathcal{L}_{11}$, respectively. 
We construct $H_g$ and generator $G$ as 
\begin{eqnarray}
H_g &=& \left( \begin{matrix} \bm{l}_{01} \\ \bm{l}_{10} \\ \bm{l}_{11} \end{matrix} \right), \\
G &=& \left( \begin{matrix} \bm{l}_{01} \\ \bm{l}_{10} \end{matrix} \right).
\end{eqnarray}
We see that ${\rm span}(\{(H_g)_i\}) = b(\mathcal{L})$, and thus $H_g$ is faithful.
A set $\{H_g \Lambda (\bm{w}G)^{\rm T} | \bm{w} \in \{00,01,10,11\}\}$ is 
\begin{eqnarray}
	\{(0,0,0)^{\rm T},(0,1,1)^{\rm T},(1,0,1)^{\rm T},(1,1,0)^{\rm T}\},
\end{eqnarray}
which is affinely independent, and thus $H_g$ is decomposable.
We can verify the same by checking the rank of 
\begin{eqnarray}
	D &=& 
	\left( \begin{matrix}
		\bm{g}(00) & \bm{g}(01) & \bm{g}(10) & \bm{g}(11) \\
		1 & 1 & 1 & 1
	\end{matrix} \right) \nonumber \\
	&=&\left( \begin{matrix}
		0 & 0 & 1 & 1\\
		0 & 1 & 0 & 1\\
		0 & 1 & 1 & 0\\
		1 & 1 & 1 & 1\\
	\end{matrix} \right)
\end{eqnarray}
to be $4$ in real vector space.
We can show that there always exists such a faithful and decomposable diagnosis matrix for all $k$ and $H_c$. See Appendix A for the proof.

\subsubsection{The neural decoder with the L2 loss function under a finite training data size}
In practical cases, the size of the training data set is limited, and hence the loss is not perfectly minimized.
This implies that the output diagnosis from the model deviates from the L2 diagnosis vector.
In such a case, it is desirable to construct a decoder such that its prediction is as robust against the deviations as possible.
We introduce a slight modification to the optimal decoding scheme in the last subsection, so that it should applicable to an output diagnosis deviated from the L2 diagnosis vector.

We denote the predicted diagnosis as $\bm{g}^{\rm P}(\bm{s}) \in \mathbb{R}^{L_g}$, which deviates from the L2 diagnosis vector.
Note that $\bm{g}^{\rm P}(\bm{s})$ cannot be represented as a linear combination of the faithful diagnosis vectors in general. 
In order to construct a decoding scheme which is robust to a small deviation, it is natural to extend the scheme employed in Sec.\,\ref{L2dec} such that we project $\bm{g}^{\rm P}(\bm{s})$ to the hyper-plane formed by affine combinations of the faithful diagnosis vectors, and then extract the coefficients $\bm{q}_{\bm{s}}^{\rm P}$ from the projected point. 
This projection and extraction is achieved as follows.
We perform QR decomposition for $D$, and obtain $D = QR$, where $Q$ is an orthogonal matrix, and $R$ is an upper-triangular matrix. 
We construct $D^{-1} = R^{-1} Q^{\rm T}$, which satisfies $D^{-1} D = I$. 
Then, we obtain a predicted vector $\bm{q}_{\bm{s}}^{\rm P}$ as 
\begin{eqnarray}
\label{decodingmethod1}
\bm{q}_{\bm{s}}^{\rm P} = D^{-1} \left( \begin{matrix} \sigma_{\bm{\delta(\bm{s})}}(\bm{g}^{\rm P}(\bm{s})) \\ 1 \end{matrix} \right),
\end{eqnarray}
where $\bm{\delta(\bm{s})} = H_g \Lambda \bm{t}(\bm{s})$.
We construct a recovery operator as 
\begin{eqnarray}
\label{decodingmethod2}
\bm{r}(\bm{s}) = \bm{w}^*(\bm{s}) G \oplus \bm{t}(\bm{s}),
\end{eqnarray}
where $\bm{w}^*(\bm{s})$ satisfies 
\begin{eqnarray}
\label{decodingmethod3}
q_{\bm{s}}^{\rm P}(\bm{w}^*(\bm{s})) = \max_{\bm{w}} q_{\bm{s}}^{\rm P}(\bm{w}).
\end{eqnarray}
Note that though elements of $\bm{q}_{\bm{s}}^{\rm P}$ may be out of $[0,1]$, the above procedure is still well-defined.
\\

\subsubsection{Criterion for diagnosis matrix}
In practice, the number of the training data set is far smaller than the total variation of syndrome vectors $\bm{s}$ when distance $d$ is larger than about $7$.
For example, according to the existing methods \cite{torlai2017neural,varsamopoulos2017decoding,baireuther2018machine,krastanov2017deep}, the size of the training data set is at most $10^9$.
On the other hand, the number of variations in the syndrome, $2^{n-k}$, becomes larger than $10^9$ at the distance $d=7$ for the [[$d^2,1,d$]] surface code. 
This implies that almost all the patterns of the syndrome vector $\bm{s}$ given in experiments are not found in the training data set.
The model should infer the L2 diagnosis vector $\bm{g}^{\rm (L2)}(\bm{s})$ of $\bm{s}$ where $\bm{s}$ is not included in the training data set.
The aim of this subsection is to propose a criterion for $H_g$ which we believe to reflect the robustness of the prediction when we use such a sparsely sampled training data set.

Since the problem is to estimate the vector-valued function $\bm{g}^{\rm (L2)}(\bm{s})$ from a sparsely sampled set of values, its difficulty should depend on how rapidly the function changes its output value as the input value $\bm{s}$ varies. 
From Eqs.\,(\ref{eq:faithful_prob}) and (\ref{def:l2diagnosis}), we see that the function is written as
\begin{eqnarray}
\bm{g}^{\rm (L2)}(\bm{s})= \mathbb{E}_{\bm{e} \sim \{p_{\bm{e}}\}} [ H_g \Lambda \bm{e}^{\rm T} |\bm{s}(\bm{e})=\bm{s}],
\end{eqnarray}
which shows that $\bm{g}^{\rm (L2)}(\bm{s})$ is implicitly determined from the two functions of errors, $\bm{g}(\bm{e}) = H_g \Lambda \bm{e}^{\rm T}$ and $\bm{s}(\bm{e})=H_c\Lambda \bm{e}^{\rm T}$.
In order to quantify how rapidly these function change, let us introduce a sensitivity $m(H)$ of a binary matrix $H$ as
\begin{eqnarray}
\label{sensitivity}
m(H) &:=& \max_{\substack{\bm{e},\bm{e}' \in \{0,1\}^{2n} \\ h(\bm{e}\oplus \bm{e}')=1} } || H \Lambda \bm{e}^{\rm T} - H \Lambda \bm{e}^{\prime {\rm T}} ||^2_2 \nonumber \\
&=& \max_{ \substack{ \bm{\bm{e}} \in \{0,1\}^{2n} \\ h(\bm{e})=1  } }  h(H \Lambda \bm{e}^{\rm T}).
\end{eqnarray}

Using the sensitivity, the variation of $\bm{s}(\bm{e})$ is bounded as
\begin{eqnarray}
||\bm{s}(\bm{e}) - \bm{s}(\bm{e}') ||^2_2 \leq m(H_c) h(\bm{e} \oplus \bm{e}'). 
\end{eqnarray}
In the case of topological codes, $m(H_c)$ is a small constant.
This is because each physical qubit is monitored by at most constant number of the stabilizer operators. 

Suppose that $\bm{g}^{\rm (L2)}(\bm{s})$ is close to one of the faithful diagnosis $\bm{g}_{\bm{s}}(\bm{w}^*)$, and let $S(\bm{s},\bm{w}^*;0)$ be the set of errors $\bm{e}$ satisfying $\bm{w}(\bm{e}) = \bm{w}^*$ and $\bm{s}(\bm{e})=\bm{s}$. 
We further define a set 
\begin{widetext}
\begin{eqnarray}
\label{eq:usefulset}
S(\bm{s},\bm{w}^*;h):= \{\bm{e}| \exists \bm{e}' \, \text{s.t.}\, \bm{e}' \in S(\bm{s},\bm{w}^*;0),  h(\bm{e} \oplus \bm{e}') \leq h \} 
\end{eqnarray}
\end{widetext}
We see that any $\bm{e} \in S(\bm{s},\bm{w}^*;h)$ produces a training data $(\bm{s}',\bm{g}')$ such that
\begin{eqnarray}
||\bm{s}'-\bm{s}||^2_2 &\le& m(H_c) h \\
\label{diagdistupper}
||\bm{g}'- \bm{g}_{\bm{s}}(\bm{w}^*)||^2_2 &\le& m(H_g) h.
\end{eqnarray}
The choice of $H_g$ also affects how precisely $\bm{g}^{\rm (L2)}(\bm{s})$ should be estimated in order to determine $\bm{w}^*$ correctly. 
To quantify this, we consider how far $\bm{g}^{\rm P}(\bm{s})$ can be deviated from a faithful diagnosis $\bm{g}_{\bm{s}}(\bm{w})$ without affecting the decoding method of Eqs.\,(\ref{decodingmethod1}) and (\ref{decodingmethod2}).
When the decoding result changes from $\bm{w}^*=\bm{w}$ to $\bm{w}^*=\bm{w}'$, the solution of Eq.\,(\ref{decodingmethod1}) should satisfy $q_{\bm{s}}^{\rm P}(\bm{w})=q_{\bm{s}}^{\rm P}(\bm{w}')$, namely, $\bm{g}^{\rm P}(\bm{s})$ should be written in the form
\begin{widetext}
\begin{eqnarray}
\label{distance}
\bm{g}^{\rm P} (\bm{s}) = \alpha (\bm{g}_{\bm{s}}(\bm{w}) + \bm{g}_{\bm{s}}(\bm{w}')) + \sum_{\bm{w}'' \neq \bm{w},\bm{w}'} \beta_{\bm{w}''} \bm{g}_{\bm{s}}(\bm{w}'').
\end{eqnarray}
\end{widetext}
We define the minimum boundary distance $M(H_g)$ so as to assure that $\bm{w}^*=\bm{w}$ as long as $||\bm{g}^{\rm P}(\bm{s})- \bm{g}_{\bm{s}}(\bm{w})||^2_2 \le M(H_g)$.
Hence $M(H_g)$ can be explicitly defined as 
\begin{widetext}
\begin{eqnarray}
\label{distancemin}
M(H_g) :=\min_{\bm{w},\bm{w}',\alpha,\{\beta_{\bm{w}''}\}} || (1-\alpha) \bm{g}(\bm{w}) - \alpha \bm{g}(\bm{w}') - \sum_{\bm{w}'' \neq \bm{w},\bm{w}'}  \beta_{\bm{w}''} \bm{g}(\bm{w}'') ||^2_2.
\end{eqnarray}
\end{widetext}
Note that the above definition is independent of $\bm{s}$, since the affine transformation $\sigma_{\bm{\delta}(\bm{s})}$ is isometric. 
$M(H_g)$ is nonzero if and only if $H_g$ is decomposable.

Regarding $M(H_g)$ as the relevant length scale, we define the following quantity to be used as a criterion for a better construction of $H_g$.
\begin{definition} {\it Normalized sensitivity ---}
	We define normalized sensitivity $N(H_g)$ of a faithful and decomposable matrix $H_g$ as
	\begin{eqnarray}
	N(H_g) := \frac{m(H_g)}{M(H_g)},
	\end{eqnarray}
	where $m(H_g)$ is a sensitivity of $H_g$ defined in Eq.\,(\ref{sensitivity}), and $M(H_g)$ is a minimum boundary distance of $H_g$ defined in Eq.\,(\ref{distancemin}).
\end{definition}

Eqs.\,(\ref{diagdistupper}) and (\ref{distancemin}) implies that an error belonging to $S(\bm{s},\bm{w}^*;h)$ with $h\sim (m(H_g)/M(H_g))^{-1}$ leads to a training data useful for estimation of $\bm{g}^{\rm (L2)}(\bm{s})$. 
We thus expect that the use of a diagnosis matrix $H_g$ with a small normalized sensitivity $N(H_g)$ enables high-performance prediction with a small training data set.

\subsubsection{Uniform data construction}
We propose specific constructions which minimize the normalized sensitivity up to the order of $d$ in the case of $k=1$.
We first consider a lower-bound of the normalized sensitivity. When a diagnosis matrix $H_g$ is faithful, each row vector of $H_g$ corresponds to a logical operator or a stabilizer operator. 
We denote the number of the logical operators in the rows of $H_g$ as $n_L$.
The minimum boundary distance $M(H_g)$ is upper-bounded by
\begin{eqnarray}
\label{eq:boundstart}
M(H_g) \leq \frac{n_L}{4}
\end{eqnarray}
Since any logical operator has at least $d$ of one-elements in its binary representation, there are at least $dn_L$ of one-elements in the diagnosis matrix. 
By denoting the number of the one-elements in the diagnosis matrix $H_g$ as $\chi(H_g)$, we have 
\begin{eqnarray}
dn_L \leq \chi(H_g).
\end{eqnarray}
Since there are $2n$ columns in $H_g$, we also have
\begin{eqnarray}
\chi(H_g) \leq 2n \max_i h( (H_g^{\rm T})_i ).
\end{eqnarray}
The sensitivity $m(H_g)$ is equal to the maximum hamming weight of the column vectors of the diagnosis matrix, namely, 
\begin{eqnarray}
\label{eq:boundend}
\max_i h( (H_g)^{\rm T}_i ) = m(H_g).
\end{eqnarray}
From Eqs.\,(\ref{eq:boundstart}) - (\ref{eq:boundend}), we obtain
\begin{eqnarray}
N(H_g) \geq \frac{2d}{n}
\end{eqnarray}
In particular, when we focus on the two-dimensional topological codes such that $n=\Theta(d^2)$ , the order of the normalized sensitivity is lower-bounded as 
\begin{eqnarray}
N(H_g) = \Omega(d^{-1}).
\end{eqnarray}
For surface codes and color codes with the single logical qubit, we found specific constructions of $H_g$ such that $N(H_g)$ scales as $\Theta(d^{-1})$. 
See appendix C for the specific constructions.
We named these constructions as {\it uniform data construction} of the data set, since logical operators corresponding to the rows of $H_g$ are chosen uniformly to cover all the physical qubits.

\subsection{Construction of data set and example}
Let us summarize the discussion in Sec \ref{LPF_sec}. 
Given the check matrix $H_c$ of the code and the error model $\{p_{\bm{e}}\}$, the whole protocol can be described as follows.
\begin{itemize}
\item {\bf Preparation: } We construct a faithful and decomposable diagnosis matrix $H_g$ with a small normalized sensitivity, possibly $N(H_g) = \Theta(d/n)$. 
We choose a pure error $\bm{t}(\bm{s})$ and a generator matrix $G$. 
We perform QR decomposition to a matrix 
\begin{eqnarray}
	D := \left( \begin{matrix} 
		\bm{g}(0^{2k}) & \cdots & \bm{g}(1^{2k}) \\ 1 & \cdots & 1 
	\end{matrix} \right),
\end{eqnarray}
where $\bm{g}(\bm{w}) = H_g \Lambda (\bm{w}G)^{\rm T}$, and obtain $Q$ and $R$.
We calculate the left inverse matrix $D^{-1}$ as 
\begin{eqnarray}
D^{-1}:= R^{-1} Q^{\rm T}.
\end{eqnarray}

\item {\bf Data generation: } We generate a set of physical errors $\{\bm{e}_0, \bm{e}_1, \ldots \}$ with the probability distribution $\{p_{\bm{e}}\}$, and generate data set $\{(\bm{s}_0, \bm{g}_0), (\bm{s}_1, \bm{g}_1), \ldots \}$ from it, where $\bm{s}_i :=H_c \Lambda \bm{e}_i^{\rm T}$ and $\bm{g}_i := H_g \Lambda \bm{e}_i^{\rm T}$.

\item {\bf Training: } The model is trained so that it can predict $\bm{g}$ from $\bm{s}$. The loss of the prediction is defined as the L2 distance between $\bm{g}$ and $\bm{g}^{\rm P}$, where $\bm{g}^{\rm P}$ is a real-valued output vector of the model.

\item {\bf Prediction: } When an observed syndrome $\bm{s}$ is given to the trained model, it predicts $\bm{g}^{\rm P}(\bm{s})$.  
In parallel, we calculate $\bm{\delta}(\bm{s})$ given by 
\begin{eqnarray}
\bm{\delta}(\bm{s}) = H_g \Lambda \bm{t}(\bm{s})^{\rm T}.
\end{eqnarray}
We calculate vector $\bm{q}_{\bm{s}}$ defined in Eq.\,(\ref{eq:probarray}) as 
\begin{eqnarray}
\bm{q}_{\bm{s}}^{\rm P} =D^{-1} \left( \begin{matrix} \sigma_{\bm{\delta}(\bm{s})}(\bm{g}^{\rm P}(\bm{s})) \\ 1 \end{matrix} \right),
\end{eqnarray} 
where $\sigma_{\bm{\delta}(\bm{s})}$ is an affine transformation such that 
\begin{eqnarray}
(\sigma_{\bm{\delta}(\bm{s})}(\bm{v}))_i = \delta_i + (-1)^{\delta_i} v_i.
\end{eqnarray}

We choose $\bm{w}^{\rm P}$ that satisfies 
\begin{eqnarray}
q_{\bm{s}}^{\rm P}(\bm{w}^{\rm P}) = \max_{\bm{w} \in \{0,1\}^{2k}} q_{\bm{s}}^{\rm P}(\bm{w}),
\end{eqnarray}
where $\{ q_{\bm{s}}^{\rm P}(\bm{w}) \}$ is the elements of $\bm{q}_{\bm{s}}^{\rm P}$.
Then, we obtain an estimated recovery operator 
\begin{eqnarray}
\bm{r}(\bm{s}) = \bm{w}^{\rm P} G \oplus \bm{t}(\bm{s}).
\end{eqnarray}
\end{itemize}

We emphasize that the choice of $\bm{t}(\bm{s})$ and $G$ dose not affect the performance of the decoder, since the success of the estimation is independent of them. 
Only the construction of $H_g$ affects the performance of the decoder.

We show a specific example of the decoding scheme. 
For simplicity, we consider the case where there is only bit-flip errors in the [[$2d^2-2d+1,1,d$]] surface code. In this case, it is enough for QEC to consider the stabilizer operator with Pauli $Z$ operators. The simplified picture of the code is shown in Fig.\,\ref{fig:sample}. 
\begin{figure*}[tp]
 \centering
  \begin{minipage}[]{0.24\hsize}
 \includegraphics[clip, width=4cm]{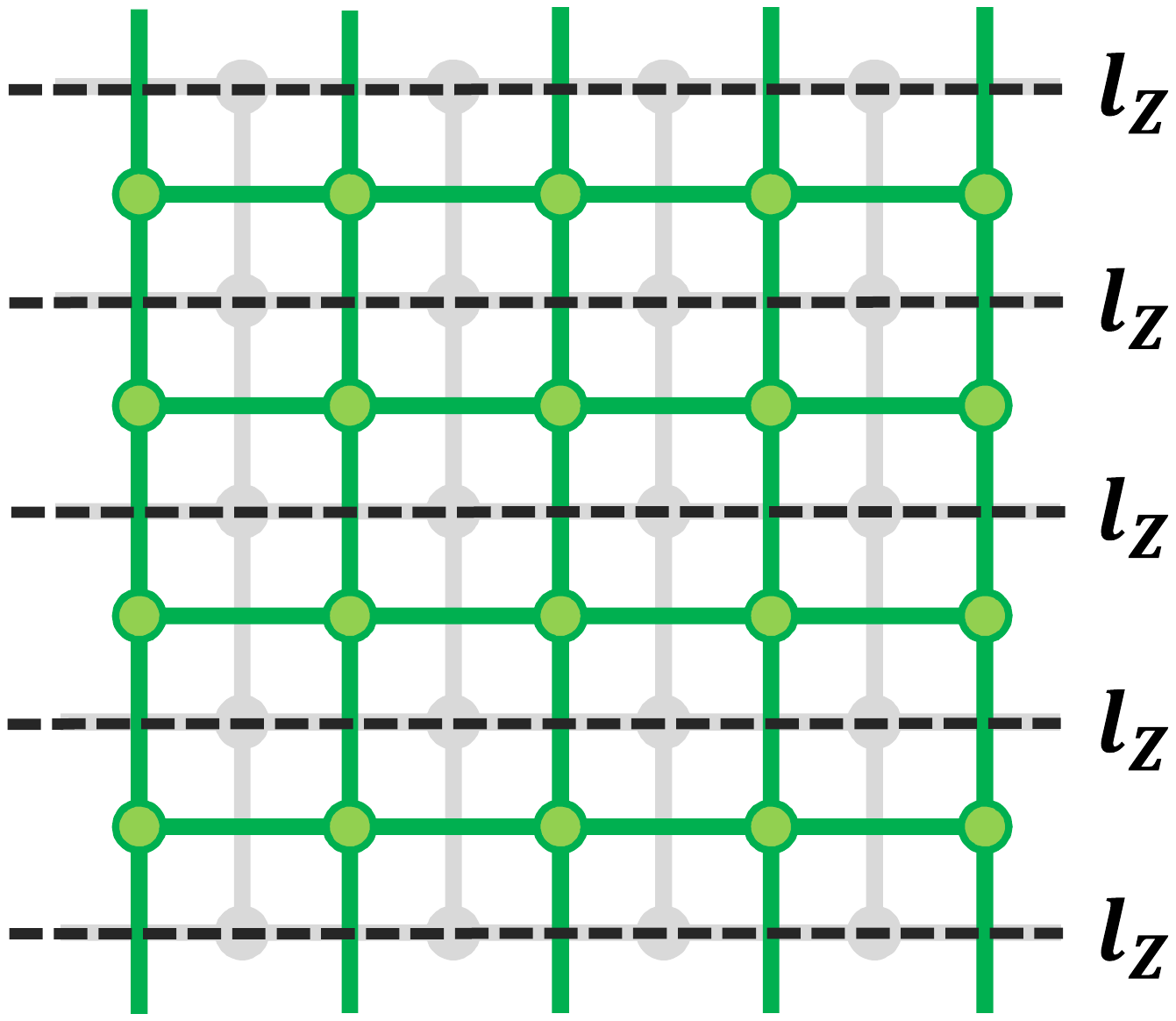}
 \subcaption{}
 \label{fig:sample1}
  \end{minipage}
  \begin{minipage}[]{0.24\hsize}
 \includegraphics[clip, width=4cm]{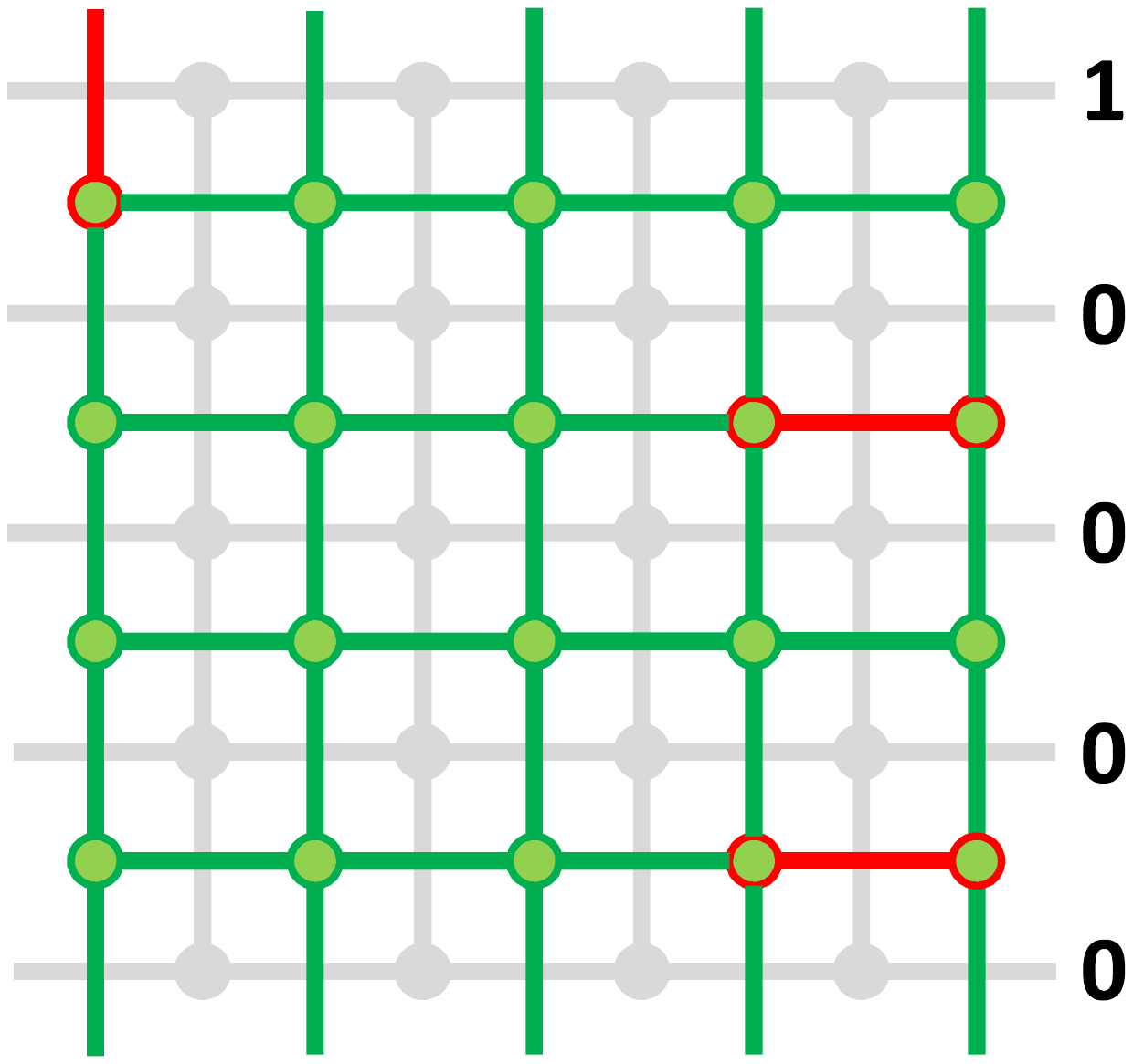}
 \subcaption{}
 \label{fig:sample2}
  \end{minipage}
  \begin{minipage}[]{0.24\hsize}
 \includegraphics[clip, width=4cm]{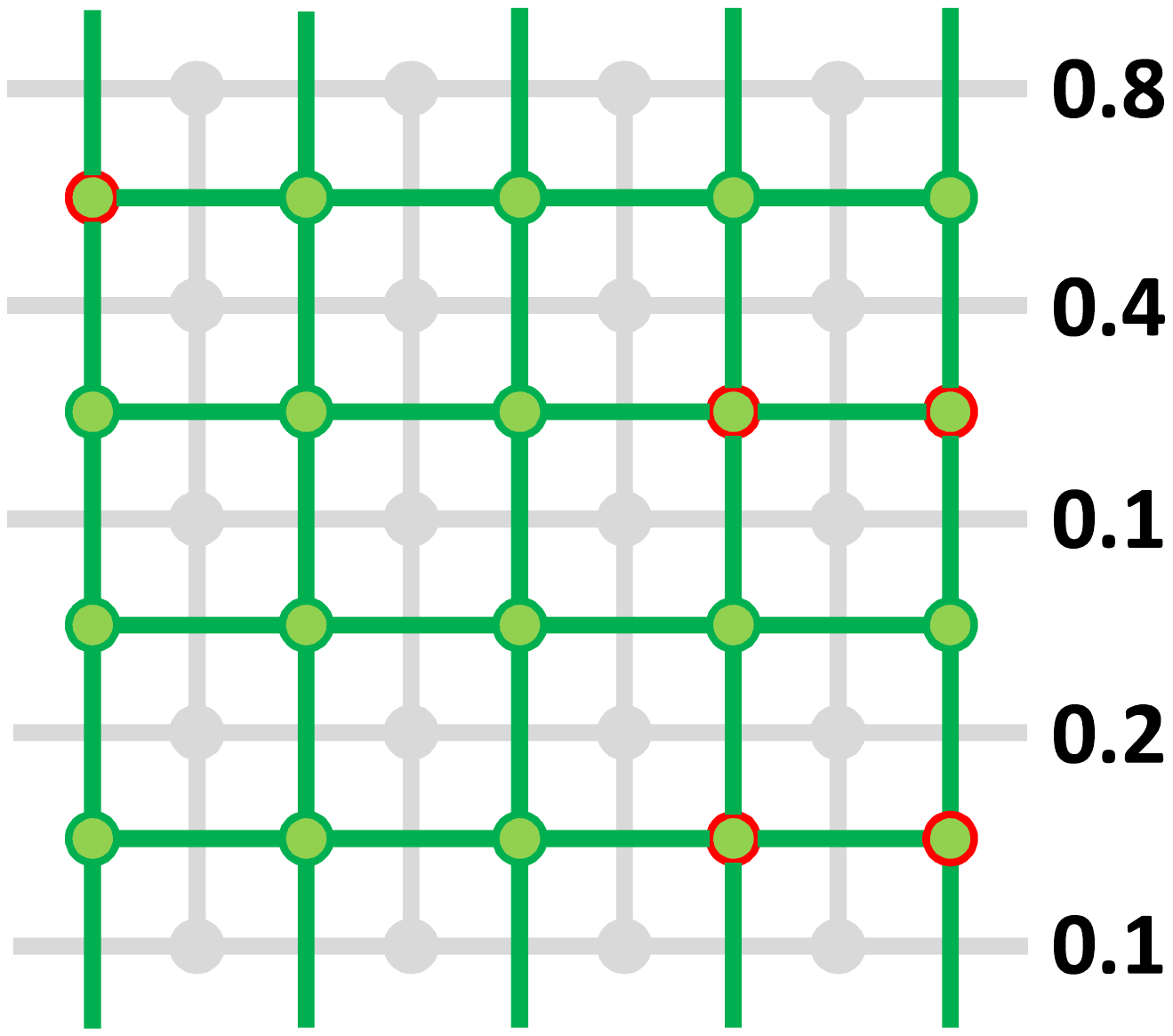}
 \subcaption{}
 \label{fig:sample3}
  \end{minipage}
  \begin{minipage}[]{0.24\hsize}
 \includegraphics[clip, width=4cm]{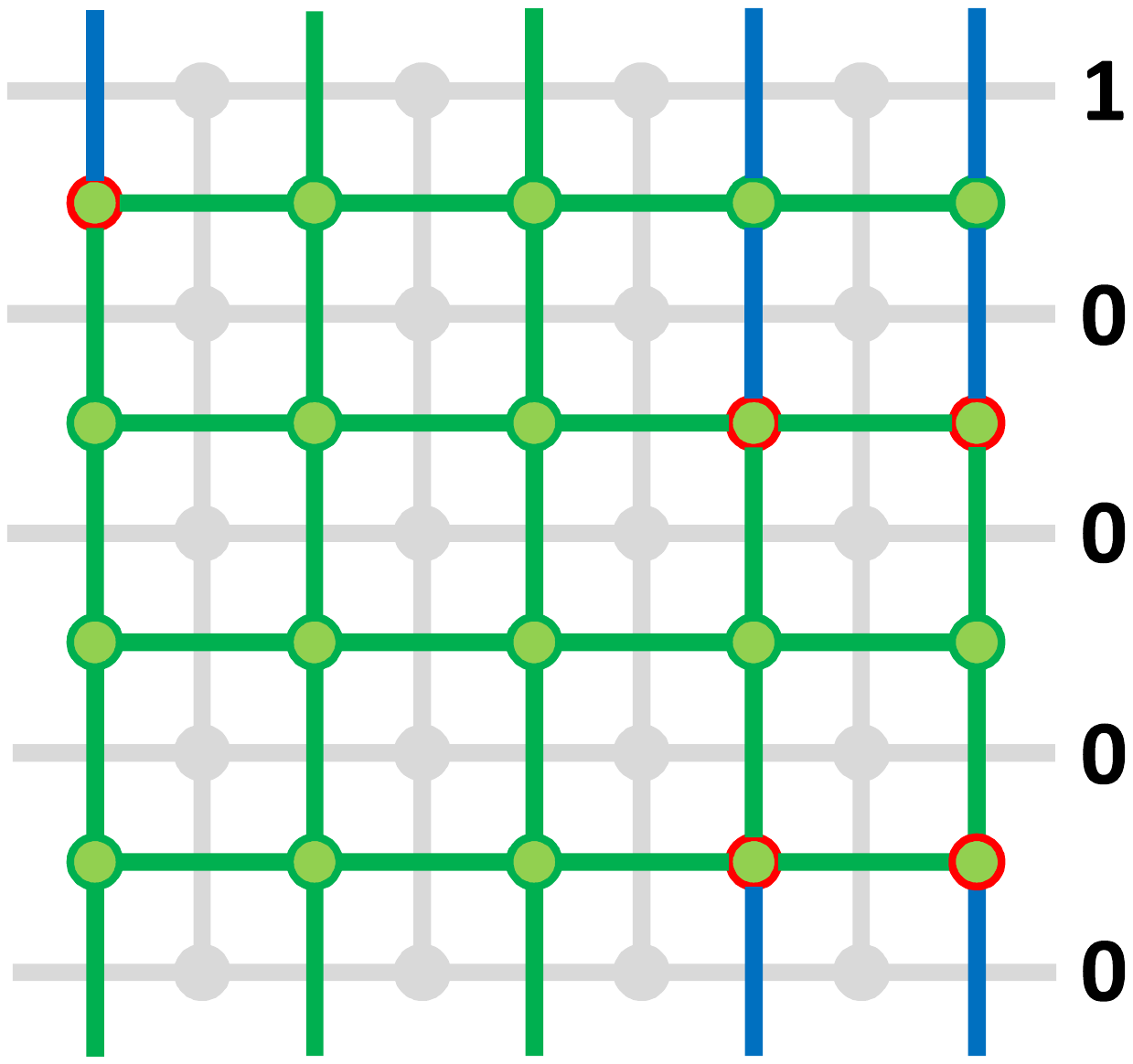}
 \subcaption{}
 \label{fig:sample4}
  \end{minipage}
  	
 \caption{The figures show the decoding process based on proposed scheme. Each picture shows only $Z$ lattice, of which the edge corresponds to whether there is a bit-flip error on the physical qubit or not, and the circle shows whether an error is detected through the syndrome measurement. (a) Five logical $Z$ operators which minimize the normalized sensitivity $\frac{m(H_g)}{M(H_g)}$.  (b) The actual physical error is drawn as red edges, and the detected syndromes as red circles. The binary numbers shown to the right is the diagnosis vector of the physical errors. The neural network learns the relation between the location of the detected syndromes and the diagnosis vector. (c) The real-valued diagnosis vector is predicted by the neural decoder. (d) With the syndrome pattern, faithful diagnosis vector is either $10000$ or $01111$. The chosen  faithful diagnosis vector is $10000$. Accordingly, we choose the recovery operator shown in the figure. In this case, the decoding succeeds. }
 \label{fig:sample}
\end{figure*}

In this picture, a bit-flip error on a physical qubit is represented by the color of the corresponding edge (green: no error, red: error), and the syndrome value is represented by the color of the circle (green: undetected, red: detected). As shown in Fig.\,\ref{fig:sample}\subref{fig:sample1}, The matrix $H_g$ is constructed with logical operators each of which is the product of the Pauli $Z$ operators on the edges crossing the dotted line.
In this case, we see $M(H_g)=O(d)$, $m(H_g)=O(1)$, and $\frac{m(H_g)}{M(H_g)} = O(d^{-1})$.

Suppose that bit-flip errors occur on a set of the physical qubits as shown in Fig.\,\ref{fig:sample}\subref{fig:sample2}. The physical error is detected with the syndrome values as shown in the same figure. The diagnosis vector is calculated as the commutation relation of the chosen logical operators and the physical error. We show the calculated diagnosis on the right side of the lattice. In the training phase, the model learns the relation between the positions of the red circles and the values of the diagnosis vector.
In the prediction phase, only the positions of the red circles are given.  
The trained neural network outputs a real-valued prediction of the diagnosis vector as shown in Fig.\,\ref{fig:sample}\subref{fig:sample3}, for example. 
From this information, we extract the probabilities of the faithful diagnosis, and we choose the faithful diagnosis which is expected to be the most probable, as shown in Fig.\,\ref{fig:sample}\subref{fig:sample4}. 
Since the chosen diagnosis vector is equivalent to the diagnosis vector generated by the actual physical error, this decoding trial is a success.

\subsection{Relation to the existing methods}
\label{section_relation_works}
In this subsection, we explain how the existing methods \cite{torlai2017neural,varsamopoulos2017decoding,baireuther2018machine,krastanov2017deep} can be treated in the linear prediction framework.
The method proposed by Varsamopoulos {\it et al.} \cite{varsamopoulos2017decoding} used an approach similar to the example shown in Sec.\,\ref{L2dec} in the case of $k=1$. 
In this method, a linear map is used for the pure error, which is called a simple decoder.
The pure error is then written in the form $\bm{t}(\bm{s})^{\rm T} = T \bm{s}$, where $T$ is a $2n \times (n-k)$ matrix satisfying $H_c \Lambda T = I$.
The label vector used in this method can essentially be regarded as being generated by a diagnosis matrix defined by 
\begin{eqnarray}
H_g = \left( \begin{matrix} \bm{l}_{01} \\ \bm{l}_{10} \\ \bm{l}_{11} \end{matrix} \right) ( I \oplus \Lambda TH_c).
\end{eqnarray}
We see this is faithful and decomposable constructions.
Let a generator matrix $G$ be 
\begin{eqnarray}
G = \left( \begin{matrix} \bm{l}_{01} \\ \bm{l}_{10} \end{matrix} \right).
\end{eqnarray}
Then, a diagnosis generated from the diagnosis matrix is 
\begin{eqnarray}
\bm{g} = H_g \Lambda \bm{e} = \left( \begin{matrix} w(\bm{e})_1 \\ w(\bm{e})_0 \\ w(\bm{e})_0 \oplus w(\bm{e})_1 \end{matrix} \right),
\end{eqnarray}
where $\bm{w}(\bm{e}) = (w(\bm{e})_0, w(\bm{e})_1)$. 
The method in Ref.\,\cite{varsamopoulos2017decoding} uses a different set of label vectors $\bm{g}'$ called one-hot representation, which has a one-to-one correspondence with $\bm{g}$ as 
\begin{eqnarray}
\bm{g} = (0,0,0)^{\rm T} & \mapsto & \bm{g}'=(1,0,0,0)^{\rm T} \\
\bm{g} = (0,1,1)^{\rm T} & \mapsto & \bm{g}'=(0,1,0,0)^{\rm T} \\
\bm{g} = (1,0,1)^{\rm T} & \mapsto & \bm{g}'=(0,0,1,0)^{\rm T} \\
\bm{g} = (1,1,0)^{\rm T} & \mapsto & \bm{g}'=(0,0,0,1)^{\rm T}.
\end{eqnarray} 
The above relation as real vectors can be written as 
\begin{eqnarray}
\bm{g}' = \frac{1}{2} \left( \begin{matrix} -1 & -1 &-1 & 1 \\ -1 & 1 & 1 & 1 \\ 1 & -1 & 1 & 1 \\ 1 & 1 & -1 & 1 \end{matrix} \right) \left( \begin{matrix} \bm{g} \\ 0 \end{matrix} \right) + \left( \begin{matrix} 1 \\ 0 \\ 0 \\ 0 \end{matrix} \right).
\end{eqnarray}
Since it is an isometric affine transformation, we expect that this transformation has little effect on the performance of the supervised machine learning. 
The matrix $H_g$ is faithful and decomposable, but its normalized sensitivity is $O(1)$. We thus expect that this decoder becomes near-optimal when the training is ideally performed, but the prediction is not robust when the size of the training data set is small.

The method proposed by Baireuther {\it et al.} \cite{baireuther2018machine} mainly focuses on a model applicable to quantum error correction when we perform various counts of repetitive stabilizer measurements by utilizing recurrent neural network. 
They use the commutation relation between the physical error and a logical $Z$ operator as the label, since they only concerned about the logical bit-flip probability with the fixed initial state in the logical space. 
We can thus consider this method as a case of the linear prediction framework.

Torlai {\it et al.} \cite{torlai2017neural}, Krastanov {\it et al.} \cite{krastanov2017deep}, and Breuckmann {\it et al.} \cite{breuckmann2018scalable} took a different approach from the above two \cite{varsamopoulos2017decoding,baireuther2018machine}. 
They used the binary representation of the physical error as the label vector.
In the linear prediction framework, it corresponds to a choice of $H_g=\Lambda$ leading to 
\begin{eqnarray}
\bm{g} = H_g \Lambda \bm{e}^{\rm T} = \bm{e}^{\rm T}
\end{eqnarray}

Since $H_g$ is not faithful, it cannot constitute an optimal decoder even with the delta diagnosis decoder.
Interestingly, the delta diagnosis decoder with this choice of $H_g$ works as an MD decoder, which can be shown by the following lemma.
\begin{lemma}
\label{lemma_physical_decoder}
If the matrix $H_{cg}$ has rank $2n$ in GF(2), there exists a map $\bm{r}^*(\bm{g},\bm{s})$ such that the decoder with $\bm{r}(\bm{s}) = \bm{r}^*(\bm{g}^{(\delta)}(\bm{s}),\bm{s})$ works as an MD decoder for arbitrary distribution $\{p_{\bm{e}}\}$.
If $H_{cg}$ does not have rank $2n$, no such map exists.
\end{lemma}
\begin{proof}
If $H_{cg}$ has rank $2n$, there exists a left inverse binary matrix $H_{cg}^{-1}$ such that $H_{cg}^{-1}H_{cg} = I$.
Then, we can obtain the physical error $\bm{e}$ as 
\begin{eqnarray}
\Lambda H_{cg}^{-1} \left( \begin{matrix} \bm{s} \\ \bm{g} \end{matrix} \right) = \bm{e}^{T}.
\end{eqnarray} 
Thus, we can obtain the most probable physical error $\bm{e}^*(\bm{s})$ from the most probable diagnosis. 

If $H_{cg}$ does not have rank $2n$ in GF(2), there exist two physical errors which generate the same pair of syndrome and diagnosis. We cannot determine which is more probable. Thus, we cannot perform MD decoding when $H_{cg}$ does not have rank $2n$.
\end{proof}

A drawback in this approach is difficulty arising when we replace a loss function with a practical one such as L2 distance. 
In order to satisfy a decomposable property in MD decoding, the length of the diagnosis must be no shorter than $2^{n+k}$ since there are $2^{n+k}$ possible candidates of the most probable physical error. This is not practical when the distance is large, and thus it requires heuristics such as repetitive sampling.

\subsection{Numerical result}
We numerically show that the uniform data construction improves the performance of the neural decoder in the case of $k=1$. 
We trained an MLP model with the uniform data construction, and compare it with other data constructions of the neural decoders. 
We also make a comparison with known decoders such as the MD decoder and the MWPM decoder. 
We choose the [[$d^2,1,d$]] surface code for the comparison, since most of the existing methods were benchmarked with this code. 
We calculated the performance for two types of error models, the bit-flip noise and depolarizing noise. 
The probability distribution of the bit-flip noise is described as follows.
\begin{eqnarray}
p_{\bm{e}} = \begin{cases} p^{w(\bm{e})}(1-p)^{n-w(\bm{e})} & \forall i>n, e_i=0 \\ 0 & \text{otherwise} \end{cases},
\end{eqnarray}
where $p$ is an error probability per physical qubit, and $w(\bm{e})$ is a weight of physical error $\bm{e}$ defined in Eq.\,(\ref{eq:weightdef}).
The probability distribution of the depolarizing noise is described as follows.
\begin{eqnarray}
p_{\bm{e}} = (p/3)^{w(\bm{e})}(1-p)^{n-w(\bm{e})}.
\end{eqnarray}
Note that the occurrences of the bit-flip and phase-flip errors are correlated in the depolarizing noise. 
We first calculated the performance when the physical error probability is around the error threshold, namely, $p=0.1$ for the bit-flip noise and $p=0.15$ for the depolarizing noise. 
The tunable hyper-parameters of the neural network, such as number of layers in network, number of neurons in each layer, and learning rate, are optimized with a grid search for each noise model and for each size of the training data set. 
See Appendix B for the details of the parameter optimization and implementation.

The performance of the neural decoder under the bit-flip noise is shown in Fig.\,\ref{fig:tr}\subref{fig:tr1}. 
The solid lines are the performance of the neural decoder with the uniform data construction. 
The bottom dashed lines represent the logical error probability achievable with the MD decoder. 
The colors red, green, blue, and cyan correspond to distances 5, 7, 9, 11, respectively.

\begin{figure*}[tp]
	\centering
	\begin{minipage}[]{0.49\hsize}
		\includegraphics[clip, width=7.5cm]{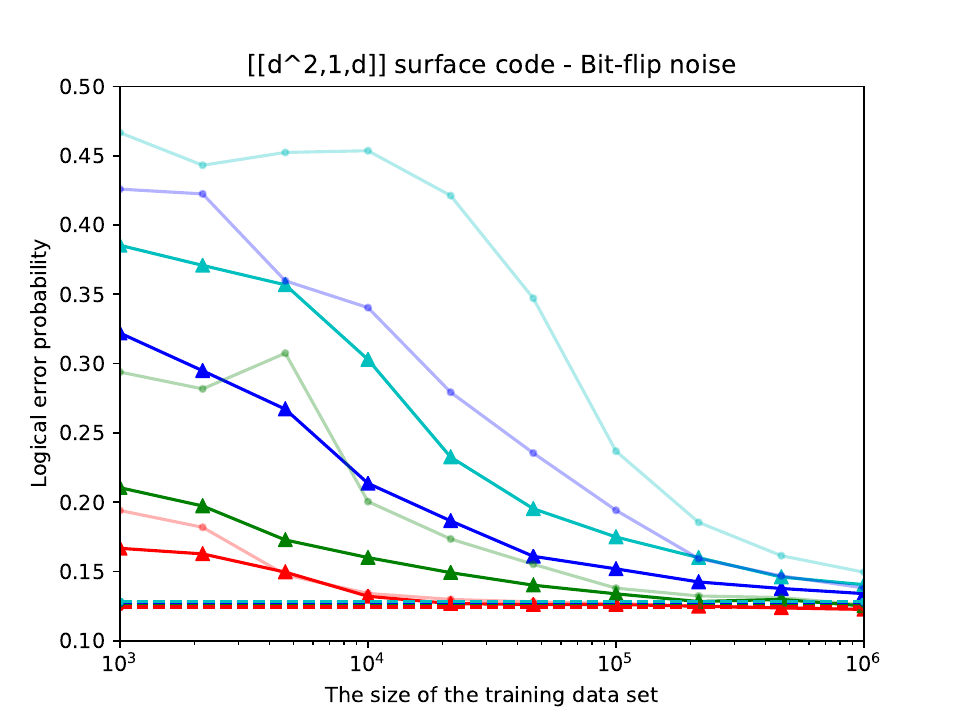}
		\subcaption{}
		\label{fig:tr1}
	\end{minipage}
	\begin{minipage}[]{0.49\hsize}
		\includegraphics[clip, width=7.5cm]{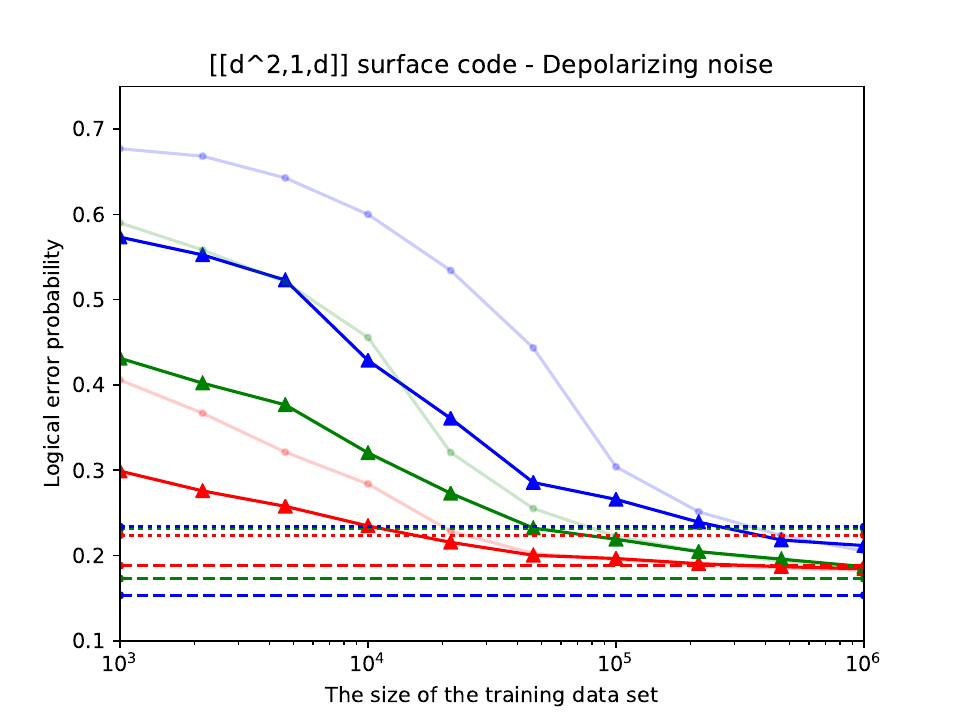}
		\subcaption{}
		\label{fig:tr2}
	\end{minipage}
	\caption{The performance comparison between the neural decoder with the uniform construction (solid lines) and that with short diagnosis construction (pale lines), the MD decoder (dashed lines), and the MWPM decoder (dotted lines) in the case of the [[$d^2,1,d$]] surface code. The logical error probabilities are plotted against of the sizes of the training data set with the fixed physical error probability $p$. We calculated the performance for distances $d=5$ (red), $7$ (green), $9$ (blue), and $11$ (cyan). (a) The case for the bit-flip noise with $p=0.1$. Note that there are no lines of MWPM decoder since the MWPM decoder is equivalent to the MD decoder in this setting. (b) The case for the depolarizing noise with $p=0.15$.}
	\label{fig:tr}
\end{figure*}

Comparing these two types of decoders, we see that the logical error probability of the neural decoder is near-optimal with $10^6$ data set at distance 11. 
On the other hand, there are gaps between the converged logical error probabilities of the neural decoder and that of the MD decoder when the distance is large. 
We speculate that these gaps are caused by imperfect learning of the spatial information of the topological codes, since it is partially improved with the network construction discussed in the next section.

We also implemented the neural decoder with short diagnosis, i.e. the construction with $N_{01}=N_{10}=N_{11}=1$, where $N_{\bm{w}}$ is a number of logical operators in the rows of $H_g$ corresponding to the class $\bm{w}$. 
This is equivalent to the construction which we showed as an example in Sec.\,\ref{L2dec}.
We call this construction, with the normalized sensitivity of $O(1)$, as short diagnosis construction, which is shown as the pale plots in Fig.\,\ref{fig:tr}\subref{fig:tr1}. 
Note that the performance of this decoder depends on the choice of the logical operators. 
We have tried this construction with various choice of the logical operators. 
The plotted data is the best among our trials. 
Although both constructions become near-optimal in the limit of large training data size, we see that the performance with the uniform data construction achieves smaller logical error probability than that with the short diagnosis construction for any size of the training data set. 
%
We have also confirmed that the performance of the neural decoder degrades when the row vectors of $H_g$ consist of the same $O(d)$ logical operators of $X$, $Y$ and $Z$. 
In this case, while the number of the rows in $H_g$ is the same as that of the uniform data construction, the sensitivity $m(H_g)$ becomes $O(d)$, which makes the normalized sensitivity $\frac{m(H_g)}{M(H_g)}$ to be $O(1)$. 
Though these results are not plotted, the performance of this construction is almost the same as the short diagnosis construction. 
These results support our argument that it is essential for the performance of the neural decoder to minimize the normalized sensitivity.

The results with the depolarizing noise are shown in Fig.\,\ref{fig:tr}\subref{fig:tr2}. 
Note that for the surface code under correlated noise such as the depolarizing noise, it is not known how an efficient MD decoder can be constructed. 
We see that the performance of the neural decoder becomes near-optimal, and is superior to that of the MWPM decoder with $10^6$ training samples at $d=5,7,9$. 

We also calculated the logical error probability in terms of the physical error probability. 
The plots in the vicinity of the threshold value are shown in Fig.\,\ref{fig:thrmlp}. 
We chose the size of the training data set as $10^6$, and calculated the performance for the distance $d=5,7,9,11$ and for the bit-flip and depolarizing noises. 
We used hyper-parameters which was the best in the calculation of Fig.\,\ref{fig:tr} when the size of the training data set is $10^6$. 
For both of the noise models, the performance is near-optimal when the distance is small. On the other hand, when the distance becomes large, the logical error probability becomes larger than that of the MWPM decoder. 
The error threshold is usually estimated with the cross point of the performance in terms of the distance. 
We see that the error threshold based on the distance is worse than that of the MWPM decoder, though the logical error probability is smaller than that of the MWPM decoder.
\begin{figure*}[tp]
	\centering
	\begin{minipage}[]{0.49\hsize}
		\includegraphics[clip, width=7.5cm]{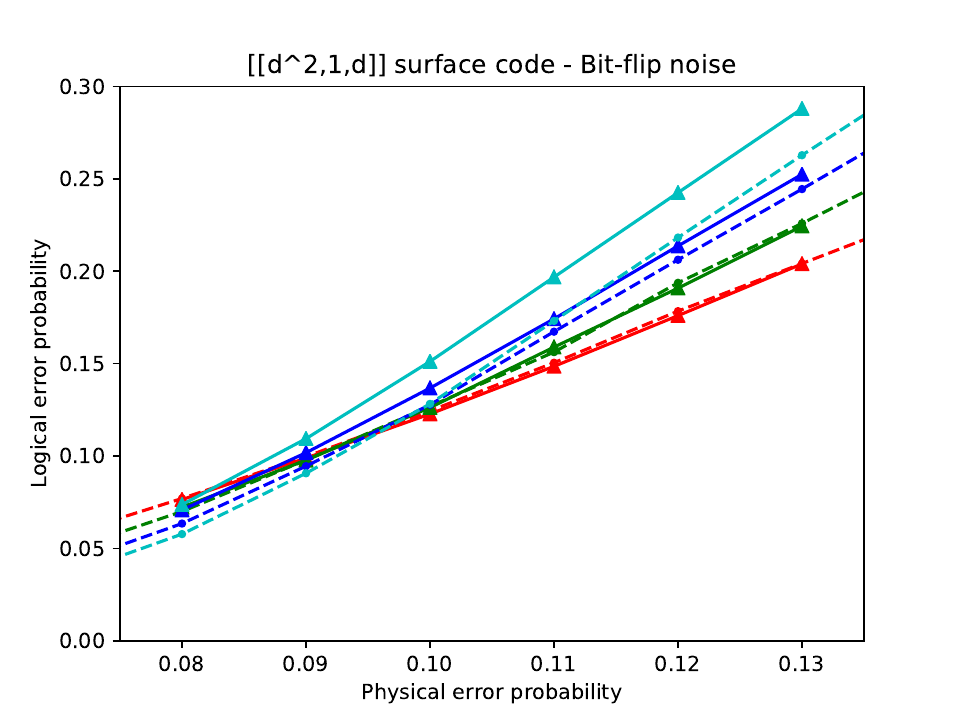}
		\subcaption{}
		\label{fig:thrmlp1}
	\end{minipage}
	\begin{minipage}[]{0.49\hsize}
		\includegraphics[clip, width=7.5cm]{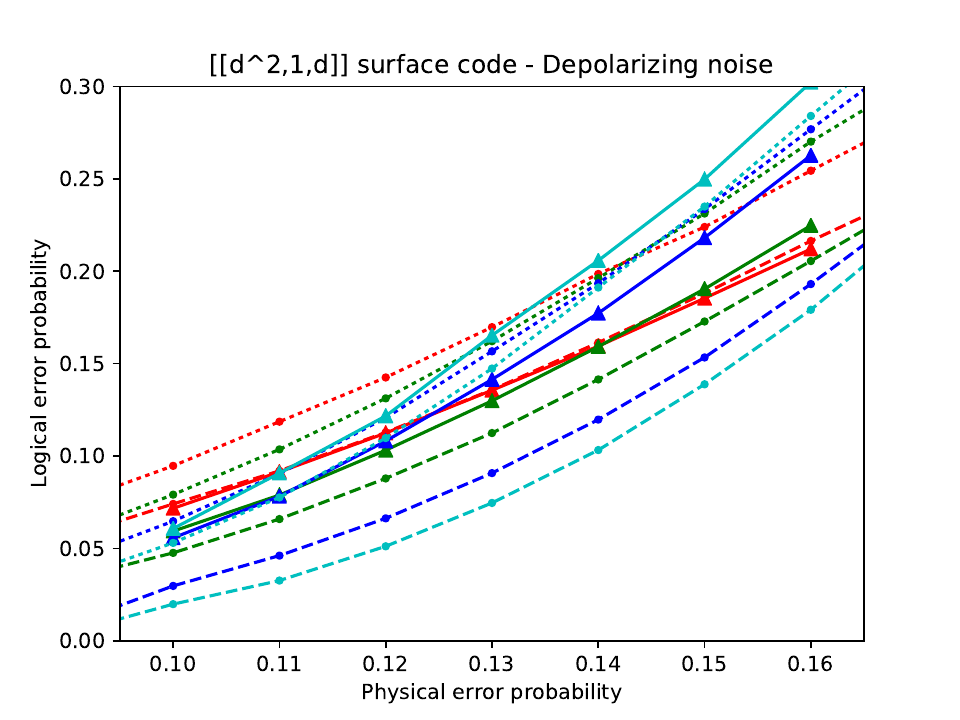}
		\subcaption{}
		\label{fig:thrmlp2}
	\end{minipage}
	\caption{The performance comparison between the neural decoder with the uniform construction (solid lines), the MD decoder (dashed lines), and the MWPM decoder (dotted lines) in the case of the [[$d^2,1,d$]] surface code. We calculated the performance for distances $d=5$ (red), $7$ (green), $9$ (blue), and $11$ (cyan) with the same $10^6$ training data set. (a) The case of the bit-flip noise. (b) The case of the depolarizing noise.}
	\label{fig:thrmlp}
\end{figure*}

The actual experiment is expected to be performed with a physical error probability sufficiently smaller than the threshold value. 
Therefore, we calculated the performance of the decoder with a small physical error probability.
The numerical results are shown in Fig.\,\ref{fig:thrmlplowp}. 
Since the training data set generated with a small value of $p$ is highly imbalanced, we trained the model with $p = 0.08$ for the bit-flip noise model, and with $p=0.11$ for the depolarizing noise model. Then, we tested the trained model with the data set generated with $p \leq 0.1$. 
We see that the logical error probability is smaller than the MWPM decoder in this region, for all the distances except $d=11$.
\begin{figure*}[tp]
	\centering
	\begin{minipage}[]{0.49\hsize}
		\includegraphics[clip, width=7.5cm]{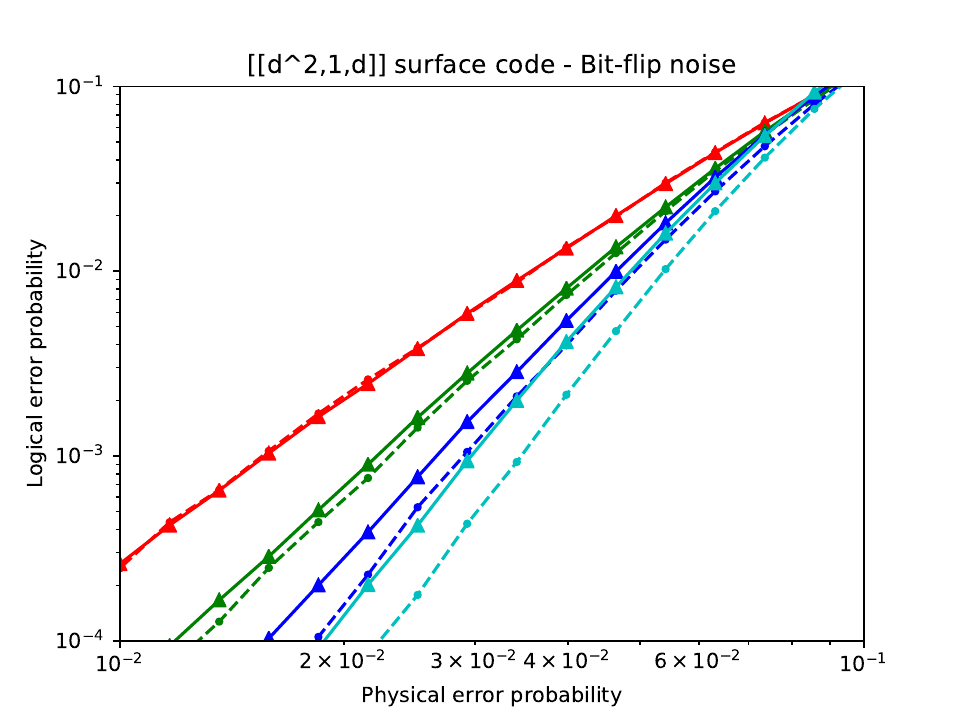}
		\subcaption{}
		\label{fig:thrmlplowp1}
	\end{minipage}
	\begin{minipage}[]{0.49\hsize}
		\includegraphics[clip, width=7.5cm]{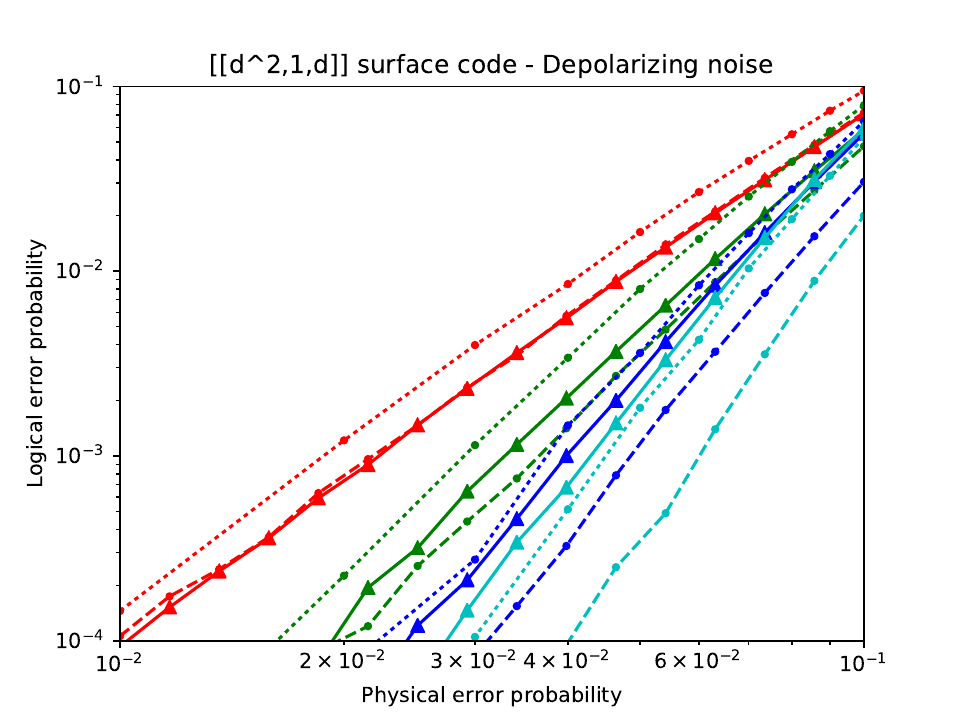}
		\subcaption{}
		\label{fig:thrmlplowp2}
	\end{minipage}
	\caption{The performance comparison between the neural decoder with the uniform construction (solid lines), the MD decoder (dashed lines), and the MWPM decoder (dotted lines) in the case of the [[$d^2,1,d$]] surface code. The neural decoder is trained with the $10^6$ training data set. We calculated the performance for distances $d=5$ (red), $7$ (green), $9$ (blue), and $11$ (cyan).  (a) The case of the bit-flip noise. The training data set is generated at the physical error probability $p=0.08$. (b) The case of the depolarizing noise. The training data set is generated at the physical error probability $p=0.11$.}
	\label{fig:thrmlplowp}
\end{figure*}

We also calculated the performance of the neural decoder for two types of color codes. 
We chose the size of training data set as $10^6$, and calculated the logical error probability for the distance $d=3,5,7,9$. 
Note that we cannot construct an efficient MD decoder in the color code even under independent bit-flip and phase-flip noise. 
The plots of the logical error probability to the physical error probability $p$ are shown in Fig.\,\ref{fig:ccplot}.
\begin{figure*}[tp]
	\centering
	\begin{minipage}[]{0.49\hsize}
		\includegraphics[clip, width=7.5cm]{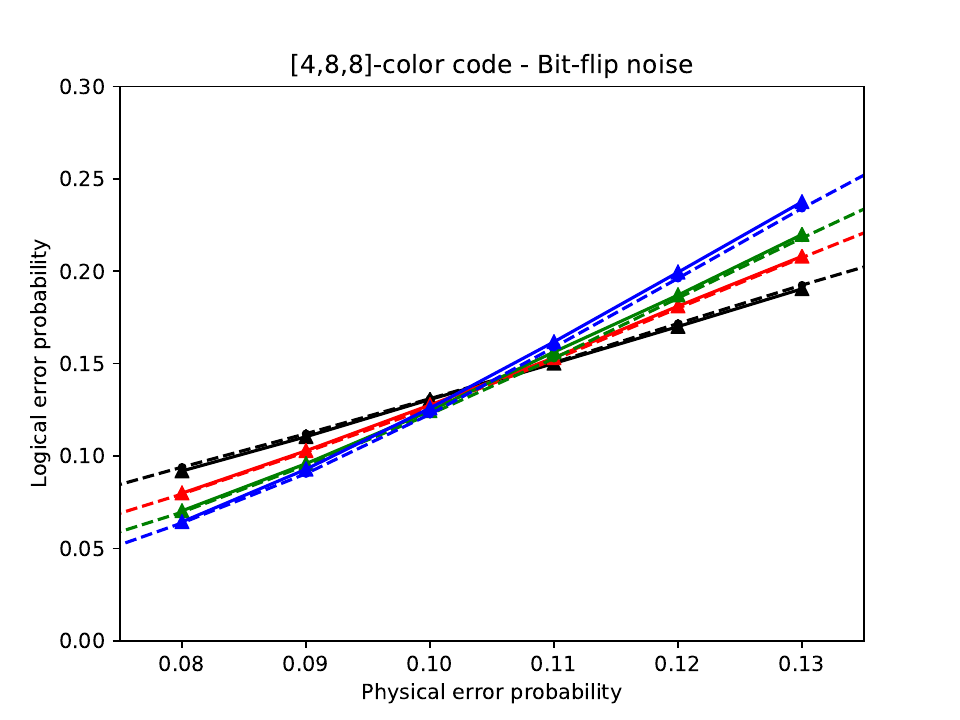}
		\subcaption{}
		\label{fig:cc1}
	\end{minipage}
	\begin{minipage}[]{0.49\hsize}
		\includegraphics[clip, width=7.5cm]{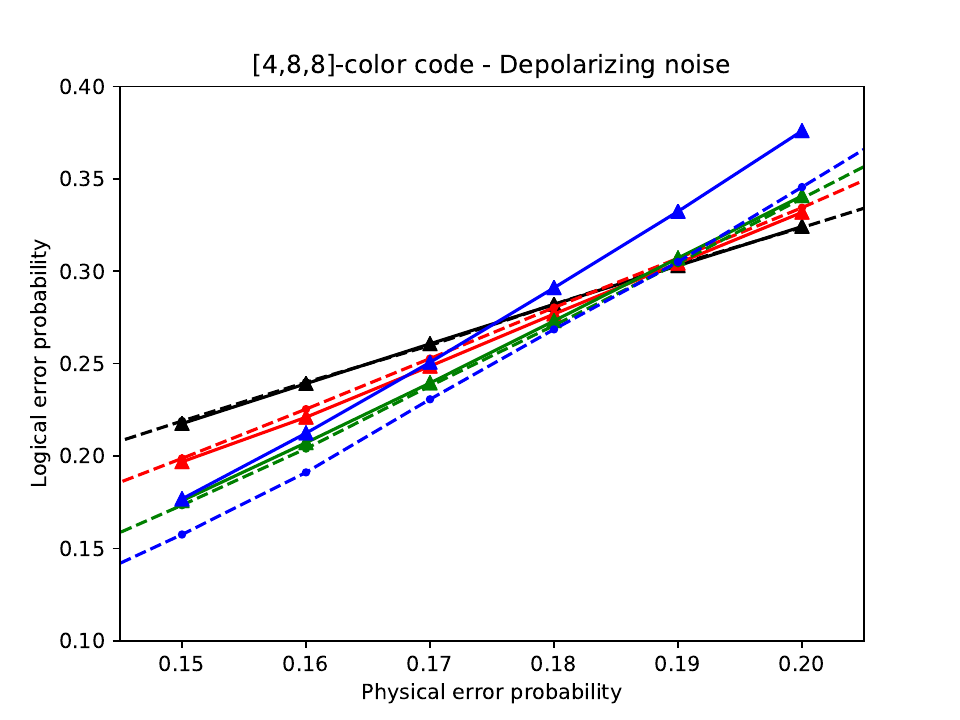}
		\subcaption{}
		\label{fig:cc2}
	\end{minipage}
	\vspace{0.01cm}
	\begin{minipage}[]{0.49\hsize}
		\includegraphics[clip, width=7.5cm]{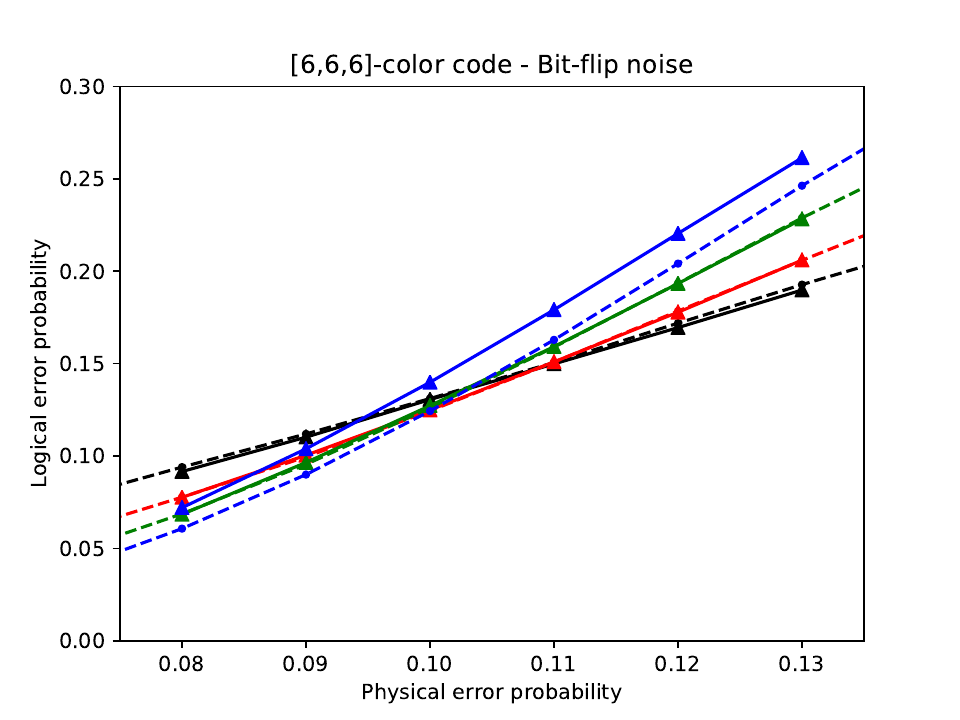}
		\subcaption{}
		\label{fig:cc3}
	\end{minipage}
	\begin{minipage}[]{0.49\hsize}
		\includegraphics[clip, width=7.5cm]{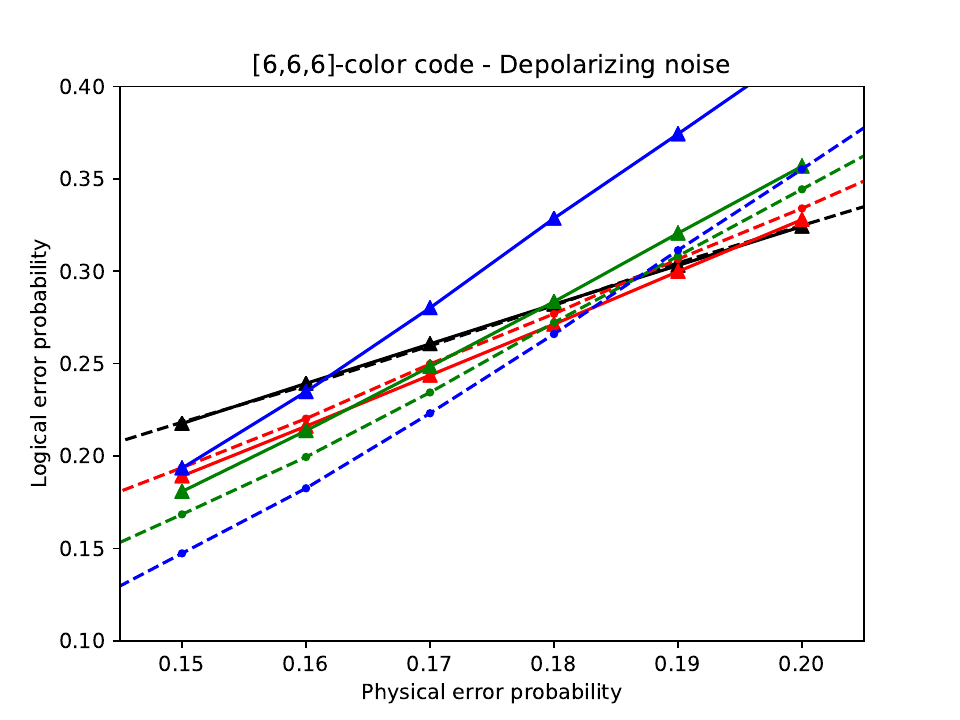}
		\subcaption{}
		\label{fig:cc4}
	\end{minipage}
	\caption{The performance comparison between the neural decoder with the uniform construction (solid lines), and the MD decoder (dashed lines) in the color codes. We calculated the performance for distances $d=3$ (black), $5$ (red), $7$ (green), and $9$ (blue)  with the $10^6$ training data set. (a) The case of the bit-flip noise in the [4,8,8]-color code. (b) The case of the depolarizing noise in the [4,8,8]-color code. (c) The case of the bit-flip noise in the [6,6,6]-color code. (d) The case of the depolarizing noise in the [6,6,6]-color code. }
	\label{fig:ccplot}
\end{figure*}
The configurations of the plots and lines are the same as that for the surface code.
In the case of the bit-flip noise, the near-optimal performance is achieved.
The performance is also near-optimal in the case of the depolarizing noise at distances except $d=9$.
We also see that the performance of the [4,8,8]-color code is better than that of [6,6,6]-color code. 
We speculate that this is because the number of the physical qubits required in the [4,8,8]-color codes is smaller than that of the [6,6,6]-color code at the same distance.
These results suggest that the neural decoder with the uniform data construction is effective also for the color codes.

\section{Utilizing spatial information}

In this section, we describe the construction of the neural network with convolutional layers.
We first discuss how the required size of the data set is expected to be suppressed if the model can utilize the spatial information of the two-dimensional quantum codes. 
Then, we introduce a construction of the neural network with convolutional layers to utilize the spatial information of the topological codes. 
We finally show the numerical results, and show that the performance of the neural decoder is improved. 

\subsection{Importance of the spatial information}
In this section, we utilize spatial information of the syndrome by using Convolutional Neural Network (CNN) as a prediction model. 
When we use MLP model, each layer is represented as a vector of neurons, and the neurons are densely connected from layer to layer. 
On the other hand, each layer of CNN is matrix-shaped, and each element of the next layer is calculated only from the local region of the previous layer using a map called a filter. 
The local filtering with the same filters can be considered as a convolution.
For the mathematical formulation, see the next subsection.

Though the CNN model is frequently used for the image recognition, this can be used for the recognition of a feature where spatial information of the feature is essential for the task. 
For example, the CNN model is expected to be effective for the classification of the phase of spin-glass \cite{carrasquilla2017machine}. 
In such a task, the local correlations of the spins are essential for prediction. Furthermore, since the patterns of local spins are translational invariant, filters which extract the feature in a local region are expected to be reused in the other regions. 
These properties match the premise of CNN, and we can expect that the performance of the CNN model is improved compared with other models such as MLP model for a fixed number of the training data set. 

In this subsection, we explain why CNN is also expected to be effective for decoding in the topological codes.
In the case of the two-dimensional topological codes, the syndrome values have natural two-dimensional arrangement. 
By carefully reshaping syndrome values as a matrix-shaped arrangement of the feature vector elements, we can explicitly let the model use local correlations of the observed syndromes using the CNN model. 
In the topological codes, a flip of a single physical qubit invokes at most a constant number (2 in the surface code, 3 in the color code) of local bit-flips in the syndrome value. 
This implies that whether two (or three) flipped syndrome bits are found in a local region or not is useful for predicting the property of the physical errors. 
For intuitive understanding, we elaborate the reason through examples. 
We consider the surface code under bit-flip errors. 
Suppose that a syndrome vector $\bm{s}$ is given in the prediction phase, and the model has encountered slightly different syndrome vectors $\bm{s}_A$ and  $\bm{s}_B$, where the difference from $\bm{s}$ are shown in Fig.\,\ref{fig:surface_diff_small}, in the training phase. 
The representation is the same as that of Fig.\,\ref{fig:sample}. 
\begin{figure}[tp]
 \centering
 \includegraphics[clip, width=7.5cm]{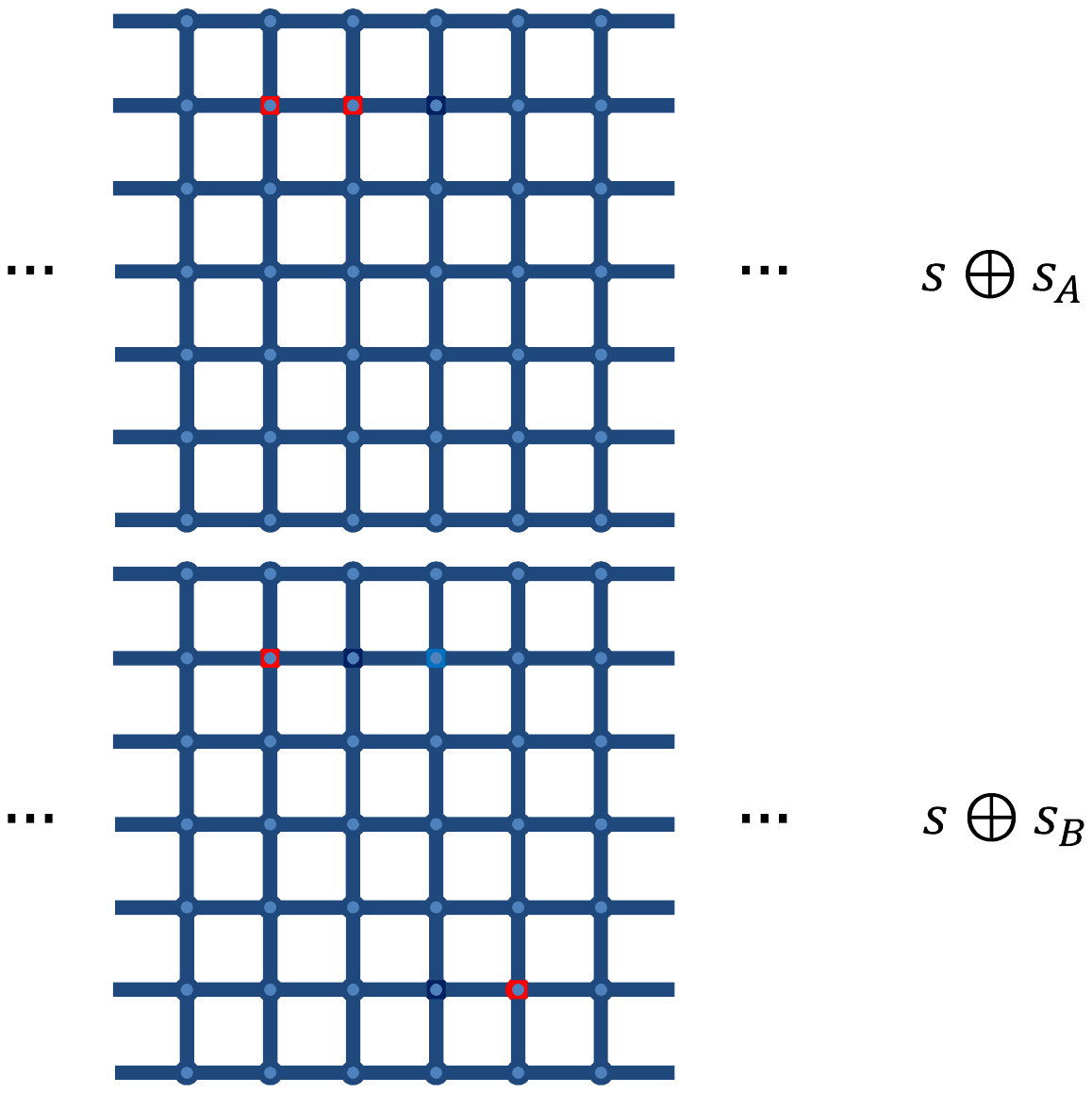}
 \caption{Example of the difference of the syndrome values. Each node in the figure corresponds to each syndrome value, and edge corresponding to the error status of each physical qubit. The color of the circle corresponds to whether the syndrome measurement detects an error.}
 \label{fig:surface_diff_small}
\end{figure}
We ignore boundary effects of the topological codes for simplicity. 
In both cases, the syndrome is two hamming distance away from the original syndrome vector $\bm{s}$, namely, $h(\bm{s}\oplus \bm{s}_A) = h(\bm{s}\oplus \bm{s}_B) = 2$.
On the other hand, there is a difference between the two syndromes in light of whether it helps the prediction of the diagnosis for the observed syndrome $\bm{s}$.
In Eq.\,(\ref{eq:usefulset}), we introduced a set of physical errors $S(\bm{s},\bm{w}^*;h)$ such that any $\bm{e} \in S(\bm{s},\bm{w}^*;h)$ with $h \sim N(H_g)^{-1}$ produces a training data useful for estimating the L2 diagnosis vector of $\bm{s}$. 
For a given $\bm{s}$ and $\bm{s}'$, if there is no vector $\bm{e}_{\delta}$ such that $H_c \Lambda \bm{e}_{\delta}^{\rm T} = \bm{s} \oplus \bm{s}'$ and $h(\bm{e}_{\delta}) \lesssim N(H_g)^{-1}$, we see no errors $\bm{e}$ with $\bm{s}(\bm{e})=\bm{s}'$ are contained in $\cup_{\bm{w} \in \{0,1\}^{2k}} S(s,\bm{w};N(H_g)^{-1})$.
In the case of $\bm{s}_A$ and $\bm{s}_B$, there is such a physical error $\bm{e}_{\delta}$ with a small hamming weight for $\bm{s}_A$, but not for $\bm{s}_B$.
Thus, if the prediction model can distinguish the samples with $\bm{s}_A$ from those with $\bm{s}_B$, it can recognize that the samples with $\bm{s}_B$ in the training data set are not relevant to the prediction for $\bm{s}$. 
The CNN model can distinguish it since it naturally utilizes the spatial information of the syndrome values. 
On the other hand, the MLP model cannot easily distinguish them since the model is not provided with the relevant spatial structure before training. 
This discussion implies the logical error probability under the fixed number of the training data set is expected to be improved with the use of a CNN model. 

\subsection{Construction of the network}
A convolutional neural network extracts patterns from image data through trainable filters that activate (produce a high value) when there are specific local patterns in the input data. 
The network usually consists of multiple convolutional layers $C^{(n)}$ each of which consists of different filtered versions of the image data $C^{(n)}_p$, indexed by a channel number $p$. The $(n-1)$-th layer with $Q$ channels are filtered to the $n$-th layer with $P$ channels with $Q\times P$ filters which we can be represented by a matrix  $f^{(n-1,n)}$.
We can describe this relation as follows. 
\begin{align}
C_{i,j,p}^{(n)}=A(\sum_{d_x}\sum_{d_y}\sum_{q}f_{d_x,d_y,q,p}^{(n-1,n)}C_{i+d_x,j+d_y,q}^{(n-1)}+b^{(n)}_p),
\end{align}
where $C_{i,j,p}^{(n)}$ is the $(i,j)$ element of the $p$-th channel in the $n$-th convolutional layer, and $f_{d_x,d_y,q,p}^{(n-1,n)}$ is the $(d_x,d_y)$ element in the $(q,p)$-th filter from the $(n-1)$-th layer to the $n$-th layer. 
Parameter $b^{(n)}_p$ is the bias added to the $p$-th channel of the $n$-th layer. A simple example is shown in Fig.\,\ref{fig:CNNlayer} where one layer has three channels and the next layer has two channels.

To use a CNN in our decoding task, 
we have to express the syndrome vector $\bm{s}$ with an appropriate  matrix representation.
We reallocate the syndrome vector for the [[$2d^2-2d+1,1,d$]] and [[$d^2,1,d$]] codes as shown in Fig.\,\ref{fig:divide}. For the [[$2d^2-2d+1,1,d$]] surface code, $\bm{s}$ is converted into two $d \times (d-1)$ matrices for the $X$ syndrome and the $Z$ syndrome. Similarly, for the [[$d^2,1,d$]] surface code, $\bm{s}$ is converted into two $(d-1) \times (d+1)/2$ matrices. We have not tried it on color codes, since it is hard to interpret the allocation of the syndromes in color codes with rectangular shapes.

Our CNN decoder consists of three convolutional layers followed by a single fully-connected hidden layer as shown in Fig.\,\ref{fig:CNNmodel}. 
At the last convolutional layer, the output channel is flattened to a single one-dimensional vector. 
Then it is used as an input for the subsequent fully-connected hidden layer. 
For each convolutional layer, the channel number is chosen to be $10 d$ for the first two layers and $5d$ for the last layer. 
Details about the model architecture is described in Appendix B. 
It is worth noting that we used the same filters for decoding of both $X$ and $Z$ flip errors, and max-pooling is not used as it is observed to reduce the performance of the decoder.

\begin{figure}[tp]
	\centering
	\includegraphics[clip, width=9.5cm]{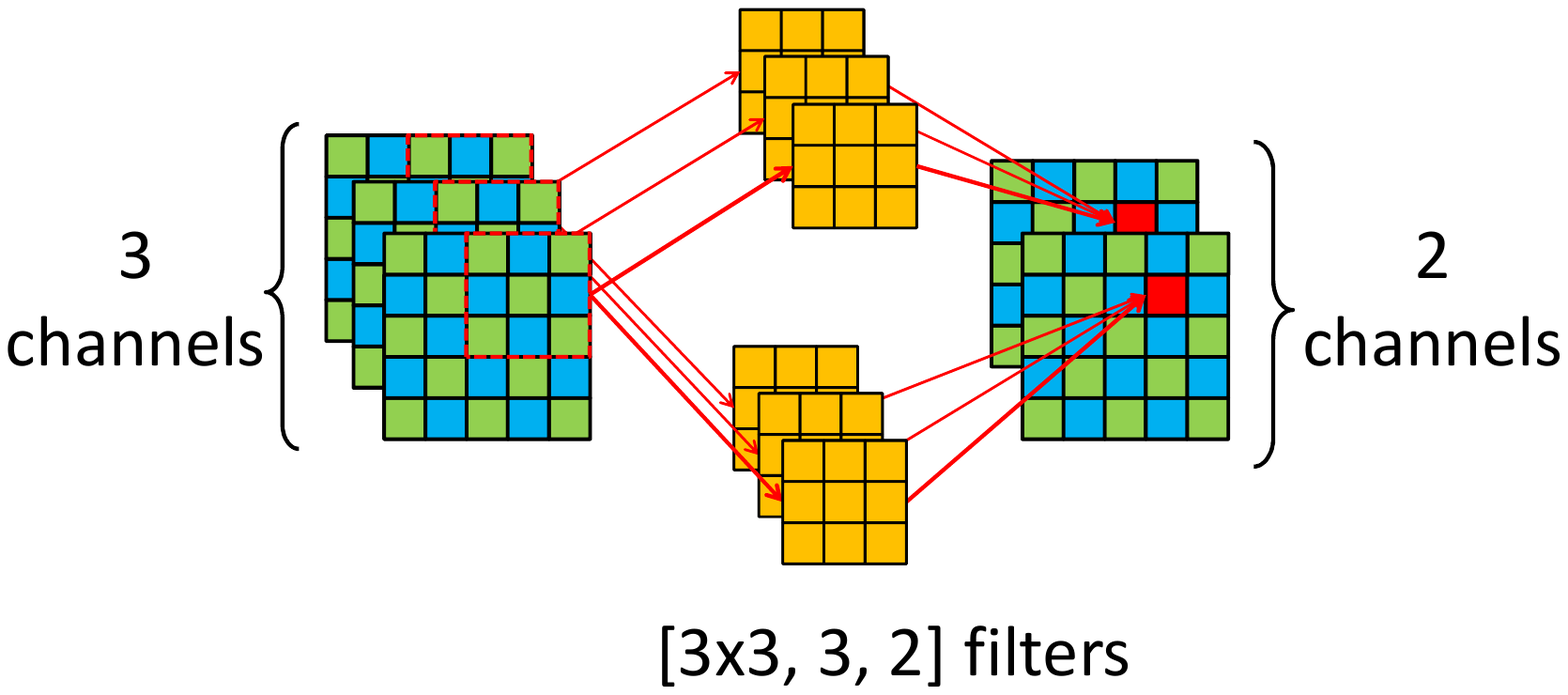}
	\caption{A simple case of convolutional layer where the input channel is three and the output channel is two.} \label{fig:CNNlayer}
\end{figure}

\begin{figure*}[tp]
	\centering
	\includegraphics[clip, width=12cm]{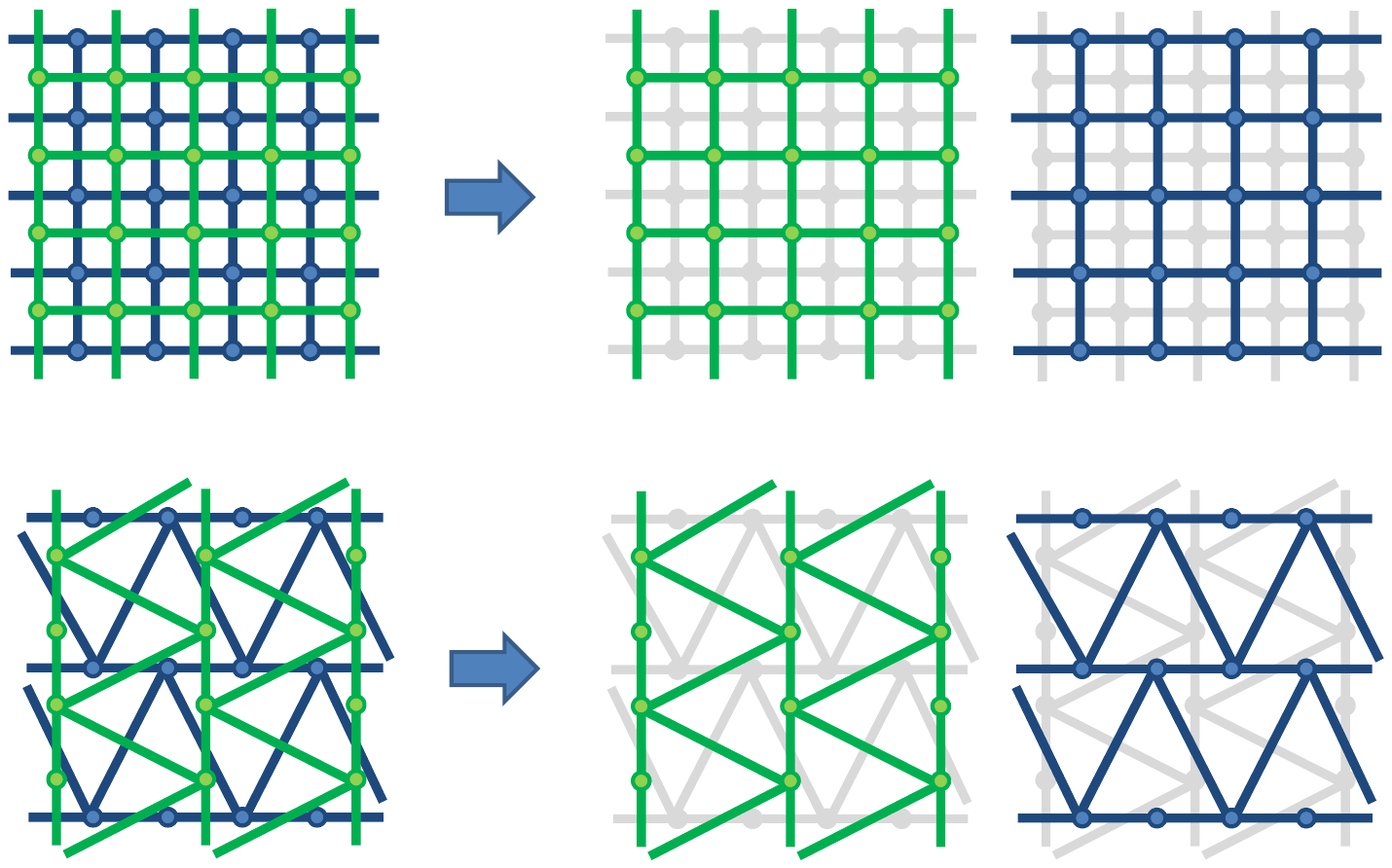}
	\caption{The above figure shows how to split and reallocate the syndrome vectors to the two input layer of the neural network. In the case of the [[$2d^2 -2d + 1, 1, d$]]-code, the lattice is split into a $(d-1) \times d$ array of syndrome, and $\pi/4$ rotated one. We input two $d \times (d-1)$ matrix as the first layer of the neural network. In the case of the [[$d^2,1,d$]]-code, we split the syndromes into two $(d-1) \times \frac{(d+1)}{2}$ arrays.  } 
	\label{fig:divide}
\end{figure*}

\begin{figure}[tp]
	\centering
	\includegraphics[clip, width=9.cm]{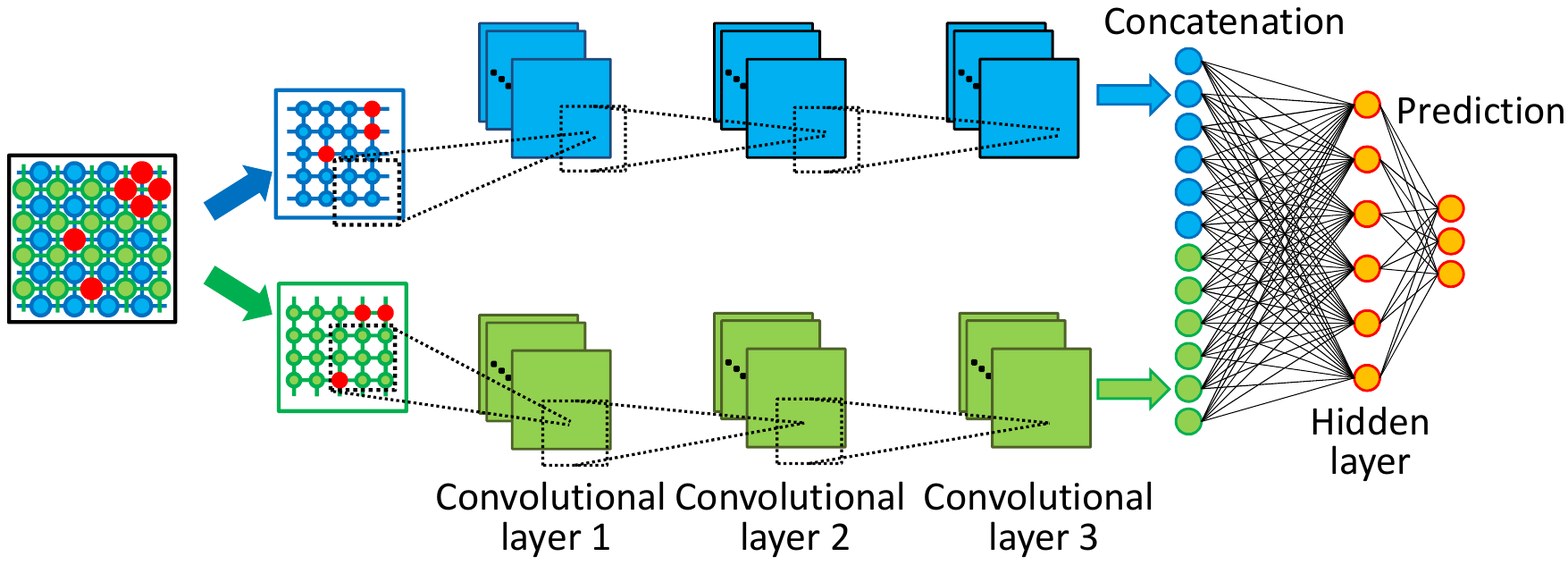}
	\caption{CNN decoder architecture used in our work. We separately pass X and Z syndrome values through the same convolutional layers, and concatenate them before feeding to the following fully-connected hidden layer.} \label{fig:CNNmodel}
\end{figure}

\subsection{Numerical result}
We call a neural decoder with the MLP model as a MLP decoder, and one with the CNN model as a CNN decoder.
We compare the performance of the CNN decoder with those of the MLP decoder, MD decoder, and MWPM decoder. 
Note that the training data set is generated with the uniform data construction. 

First, we compare the performance of the CNN decoder and that of the MLP decoder in the case of the surface codes. 
The numerical results are shown in Fig.\,\ref{fig:mlp_cnn}. 
\begin{figure*}[tp]
	\centering
	\begin{minipage}[]{0.49\hsize}
		\includegraphics[clip, width=8.0cm]{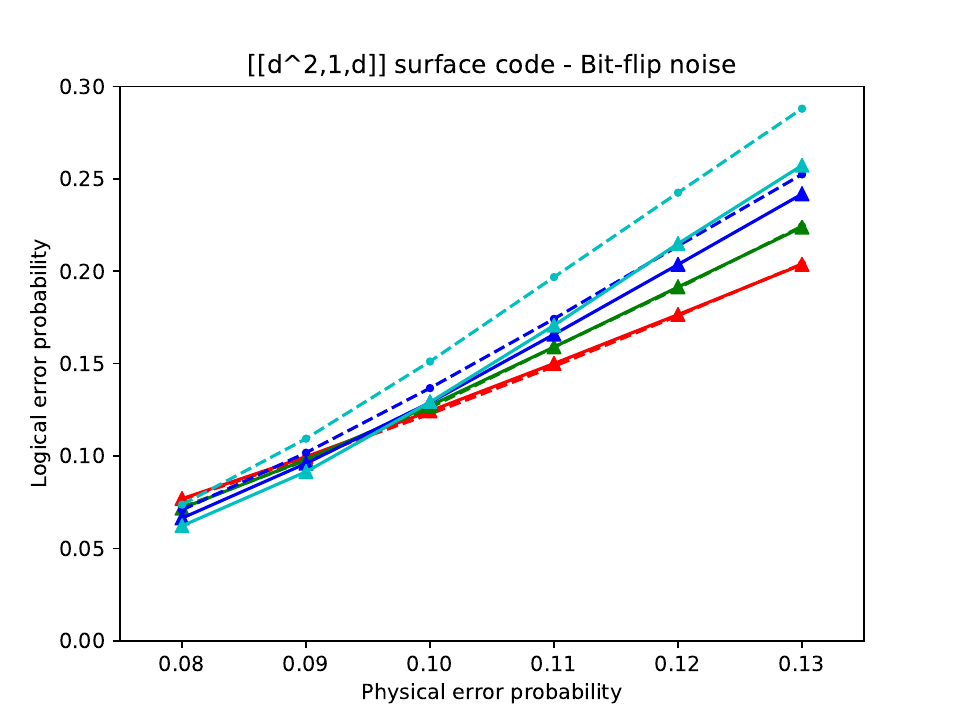}
   \subcaption{}
	\end{minipage}
	\begin{minipage}[]{0.49\hsize}
		\includegraphics[clip, width=8.0cm]{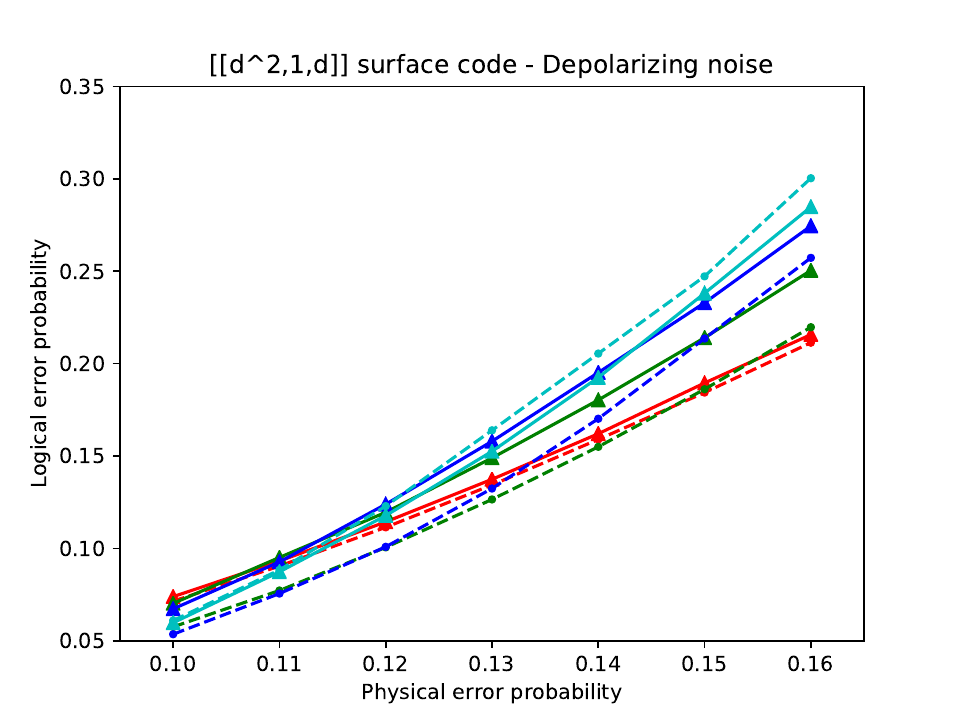}
   \subcaption{}
	\end{minipage}
	\vspace{1cm}
	\begin{minipage}[]{0.49\hsize}
		\includegraphics[clip, width=8.0cm]{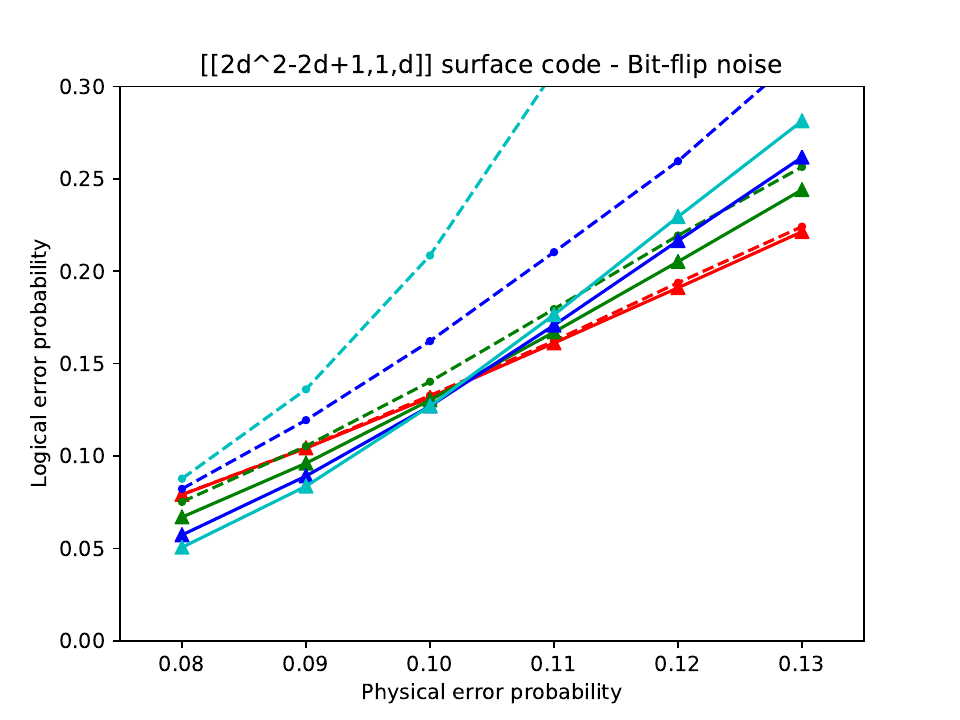}
   \subcaption{}
	\end{minipage}
	\begin{minipage}[]{0.49\hsize}
		\includegraphics[clip, width=8.0cm]{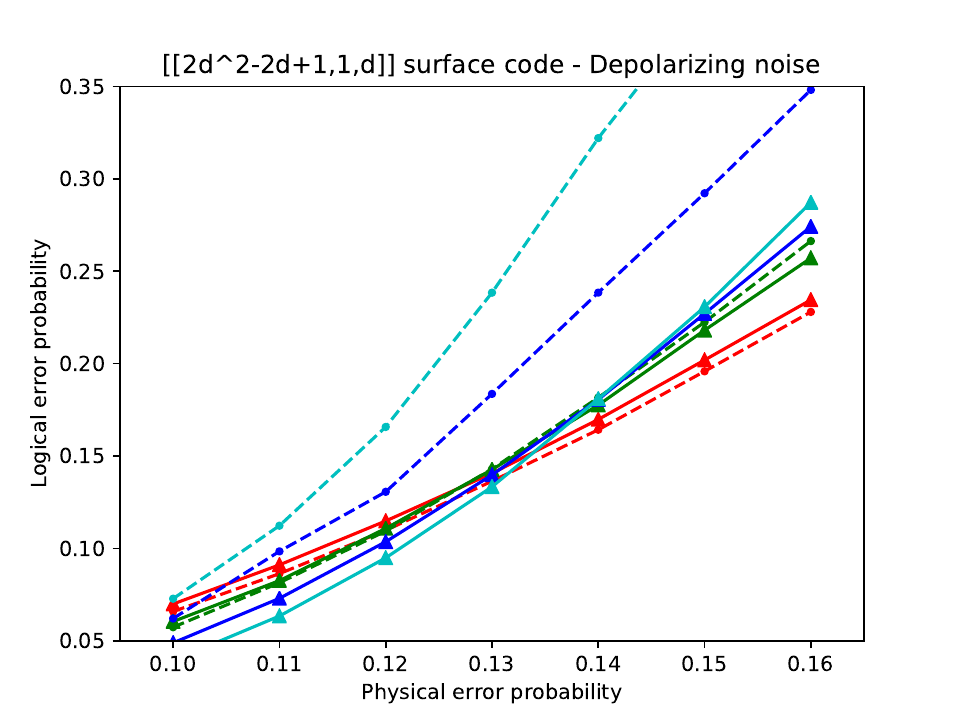}
   \subcaption{}
	\end{minipage}
	\caption{ The performance comparison between the CNN decoder (solid lines) and the MLP decoder (dashed lines) in the case of the surface code. We calculated the performance for distance $d=5$ (red), $7$ (green), $9$ (blue), and $11$ (cyan). (a) The bit-flip noise in the [[$d^2,1,d$]] code. (b) the depolarizing noise in [[$d^2,1,d$]] code. (c) the bit-flip noise in [[$2d^2-2d+1,1,d$] code. (d) the depolarizing noise in [[$2d^2-2d+1,1,d$]] code.  }
	\label{fig:mlp_cnn}
\end{figure*}
In this figure, the solid lines and dashed lines are the logical error probability for the CNN decoder and the MLP decoder, respectively. 
The colors red, green, blue, and cyan correspond to distances $d=5$, $7$, $9$ and $11$, respectively.
For both types of the surface codes, the CNN decoder shows superior performances to that of the MLP decoder at large distances. 
In particular, in the case of the [[$2d^2-2d+1,1,d$]] surface code, the CNN decoder shows significant improvement of a logical error probability. 
We see that the CNN model is effective for improving the performance of the neural decoder at large distances.
On the other hand, we see that the CNN decoder shows inferior performances to the MLP decoder at a small distance. 
We speculate the reason of this as follows.
The CNN model assumes that the local features can be extracted by using the same filter everywhere. 
Such an assumption is not necessarily true when the distance is small, since almost all the filtered local regions, of size $3\times 3$ for example, are on or near to the boundary of the two-dimensional codes.
Note that we tried to avoid this problem by padding the boundaries with various values, such as $0.5$ or $-1$, but the performance in the small distance did not improve. 

Next, we compared the performance of the CNN decoder with those of the MD decoder and the MWPM decoder. 
The results are shown in Fig.\,\ref{fig:res_cnn}.
The solid lines, the dashed lines, and the dotted lines are the logical error probability for the CNN decoder, the MD decoder, and the MWPM decoder, respectively. 
The colors red, green, blue, and cyan correspond to distances $d=5$, $7$, $9$ and $11$, respectively. 
In the case of the bit-flip noise, we see that the logical error probabilities of the CNN decoder is equal to or slightly better than that of the MD decoder.
In the case of the depolarizing noise, though there are gaps between the performances of the CNN decoder and the MD decoder, the performance of the CNN decoder is superior or comparable to that of the MWPM decoder even at the distance $d=11$. 

\begin{figure*}[tp]
	\centering
  \begin{minipage}[]{0.49\hsize}
	\includegraphics[clip, width=8.0cm]{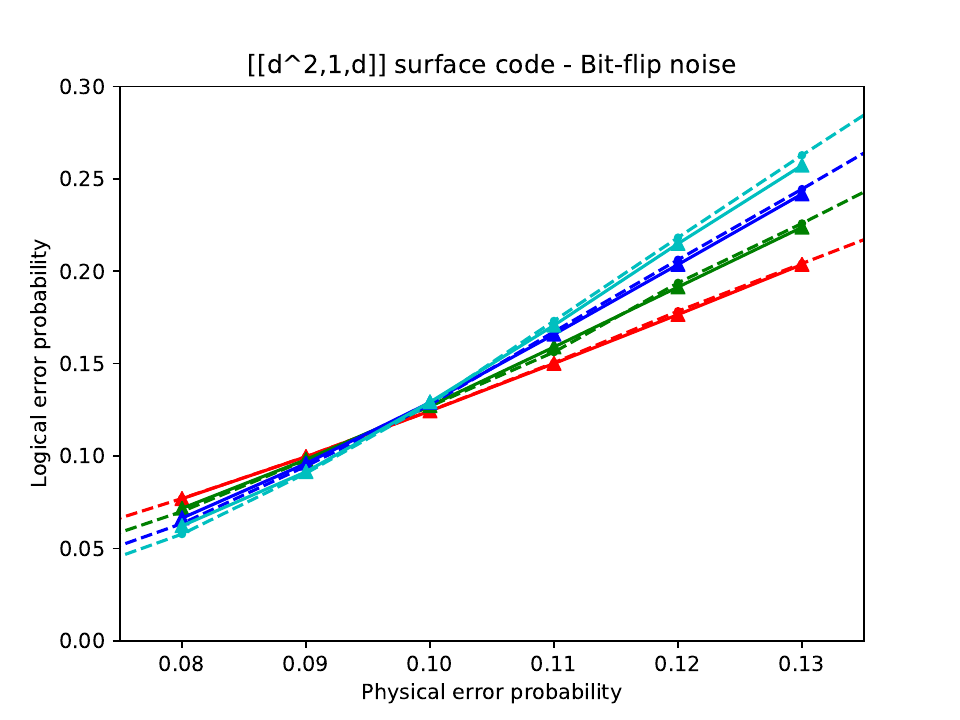}
   \subcaption{}
  \end{minipage}
  \begin{minipage}[]{0.49\hsize}
	\includegraphics[clip, width=8.0cm]{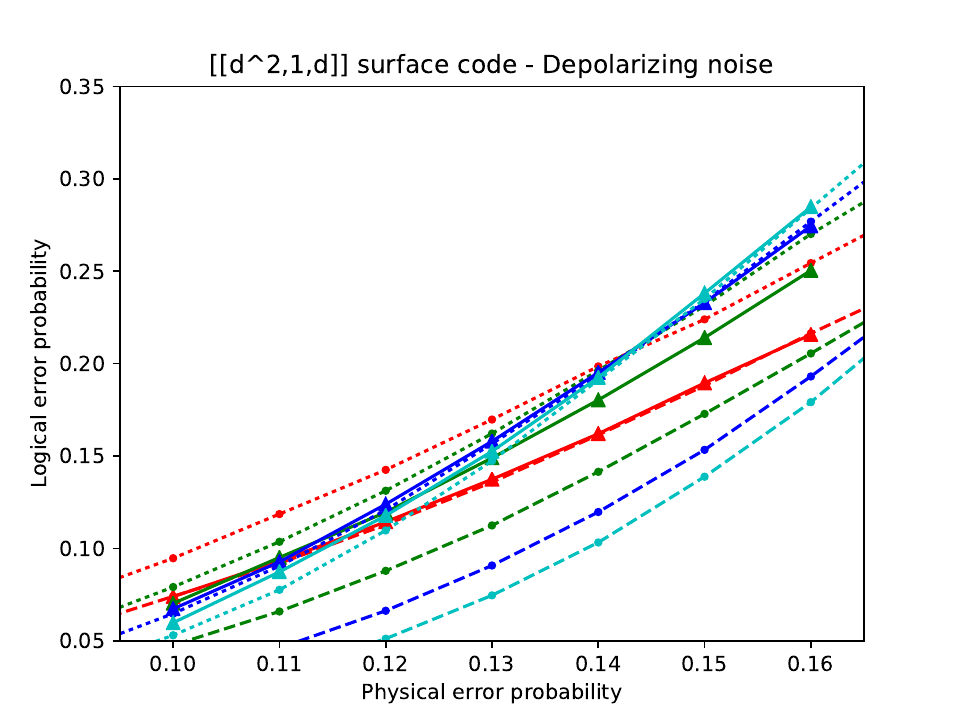}
   \subcaption{}
  \end{minipage}
	\vspace{1cm}
  \begin{minipage}[]{0.49\hsize}
	\includegraphics[clip, width=8.0cm]{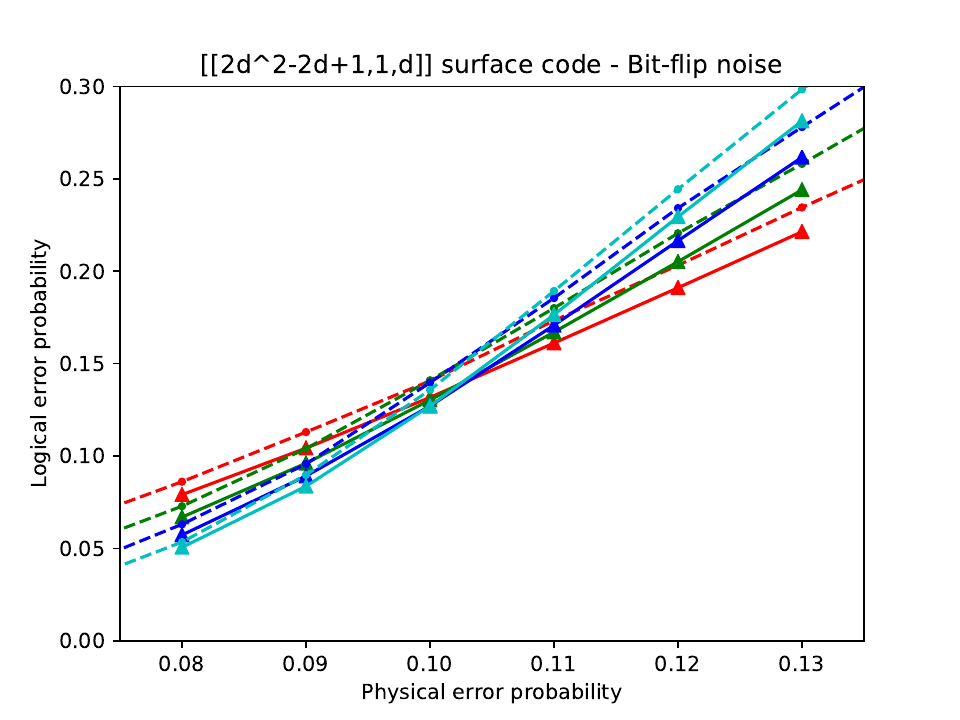}
   \subcaption{}
  \end{minipage}
  \begin{minipage}[]{0.49\hsize}
	\includegraphics[clip, width=8.0cm]{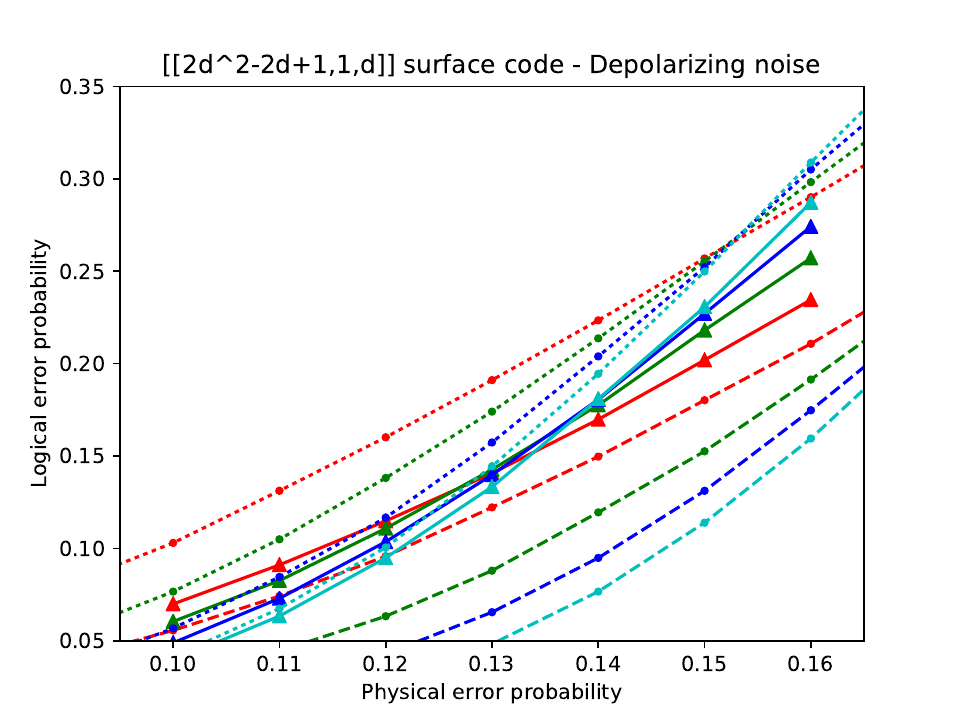}
   \subcaption{}
  \end{minipage}
	\caption{ The performance comparison between the CNN decoder (solid lines), the MD decoder (dashed lines), and the MWPM decoder (dotted lines) in the surface codes. We calculated the performance for distance $d=5$ (red), $7$ (green), $9$ (blue), and $11$ (cyan). (a) the bit-flip noise in the [[$d^2,1,d$]] code. (b) the depolarizing noise in [[$d^2,1,d$]] code. (c) the bit-flip noise in [[$2d^2-2d+1,1,d$]] code. (d) the depolarizing noise in [[$2d^2-2d+1,1,d$]] code.}
	\label{fig:res_cnn}
\end{figure*}

We also calculated the logical error probability of the CNN decoder at a small physical error probability $p$ in the case of the $[[2d^2-2d+1,1,d]]$ surface code.
We trained the CNN decoder at $p=0.08$ for the bit-flip noise model, and at $p=0.11$ for the depolarizing noise model. Then, the decoder is tested with the data set generated with small physical error probabilities.
The plots are shown in Fig.\,\ref{fig:cnn_lowp}. 
\begin{figure*}[tp]
	\centering
	\begin{minipage}[]{0.49\hsize}
		\includegraphics[clip, width=8.0cm]{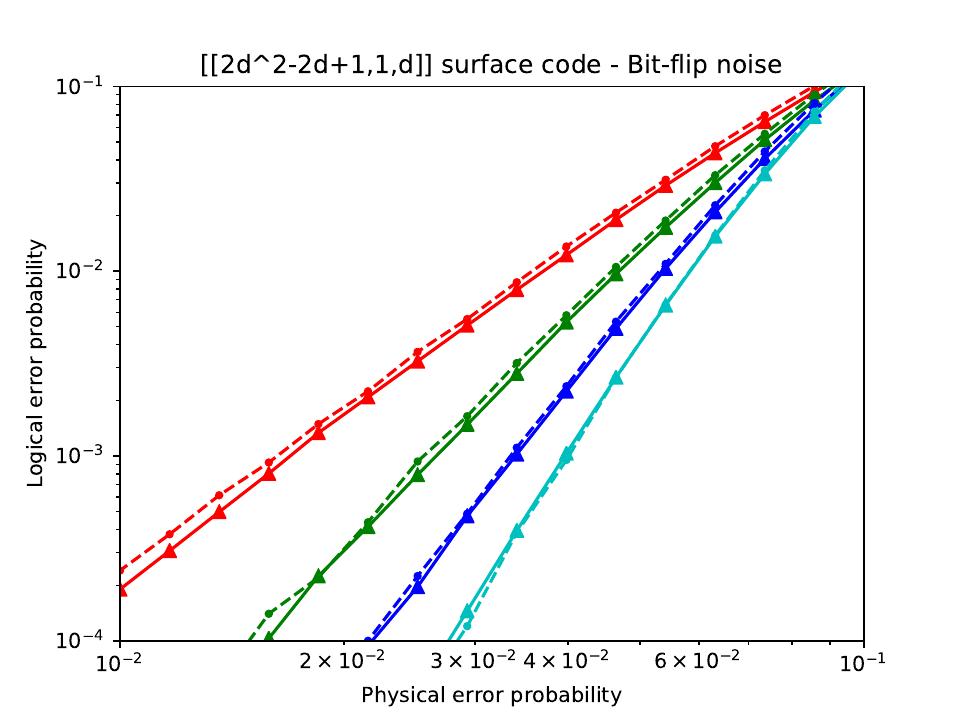}
   \subcaption{}
	\end{minipage}
	\begin{minipage}[]{0.49\hsize}
		\includegraphics[clip, width=8.0cm]{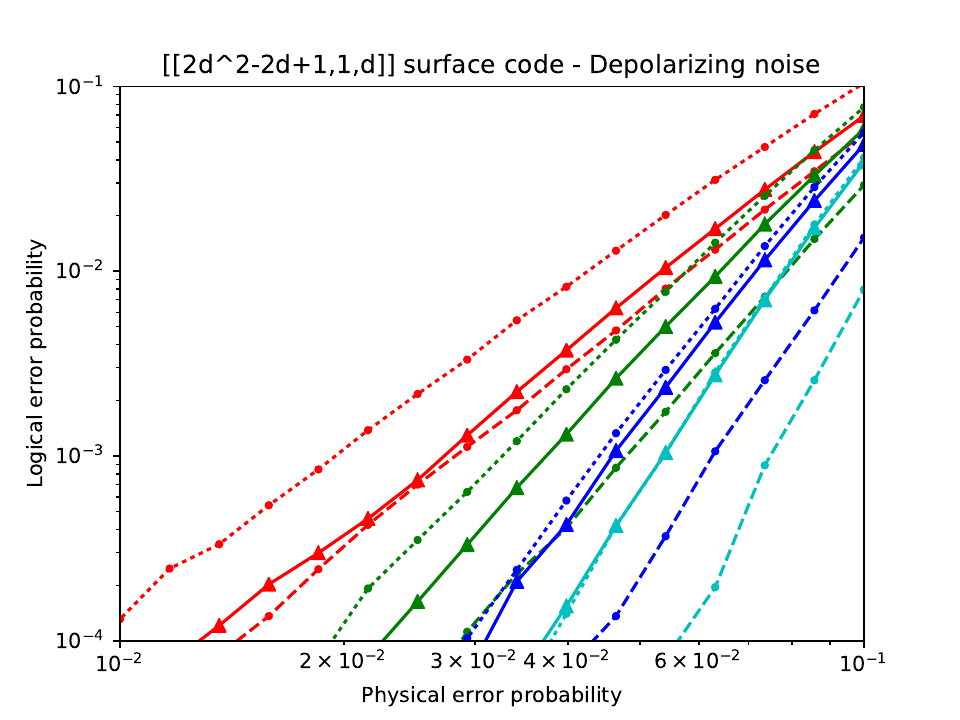}
   \subcaption{}
	\end{minipage}
	\caption{ The performance comparison between the CNN decoder (solid lines) and the MD decoder (dashed lines), and the MWPM decoder (dotted lines) in the case of the [[$2d^2-2d+1,1,d$]] surface code, where the decoders are trained with the training data set generated at the fixed error rate. We calculated the performance for distance $d=5$ (red), $7$ (green), $9$ (blue), and $11$ (cyan). (a) The case of the bit-flip noise. The training data set is generated at the physical error probability $p=0.08$. (b) The case of the depolarizing noise. The training data set is generated at the physical error probability $p=0.11$.}
	\label{fig:cnn_lowp}
\end{figure*}
In the case of the bit-flip noise, the CNN decoder achieves the performance close to the MD decoder also at small physical error probabilities.
In the case of the depolarizing noise, the performance of the neural decoder with CNN decoder is superior to that of the MWPM decoder at $d=9$, and comparable at $d=11$.
We can say that the CNN model is effective also for a use of neural decoders at small physical error probabilities.

\section{Conclusion}
In this paper, we theoretically analyzed mechanism of machine-learning-based decoders for QEC, and proposed a general direction to construct the data set and the neural network. Then, we have numerically shown that our direction is effective compared with the existing works.

Since the formalism of the machine learning is flexible, there are many possible ways to reduce the decoding problem in QEC to the task of the machine learning. 
In order to clarify what is the best way of reduction, we introduced the linear prediction framework. 
This framework essentially includes the existing methods as specific cases, and enables us to discuss conditions for satisfying natural requirements for a good decoder for QEC. 
In particular, we have derived the condition to perform the optimal decoding in the limit of a large training data size. 
We also introduced a measure, normalized sensitivity, which represents a properly-scaled bound on the deviation in the prediction target resulting from a small change in the physical error pattern. 
We proposed to use this measure as a criterion for constructing a better decoder.
We then proposed a general direction for constructing the data set, uniform data construction, which can be applicable to general topological codes.
We numerically confirmed that the performance of the neural decoder is improved with the uniform data construction. 
Our decoder was found to be superior to known efficient decoders, such as neural decoders proposed in the existing methods and the decoder based on the reduction to minimum-weight perfect matching. 
We also confirmed that the performance of our neural decoder is near-optimal in various situations by comparing it with the minimum-distance decoder, which is known to be near-optimal but not efficient in general. 
We also confirmed that the neural decoder can achieve near-optimal performance not only for surface codes but also for color codes.

Another important factor of the neural decoder is construction of the neural network. 
We discussed the importance of the spatial information of the syndrome measurement in order to let the prediction model recognize useful samples from a given training data set. 
To utilize the spatial information, we proposed a neural decoder with the convolutional neural network. We numerically observed that the performance of the neural decoder is further improved with this network construction in the surface code.
In particular, we showed that the proposed neural decoder achieves a smaller logical error probability than that of the decoder based on minimum-weight perfect matching even at distance $d=11$ with a training data set size of $10^6$. 

Since using machine learning for QEC is an emergent field, there are still many possible extensions and directions of the neural decoders. 
As we detailed in Appendix B, the prediction time of the neural decoders is smaller than that of the MD decoder, but larger than that of the MWPM decoder in our desktop PC. Since the prediction of the neural decoders can be done with simple matrix multiplications, the time for prediction can be further made short by using an optimized hardware such as field-programmable gate array (FPGA), which is popularly used in experiments.
While we have discussed only a label linearly generated in GF(2), the performance may be more improved by allowing labels nonlinearly generated from the physical error. 
For example, the relation between the syndrome values and the weight of the physical error, which cannot be generated linearly in GF(2), can be trained and predicted independently with a neural network. 
Then, the recovery map can be predicted with the syndrome values and the predicted weight with another neural network.
The linear prediction framework also limits the sample in the training data set to that is sampled from the assumed physical error distribution. 
However, the distribution which is the best for the training is not necessarily the same as the actual distribution. For example, we saw that the prediction model trained at the physical error probability around the threshold value shows high-performance also at low physical error probabilities.
There can be a more artificial way to construct the training data set to achieve the performance with a smaller size of the training data set. 
In the numerical investigation, we observed that the required amount of the data set becomes exponentially large in terms of the distance. 
This may be suppressed by renormalizing the matrix representation of the syndrome with trained filters as done in the renormalization group decoder \cite{duclos2010fast}. 
We expect that CNN is also applicable to the color codes by using non-rectangle filters.
When the stabilizer measurements themselves suffer from noise, stabilizer measurements are often repetitively performed during QEC. In such a case, the length of the syndrome data is not fixed. In our construction, we need to train the neural network again whenever the length of the syndrome data changes. 
The studies of Refs.\,\cite{baireuther2018machine,breuckmann2018scalable} focused on removing this drawback by utilizing recurrent neural network and convolutional neural network. 
Using the technique proposed in Refs.\,\cite{baireuther2018machine,breuckmann2018scalable}, our neural decoder may be also applicable to the cases when we perform repetitive stabilizer measurements.

\section*{Acknowledgements}
This work is supported by KAKENHI Grant No. 16H02211; PRESTO, JST, No. JPMJPR1668; CREST, JST, Grants No. JPMJCR1671, No. JPMJCR1673; ERATO, JST, Grant No. JPMJER1601; and Photon Frontier Network Program, MEXT.
Y.S. is supported by Advanced Leading Graduate Course for Photon Science.
D.A. and Y.S. contributed equally to this work. 
Y.S. contributes to the construction of the data set in Sec.\,III, and A.D. to the construction of the network in Sec.\,IV. K.F. and M.K. motivate and supervise the idea and discussion of this paper.

\section*{Appendix A: Proof of the lemmas}

\subsection*{Proof of the converse part in Lemma \ref{lemma_logical_decoder}}
Here we prove the last statement of Lemma \ref{lemma_logical_decoder}. 
When Eq.\,(\ref{eq:faithfuldef}) does not hold, either (i) there exists $\bm{e}_1$ such that 
\begin{eqnarray}
\bm{e}_1 &\notin& \mathcal{L}_0, \\
H_{cg} \Lambda \bm{e}_1^{\rm T} &=& 0,
\end{eqnarray}
or (ii) there exists $\bm{e}_1$ such that 
\begin{eqnarray}
\bm{e}_1 &\in& \mathcal{L}_0, \\
H_{cg} \Lambda \bm{e}_1^{\rm T} &\neq& 0.
\end{eqnarray}

For (i), consider two probability distributions $\{p_{\bm{e}}\}$ and $\{p'_{\bm{e}}\}$ such that 
\begin{eqnarray}
{\rm Pr}_{\bm{e} \sim \{p_{\bm{e}}\}  }  \left[ \bm{e} = 0 | \bm{s}(\bm{e}) = 0 \right] &=& 0.75 ,\\
{\rm Pr}_{\bm{e} \sim \{p_{\bm{e}}\}  }  \left[ \bm{e} = \bm{e}_1 | \bm{s}(\bm{e}) = 0 \right] &=& 0.25 ,
\end{eqnarray}
and
\begin{eqnarray}
{\rm Pr}_{\bm{e} \sim \{p'_{\bm{e}}\}  }  \left[ \bm{e} = 0 | \bm{s}(\bm{e}) = 0 \right] &=& 0.25 ,\\
{\rm Pr}_{\bm{e} \sim \{p'_{\bm{e}}\}  }  \left[ \bm{e} = \bm{e}_1 | \bm{s}(\bm{e}) = 0 \right] &=& 0.75.
\end{eqnarray}
An optimal decoder for each case succeeds with probability $0.75$ given $\bm{s}=0$.
On the other hand, since $\bm{g}^{(\delta)}(0) =0$ in both cases, only the value of $\bm{r}^*(0,0)$ is relevant.
Since $\bm{w}(0) \neq \bm{w}(\bm{e}_1)$, any choice of $\bm{r}^*(0,0)$ leads to a success probability no greater than 0.25 for at least one of the cases.

For (ii), choose $\bm{w} \neq 0$, and if $H_g \Lambda (\bm{w}G)^{\rm T} \neq 0$, define 
\begin{eqnarray}
\bm{e}_2:= \bm{w}G.
\end{eqnarray}
Otherwise, define 
\begin{eqnarray}
\bm{e}_2:= \bm{e}_1 \oplus \bm{w}G.
\end{eqnarray}
It ensures that $\bm{s}(\bm{e}_2)=0$ and $\bm{g}_2:=H_g \Lambda \bm{e}_2 \neq 0$.
Consider two probability distributions $\{p_{\bm{e}}\}$ and $\{p'_{\bm{e}}\}$ such that 
\begin{eqnarray}
{\rm Pr}_{\bm{e} \sim \{p_{\bm{e}}\}  }  \left[ \bm{e} = 0 | \bm{s}(\bm{e}) = 0 \right] &=& 0.4 ,\\
{\rm Pr}_{\bm{e} \sim \{p_{\bm{e}}\}  }  \left[ \bm{e} = \bm{e}_1 | \bm{s}(\bm{e}) = 0 \right] &=& 0.0 , \\
{\rm Pr}_{\bm{e} \sim \{p_{\bm{e}}\}  }  \left[ \bm{e} = \bm{e}_2 | \bm{s}(\bm{e}) = 0 \right] &=& 0.6 ,
\end{eqnarray}
and
\begin{eqnarray}
{\rm Pr}_{\bm{e} \sim \{p'_{\bm{e}}\}  }  \left[ \bm{e} = 0 | \bm{s}(\bm{e}) = 0 \right] &=& 0.3 ,\\
{\rm Pr}_{\bm{e} \sim \{p'_{\bm{e}}\}  }  \left[ \bm{e} = \bm{e}_1 | \bm{s}(\bm{e}) = 0 \right] &=& 0.3 , \\
{\rm Pr}_{\bm{e} \sim \{p'_{\bm{e}}\}  }  \left[ \bm{e} = \bm{e}_2 | \bm{s}(\bm{e}) = 0 \right] &=& 0.4.
\end{eqnarray}
An optimal decoder for each case succeeds with probability $0.6$ given $\bm{s}=0$. On the other hand, since $\bm{g}^{(\delta)}(0)=\bm{g}_2$ in both cases, only the value of $\bm{r}^*(\bm{g}_2,0)$ is relevant. Since $\bm{w}(0) \neq \bm{w}(\bm{e}_2)$, any choice of $\bm{r}^*(\bm{g}_2,0)$ leads to a success probability no greater than $0.4$ for at least one of the cases.

\subsection*{Proof of the converse part in Lemma\,\ref{lemma:decomposable}}
When the diagnosis matrix is not decomposable, 
there exists a non-empty subset $\mathcal{W} \subset \{0,1\}^{2k}$ such that 
\begin{eqnarray}
\label{eq:split}
	\sum_{\bm{w} \in \mathcal{W}} \alpha_{\bm{w}}\bm{g}_{0}(\bm{w}) = \sum_{\bm{w} \in  \{0,1\}^{2k} \setminus \mathcal{W} } \beta_{\bm{w}} \bm{g}_{0}(\bm{w}),
\end{eqnarray}
where $\alpha_{\bm{w}}, \beta_{\bm{w}} \geq 0$ and 
\begin{eqnarray}
\label{eq:partition}
\Gamma:= \sum_{\bm{w} \in \mathcal{W}} \alpha_{\bm{w}} = \sum_{\bm{w} \in \{0,1\}^{2k} \setminus \mathcal{W}} \beta_{\bm{w}} > 0.
\end{eqnarray}
Consider two probability distributions $\{p_{\bm{e}}\}$ and $\{p'_{\bm{e}}\}$ such that 
\begin{widetext}	
\begin{eqnarray}
	{\rm Pr}_{\bm{e} \sim \{p_{\bm{e}}\}_A} \left[ \bm{w}(\bm{e}) = \bm{w}, \bm{l}(\bm{e}) = \bm{l} | \bm{s}(\bm{e}) = 0 \right] &=& \begin{cases} \alpha_{\bm{w}}/\Gamma & \bm{w} \in \mathcal{W}, \bm{l}=0 \\ 0  & \text{otherwise} \end{cases}, \\
	{\rm Pr}_{\bm{e} \sim \{p'_{\bm{e}}\}_B} \left[ \bm{w}(\bm{e}) = \bm{w}, \bm{l}(\bm{e}) = \bm{l} | \bm{s}(\bm{e}) = 0 \right] &=& \begin{cases} \beta_{\bm{w}}/\Gamma  & \bm{w} \notin \mathcal{W}, \bm{l}=0 \\ 0 & \text{otherwise} \end{cases}.
\end{eqnarray}
\end{widetext}

From Eq.\,(\ref{eq:split}), the L2 diagnosis vector $\bm{g}^{(\rm L2)}(0)$ is identical for the two distributions. On the other hand, 
the most probable class $\bm{w}$ is different for the two probability distributions.
This means that a single decoder cannot perform the optimal decoding for both of the two distributions.

\subsection*{Proof of the existence of faithful and decomposable diagnosis matrices}
In the main text, we showed that a diagnosis matrix should be faithful and decomposable for performing optimal decoding in the ideal limit of the training process, and showed an example for the case $k=1$.
On the other hand, it is not trivial that there exists a faithful and decomposable construction of a diagnosis matrix for an arbitrary stabilizer code. 
We show that diagnosis matrix $H_g = WG$, where $W$ is a $2^{2k} \times 2k$ binary matrix of which the $i$-th row is a $2k$-bit binary representation of an integer $i$, is always faithful and decomposable for an arbitrary stabilizer code and for an arbitrary number of logical qubits $k$.
Since row vectors of $H_g$ contains all the logical operators, it is trivial that ${\rm span}(\{(H_{cg})_i\})$ is equivalent to the logical space $\mathcal{L}$, and $H_g$ is faithful.
The condition for decomposability is equivalent to the condition that $\{\bm{g}(\bm{w}) | \bm{w} \in \{0,1\}^{2k}\}$, where $\bm{g}(\bm{w}) := H_g \Lambda (\bm{w}G)^{\rm T}$, is affinely independent in real vector space. 
To show the latter, we first prove that for any pair of binary vectors $\bm{w}, \bm{w'} \in \{0,1\}^{2k}$ such that $\bm{w} \neq \bm{w}'$, the weight of $\bm{g}(\bm{w}) \oplus \bm{g}(\bm{w}')$ is $2^{2k-1}$.
A $2^k$-bit sequence $\bm{g}(\bm{w}) \oplus \bm{g}(\bm{w}')$ is given by
\begin{eqnarray}
\bm{g}(\bm{w}) \oplus \bm{g}(\bm{w}') = W G\Lambda G^{\rm T} (\bm{w}\oplus \bm{w}')^{\rm T}.
\end{eqnarray}
Since matrix $G \Lambda G^{\rm T}$ is invertible and since $\bm{w}\oplus \bm{w}' \neq 0$, we have $G\Lambda G^{\rm T} (\bm{w}\oplus \bm{w}')^{\rm T} \neq 0$. 
Since matrix $W$ contains all the possible $2k$-bit sequence, the half elements in the sequence $\bm{g}(\bm{w}) \oplus \bm{g}(\bm{w}')$ are $1$, and the others are $0$.
Thus, the weight of $\bm{g}(\bm{w}) \oplus \bm{g}(\bm{w}')$ is $2^{2k-1}$.

Let $\bm{v} := (1,\ldots, 1)^{\rm T}$ be a real vector of order $2^{2k}$. We define a set of vectors $\bm{h}(\bm{w}) := 2 \bm{g}(\bm{w}) - \bm{v}$ for $\bm{w} \in \{0,1\}^{2k}$, where this calculation is done in real vector space.
Note that this map is equivalent to replace $0$ and $1$ to $1$ and $-1$, respectively.
Since this map from $\bm{g}(\bm{w})$ to $\bm{h}(\bm{w})$ is affine, $\{\bm{g}(\bm{w})\}$ is affinely independent if $\{\bm{h}(\bm{w})\}$ is linearly independent.
The inner product $\bm{h}(\bm{w}) \bm{h}(\bm{w}')^{\rm T}$ for $\bm{w} \neq \bm{w}'$ can be calculated as 
\begin{eqnarray}
\label{eqbinary}
\bm{h}(\bm{w}) \bm{h}(\bm{w}')^{\rm T} &=& \sum_{i} h(\bm{w})_i h(\bm{w}')_i \nonumber \\
&=& 2^{2k} - 2w(\bm{g}(\bm{w}) \oplus \bm{g}(\bm{w}')) \nonumber \\
&=& 0.
\end{eqnarray}
We used the fact that $h(\bm{w})_i h(\bm{w}')_i$ is $1$ if $g(\bm{w})_i = g(\bm{w}')_i$, and $-1$ otherwise.
Thus, the set of vector $\{\bm{h}(\bm{w})\}$ is linearly independent in the real vector space, and the set of vectors $\{\bm{g}(\bm{w})\}$ is affinely independent.
This means that $H_g = W G$ is faithful and decomposable for an arbitrary stabilizer code.

\section*{Appendix B : Additional information for the implementation of the decoders}


We describe the detail of the implementation of our model, training process, and decoders for the reference. 
We chose rectified linear units (ReLU(x)=max(0,x)) and a sigmoid function (S(x)=$1/(1+e^{-x})$) as the activation function for the hidden layer and that for the final output layer, respectively. 
Batch normalization was deployed in all of our models and was found to be effective. 
We also used L2 regularization to avoid over-fitting of the model. 
In the training phase, the Adam optimization method \cite{kingma2014adam} was used. The learning rate was exponentially decreased, and its schedule was optimized by hand.
The network was built with the tensorflow v1.2 platform.

\subsection*{Details about the multilayer perceptron}

We optimized the following parameters of multilayer perceptron using grid-search: number of neurons per layer (\#unit), number of hidden layers (\#layer), size of the batch (\#batch), and coefficients of the L2 regularization ($\beta$). 
The parameters were searched in the range ${\rm \#unit} \in \{d^2,d^3,d^4\}$, $\beta \in \{0,0.01,0.1\} $, $\#batch \in \{100,500\}$, and ${\rm \#layer} \in \{2,3,4\}$. 
Note that in the case of $d=11$, we tuned ${\rm \#unit}$ by hand since we cannot choose ${\rm \#unit}=d^4$ due to the memory limit of GPU. 
We started the training with learning rate $10^{-3}$, and it was decreased to $10^{-5}$ according to a schedule which was optimized by hand.
We optimized these parameters for each construction of the diagnosis matrix, distance, physical error probability, error model, and size of the training data set. 
We chose the configuration which achieves the smallest logical error probability for an independently generated a validation data set of size $10^5$.
Then, the logical error probability is calculated using another $10^6$ test data set.

\subsection*{Details about the Convolutional Neural Network}

Our CNN model consists of three convolutional layers on top of a single fully-connected hidden layer. For each convolutional layer, the channel number was chosen to be $10d$ for the first two layers and $5d$ for the last layer. 
We chose batch size as 100 in the training of the CNN model. 

The network architecture was the same for both bit-flip and depolarizing noise models in the [[$2d^2-2d+1,1,d$]] surface code, and is described in TABLE\,\ref{table1}. 
\begin{table}
	\centering
	\begin{tabular}{ |p{1.7cm}|p{2.5cm}|p{1.8cm}|p{1.7cm}|  }
		\hline
		Distance & Filter size &Channel number& Neuron number \\
		\hline
		5   &[2x2],[3x3],[3x3]&50, 50, 25& 1000\\
		7 &[2x2],[3x3],[4x4]& 70, 70, 35  &3000\\
		9 &[3x3],[4x4],[5x5] & 90, 90, 45&  5000\\
		11 &[4x4],[5x5],[6x6] & 110, 110, 55&  7000\\
		\hline
	\end{tabular}
	\caption{Network architecture of the [[$2d^2-2d+1,1,d$]] surface code.}
	\label{table1}
\end{table}
As for the [[$d^2,1,d$]] code with the bit-flip and depolarizing noise models, we used the network architecture described in TABLE\,\ref{table2}. 
The filter stride was set to 1 in all directions.
\begin{table}
	\centering
	\begin{tabular}{ |p{1.7cm}|p{2.5cm}|p{1.8cm}|p{1.7cm}|  }
		\hline
		Distance & Filter size &Channel number& Neuron number \\
		\hline
		5   &[2x2],[3x3],[3x3]&50, 50, 25& 1000\\
		7 &[2x2],[3x3],[3x4]& 70, 70, 35  &3000\\
		9 &[2x3],[3x4],[4x5]& 90, 90, 45&  5000\\
		11 &[2x4],[3x5],[4x6]& 110, 110, 55&  7000\\
		\hline
	\end{tabular}
	\caption{Network architecture of the [[$d^2,1,d$]] surface code.}
	\label{table2}
\end{table}

\subsection*{Implementation of the minimum-distance decoder}
The minimum-distance decoder of the surface code under the bit-flip noise can be implemented by reducing the problem into the minimum-weight perfect matching. The minimum-weight perfect matching can be efficiently solved with Blossom algorithm \cite{edmonds1965paths}. We used Kolmogorov's implementation of Blossom algorithm \cite{kolmogorov2009blossom}. 
In the other cases, we reduced the problem into the following instance of integer programming. 
\begin{eqnarray}
{\rm Minimize} \: w(\bm{e}) \: {\rm s.t.} \: H_c \Lambda \bm{e}^{\rm T} = s
\end{eqnarray}
This problem was solved with IBM ILOG CPLEX.
We obtained at least $10^5$ samples for each plot. 
In all the cases, the solver reached the optimal solution. 

\subsection*{Time for single prediction, implementation and environment}
We measured the time for single decoding on the [[$2d^2-2d+1,1,d$]] surface code with $d=11$ and $p=0.15$ under the depolarizing noise for the MD decoder, MWPM decoder, and the proposed neural decoders with the MLP and CNN models. Note that the times of the MD decoder and the MWPM decoder depend on the physical error probability. 

We used IBM ILOG CPLEX via python-wrapper for constructing the MD decoder. The program was executed on Intel Xeon E5-2687W v4 with default settings. The MD decoder takes about 330 milliseconds per decoding. Note that the time may be improved by optimizing the settings of CPLEX.

The Kolmogorov's implementation of Blossom algorithm \cite{edmonds1965paths,kolmogorov2009blossom} was used for the MWPM decoder. We compiled the codes with Microsoft Visual C++ 2015 and with O2 option. The program was executed on Intel Core i7-6700 without parallelization. The MWPM decoder took about 56 microseconds per decoding. 

The proposed neural decoders were implemented with python and tensorflow. We measured the time for single prediction when we set batch size as 1, the number of layer as 2, the number of units per layer as 7000 for the MLP model. The configuration of the CNN model is shown in TABLE\,\ref{table1}. The computation was performed using Intel Core i7-6700 and  GeForce GTX 1060 6GB. The proposed neural decoders with the MLP and CNN models took 2.2 milliseconds and 7 milliseconds, respectively, for feed-forwarding the input data and finding the most probable class $\bm{w}$. 
Since the prediction of the neural decoders can be done with simple matrix multiplications, we expect that the time for single prediction of the neural decoder can be made shortened by using an optimized hardware, such as FPGA, for example.

\section*{Appendix C : The specific choices of the uniform data construction}

We have introduced the uniform data construction in Sec.\,III. 
In this appendix, we show specific uniform data construction for the surface and color codes. 

%
We choose $3d$ logical operators for the [[$2d^2-2d+1,1,d$]] surface code by using two patterns as shown in Fig.\,\ref{fig:sc_assign}.
\begin{figure*}[tp]
	\centering
  \begin{minipage}[]{0.49\hsize}
		\includegraphics[clip,width=7.5cm]{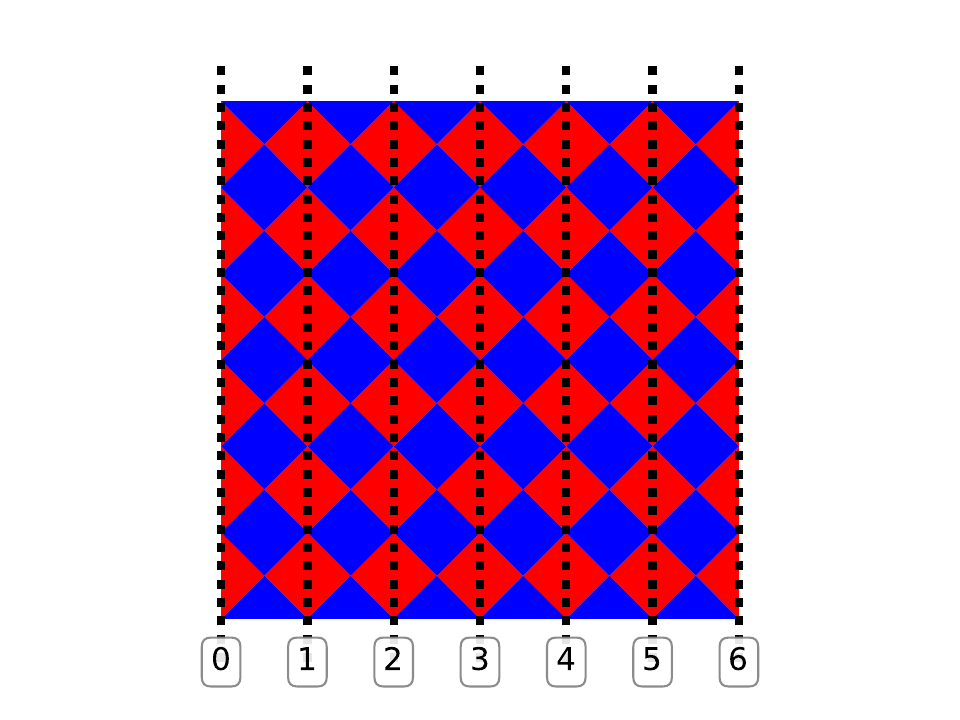}
		\subcaption{Pattern 1}
  \end{minipage}
  \begin{minipage}[]{0.49\hsize}
		\includegraphics[clip,width=7.5cm]{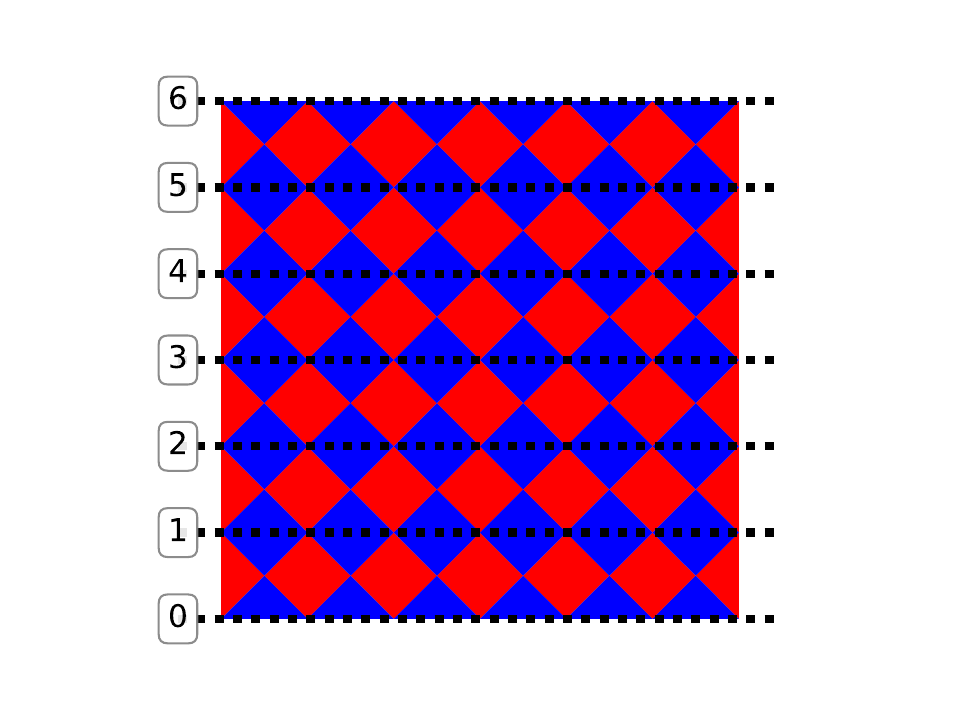}
		\subcaption{Pattern 2}
  \end{minipage}
	\caption{Logical operators used for the construction of a diagnosis matrix for the surface codes [[$2d^2-2d+1,1,d$]]. Each dotted black line corresponds to a chosen logical operator.}
	\label{fig:sc_assign}
\end{figure*} 
For pattern 1, each dotted line corresponds to a logical $X$ operator, which is the product of the Pauli $Z$ operators on the vertices on the line. 
For pattern 2, each dotted line corresponds to a logical $Z$ operator, which is the product of the Pauli $X$ operators on the vertices on the line.
We choose $d$ logical $Y$ operators written as the product of the $i$-th logical $X$ operator and the $i$-th logical $Z$ operator for $i=0,\ldots, d-1$.
We choose $3d$ logical operators for the [[$d^2,1,d$]] surface code with two patterns as shown in Fig.\,\ref{fig:scr_assign}. 
\begin{figure*}[tp]
	\centering
  \begin{minipage}[]{0.49\hsize}
		\includegraphics[clip,width=7.5cm]{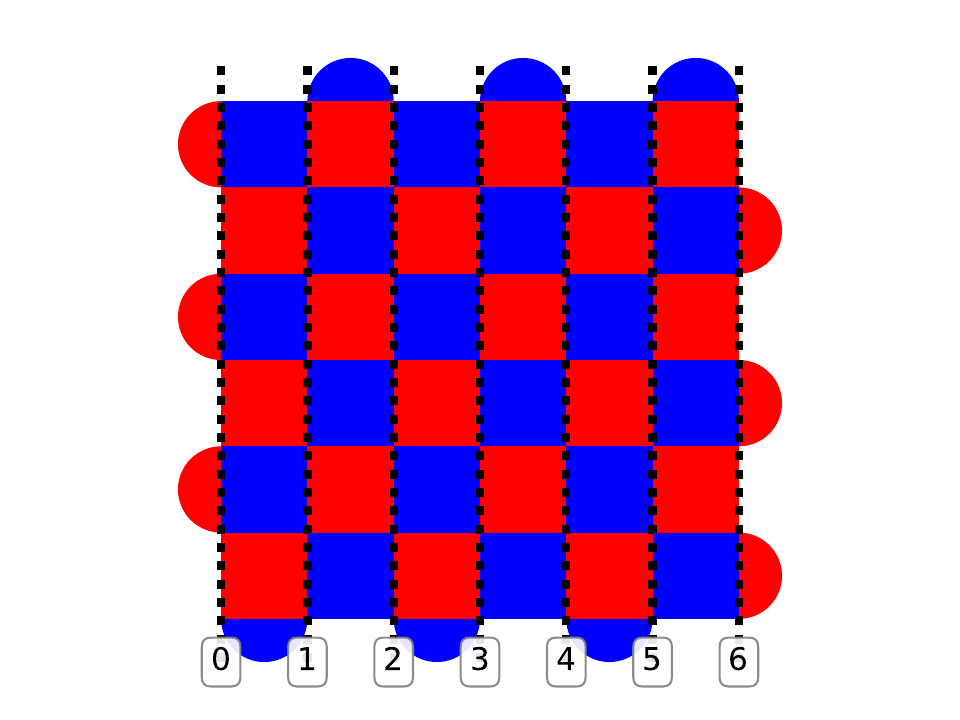}
		\subcaption{Pattern 1}
  \end{minipage}
  \begin{minipage}[]{0.49\hsize}
		\includegraphics[clip,width=7.5cm]{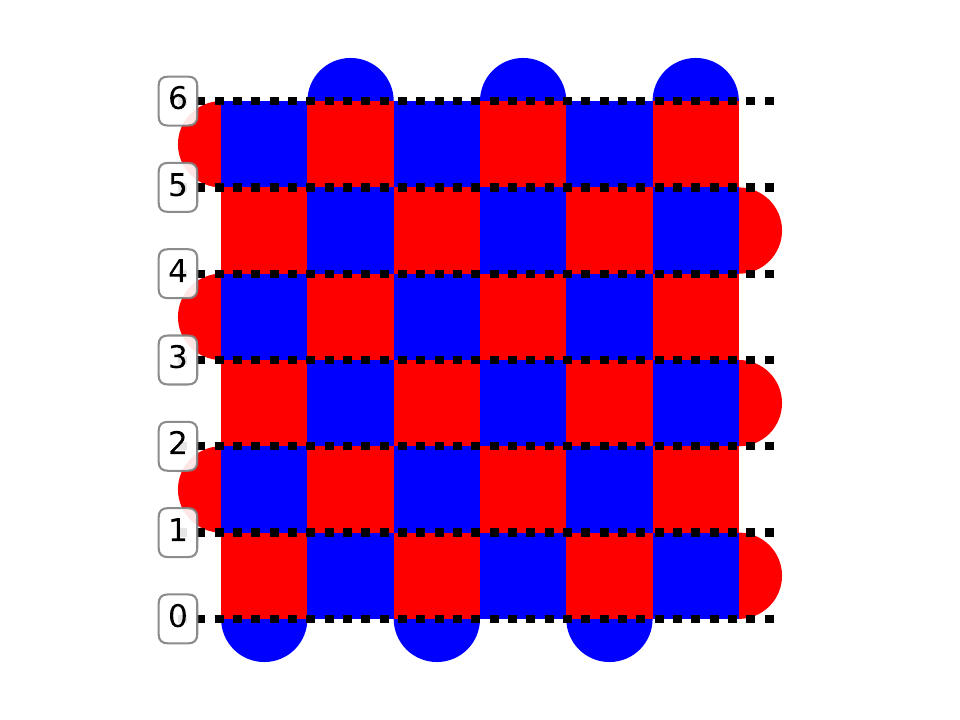}
		\subcaption{Pattern 2}
  \end{minipage}
	\caption{Logical operators used for the construction of a diagnosis matrix for the surface codes [[$d^2,1,d$]]. Each dotted black line corresponds to a chosen logical operator.}
	\label{fig:scr_assign}
\end{figure*}
The rule of choice is the same as that of the [[$d^2,1,d$]] surface code.

%
We choose $\frac{9}{2}(d+1)$ logical operators for the [6,6,6]-color code as shown in Fig.\,\ref{fig:cc666_assign}.
\begin{figure*}[tp]
	\centering
  \begin{minipage}[]{0.32\hsize}
		\includegraphics[clip,width=5cm]{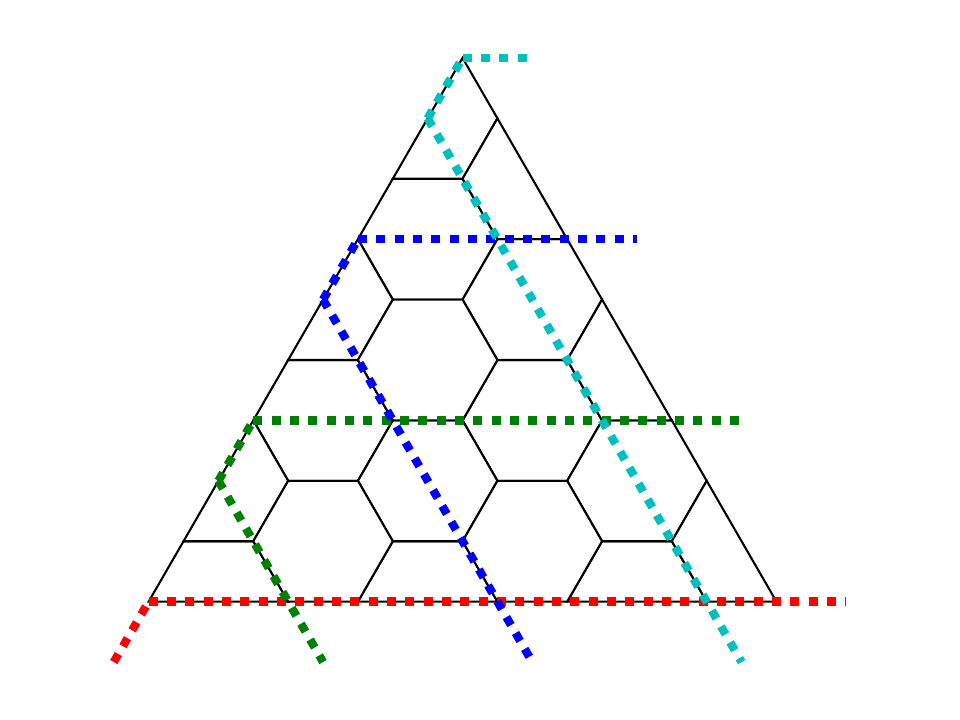}
		\subcaption{Pattern 1}
  \end{minipage}
  \begin{minipage}[]{0.32\hsize}
		\includegraphics[clip,width=5cm]{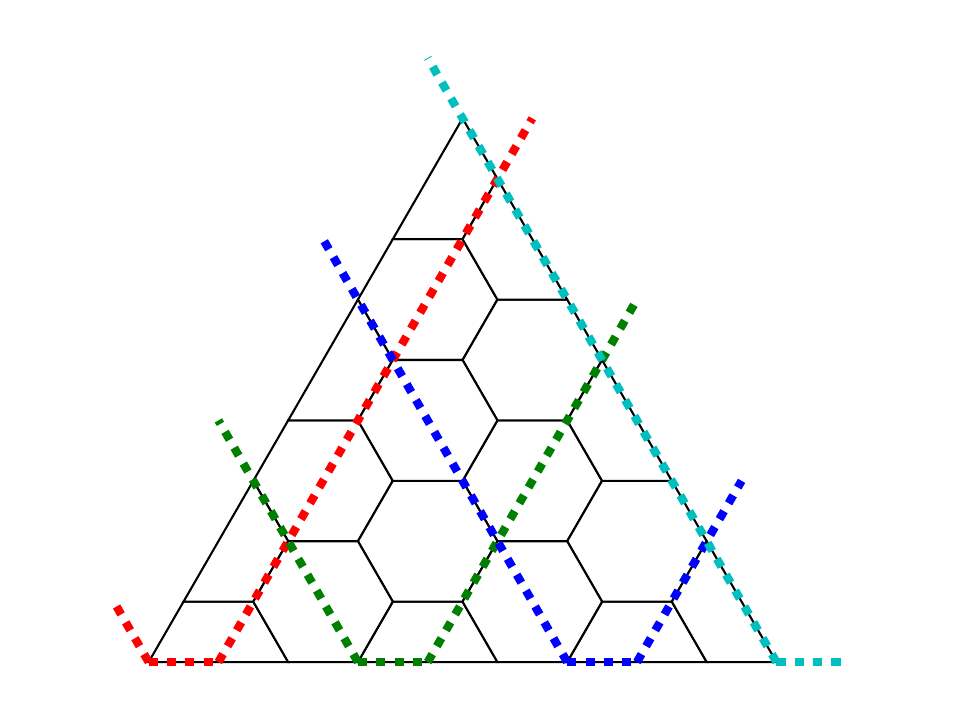}
		\subcaption{Pattern 2}
  \end{minipage}
  \begin{minipage}[]{0.32\hsize}
		\includegraphics[clip,width=5cm]{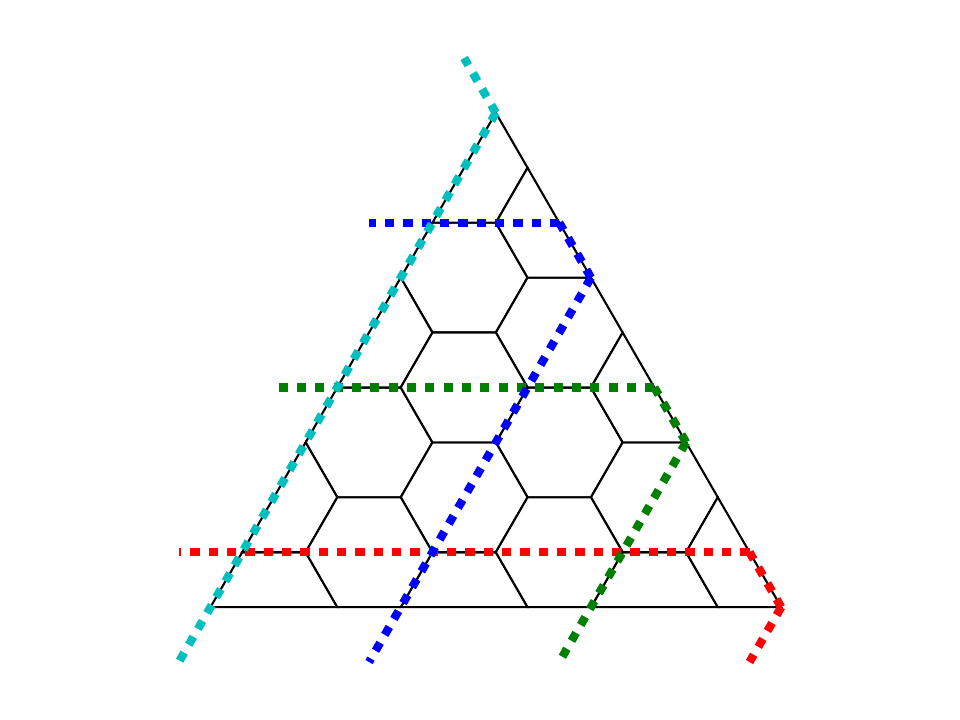}
		\subcaption{Pattern 3}
  \end{minipage}
	\caption{Logical operators used for the construction of a diagnosis matrix for the [6,6,6]-color codes. Each colored line corresponds to chosen logical operators. The lines are colored only for visibility, and are not related to the colors of color codes.}
	\label{fig:cc666_assign}
\end{figure*}
There are $\frac{1}{2}(d+1)$ lines for each pattern.
In all of the three patterns, each line corresponds to the logical $X$-, $Z$-, and $Y$-operators on the physical qubits on the line.
We choose $6(d+1)$ logical operators for the [4,8,8]-color code as shown in Fig.\,\ref{fig:cc488_assign}.
\begin{figure*}[tp]
	\centering
  \begin{minipage}[]{0.49\hsize}
		\includegraphics[clip,width=7.5cm]{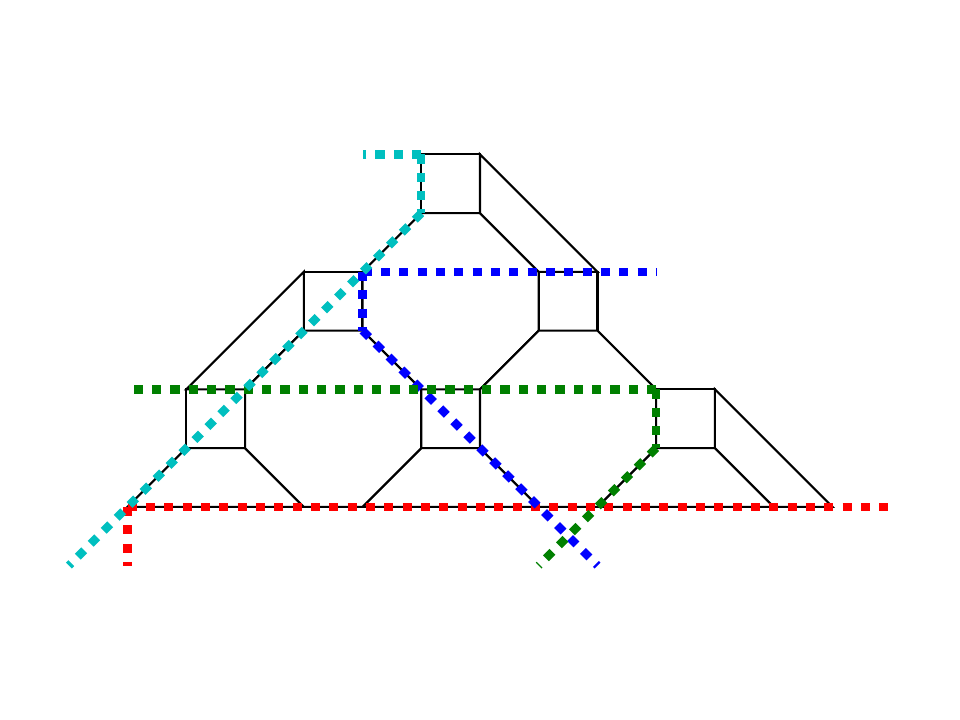}
		\subcaption{Pattern 1}
  \end{minipage}
  \begin{minipage}[]{0.49\hsize}
		\includegraphics[clip,width=7.5cm]{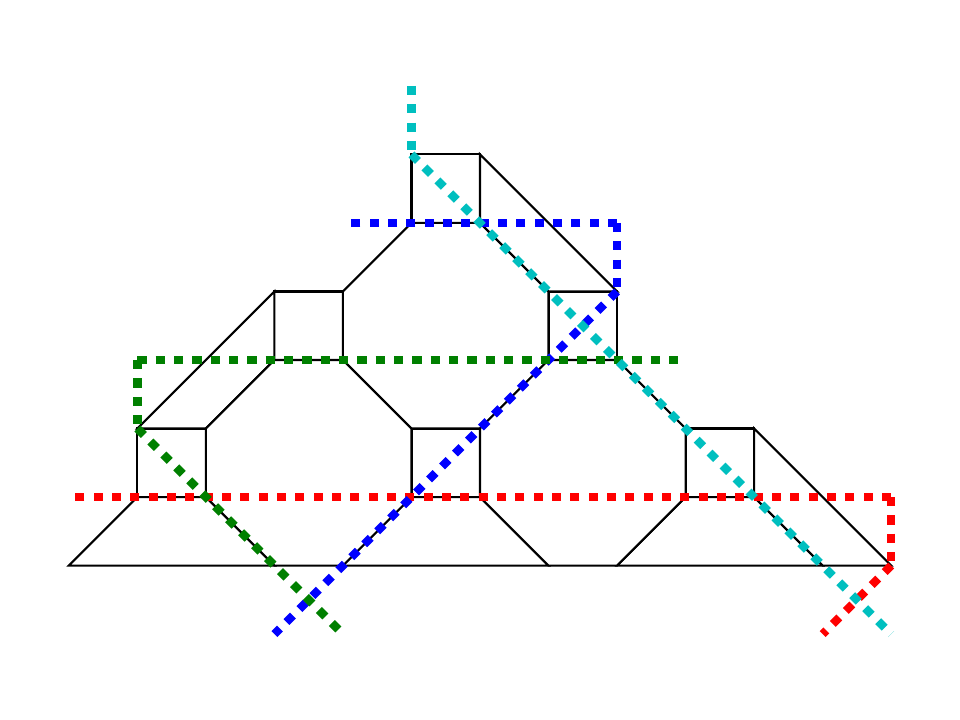}
		\subcaption{Pattern 2}
  \end{minipage}
	\vspace{1mm}
  \begin{minipage}[]{0.49\hsize}
		\includegraphics[clip,width=7.5cm]{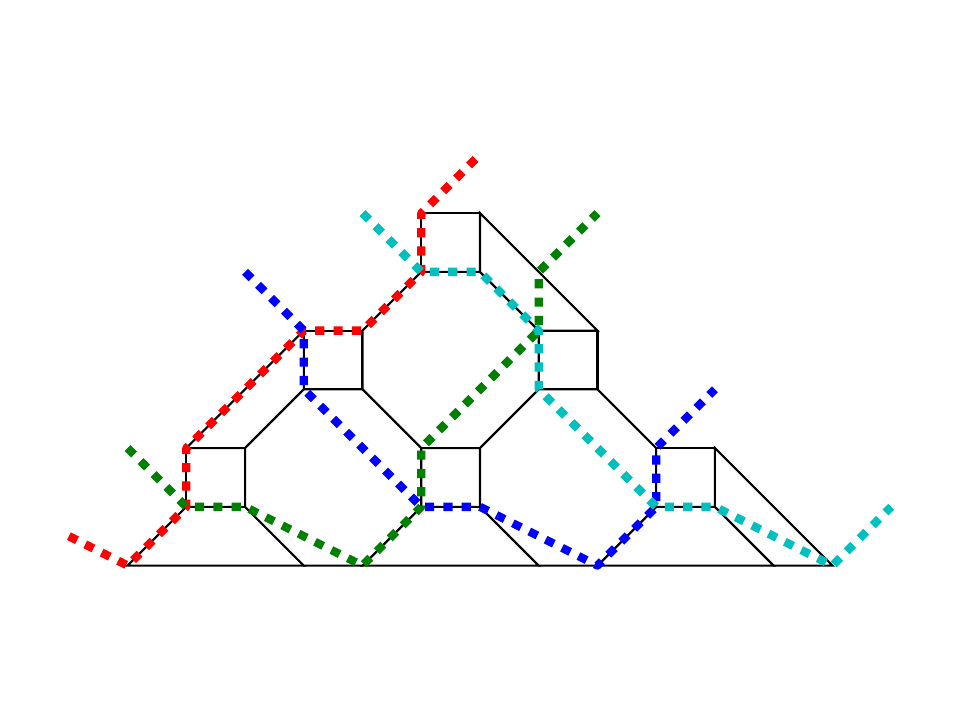}
		\subcaption{Pattern 3}
  \end{minipage}
  \begin{minipage}[]{0.49\hsize}
		\includegraphics[clip,width=7.5cm]{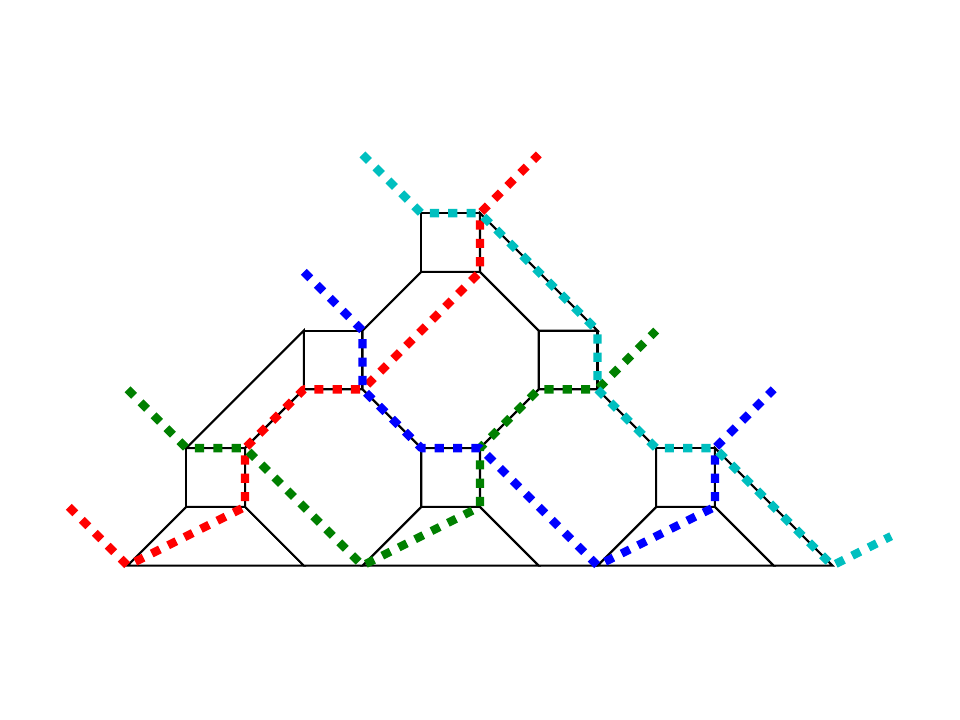}
		\subcaption{Pattern 4}
  \end{minipage}
	\caption{Logical operators used for the construction of a diagnosis matrix for the [4,8,8]-color codes. Each colored line corresponds to chosen logical operators. The lines are colored only for visibility, and are not related to the colors of color codes.}
	\label{fig:cc488_assign}
\end{figure*}
There are $\frac{1}{2}(d+1)$ lines for each pattern.
The choice of the logical operators is the same as that of the [6,6,6]-color codes.

In all the patterned choice of the logical operators, we can verify that the sensitivity is constant, since every physical qubit is measured by at most constant number of logical operators. 
On the other hand, the minimum boundary distance is scaled as $O(d)$, since the same number $O(d)$ of logical $X$-, $Y$-, and $Z$-operators are used.
Thus, the normalized sensitivity is scaled as $O(d^{-1})$ with these choices.

\bibliographystyle{apsrev4-1}
\bibliography{cite1,cite2}

\end{document}